\documentclass[11pt,a4paper,reqno]{amsart}
\usepackage{amsthm,amsfonts,amsmath,amssymb,amsrefs,bookmark,appendix}
\usepackage{amsxtra, mathrsfs}
\usepackage{colordvi}
\usepackage[usenames,dvipsnames]{color}
\usepackage{dsfont}
\usepackage{graphics}
\usepackage{tikz}
\usetikzlibrary{cd,arrows}

\theoremstyle{plain}
\newtheorem{theorem}{Theorem}
\newtheorem{lemma}[theorem]{Lemma}
\newtheorem{corollary}[theorem]{Corollary}
\newtheorem{proposition}[theorem]{Proposition}

\theoremstyle{remark}
\newtheorem{rem}[theorem]{\bf Remark}

\newtheorem{example}{Example}

\numberwithin{equation}{section}

\DeclareMathOperator{\spec}{spec}
\DeclareMathOperator{\supp}{supp}
\DeclareMathOperator{\loc}{loc}
\DeclareMathOperator{\dist}{dist}
\DeclareMathOperator{\dom}{dom}
\DeclareMathOperator{\vol}{vol}
\DeclareMathOperator{\sing}{sing}
\DeclareMathOperator{\ran}{ran}

\DeclareMathOperator{\link}{Link}

\DeclareMathOperator{\Sd}{\mathcal{S}_{disc}}

\DeclareMathOperator{\Sf}{sf}

\DeclareMathOperator{\err}{err}
\DeclareMathOperator{\Int}{int}
\DeclareMathOperator{\ext}{ext}
\DeclareMathOperator{\per}{per}

\DeclareMathOperator{\clos}{clos}
\DeclareMathOperator{\slim}{s--lim}

\renewcommand{\phi}{\varphi}
\newcommand{\eps}{\varepsilon}
\newcommand{\dint}{\displaystyle\int}

\newcommand{\norm}[1]{\lVert #1 \rVert}

\newcommand{\cip}[2]{\langle #1, #2 \rangle}
\newcommand{\wt}[1]{\widetilde{#1}}
\newcommand{\nn}{\nonumber}

\newcommand{\D}{\mathbb{D}}
\newcommand{\R}{\mathbb{R}}
\newcommand{\C}{\mathbb{C}}
\newcommand{\Z}{\mathbb{Z}}
\newcommand{\T}{\mathbb{T}}
\newcommand{\N}{\mathbb{N}}
\renewcommand{\S}{\mathbb{S}}

\newcommand{\cB}{\mathcal{B}}

\newcommand{\cD}{\mathcal{D}}
\newcommand{\cE}{\mathcal{E}}

\newcommand{\cG}{\mathcal{G}}
\newcommand{\cH}{\mathcal{H}}

\newcommand{\cL}{\mathcal{L}}

\newcommand{\cT}{\mathcal{T}}

\newcommand{\rT}{\mathrm{T}}
\renewcommand{\d}{\mathrm{d}}

\newcommand{\sC}{\mathscr{C}}
\newcommand{\sD}{\mathscr{D}}

\newcommand{\sK}{\mathscr{K}}

\newcommand{\sS}{\mathscr{S}}
\newcommand{\sT}{\mathscr{T}}

\newcommand{\ua}{\underline{\alpha}}
\newcommand{\uS}{\underline{S}}

\newcommand{\bn}{\boldsymbol{n}}
\newcommand{\bA}{\boldsymbol{A}}
\newcommand{\bB}{\boldsymbol{B}}

\newcommand{\bN}{\boldsymbol{N}}
\newcommand{\bT}{\boldsymbol{T}}
\newcommand{\bx}{\boldsymbol{x}}
\newcommand{\be}{\boldsymbol{e}}
\newcommand{\bv}{\boldsymbol{v}}
\newcommand{\bp}{\boldsymbol{p}}

\newcommand{\bu}{\boldsymbol{u}}
\newcommand{\bS}{\boldsymbol{S}}

\newcommand{\bsigma}{\boldsymbol{\sigma}}
\newcommand{\Sig}{\Sigma}
\newcommand{\Cg}{C_{\gamma}}
\newcommand{\bal}{\boldsymbol{\alpha}}
\newcommand{\bet}{\boldsymbol{\beta}}
\newcommand{\bs}{\mathbf{s}}

\newcommand{\Tf}{[0,1]_{\per}}

\begin{document}
\title[]{Analysis of zero modes for Dirac operators with magnetic links}
\date{}

\author[F. Portmann]{Fabian Portmann}
\address[F. Portmann]{QMATH, Department of Mathematical Sciences, University of Copenhagen
Universitetsparken 5, 2100 Copenhagen, DENMARK\newline
Current Address: IBM Switzerland, Vulkanstrasse 106 Postfach, 8010 Z\"{u}rich, Switzerland} 
\email{f.portmann@bluewin.ch}

\author[J. Sok]{J\'er\'emy Sok}
\address[J. Sok]{QMATH, Department of Mathematical Sciences, University of Copenhagen
Universitetsparken 5, 2100 Copenhagen, DENMARK\newline
Current Address: DMI, Universit\"{a}t Basel, Spiegelgasse 1, 4051, Basel, Switzerland} 
\email{jeremyvithya.sok@unibas.ch}

\author[J. P. Solovej]{Jan Philip Solovej}
\address[J. P. Solovej]{QMATH, Department of Mathematical Sciences, University of Copenhagen
Universitetsparken 5, 2100 Copenhagen, DENMARK} 
\email{solovej@math.ku.dk}

\thanks{The authors acknowledge support from the ERC grant Nr.\ 321029 ``The mathematics of the structure of matter" and
VILLUM FONDEN through the QMATH Centre of Excellence grant. nr.\ 10059.
This work has been done when all the authors were working at the university of Copenhagen.}

\begin{abstract}
	In this paper we provide a means to approximate Dirac operators with magnetic fields supported on
	links in $\S^3$ (and $\R^3$) by Dirac operators with smooth magnetic fields.
	Then we proceed to prove that under certain assumptions, the spectral flow of paths along these operators
	is the same in both the smooth and the singular case. We recently characterized the spectral flow of such paths
	in the singular case. This allows us to show the existence
	of \emph{new} smooth, compactly supported magnetic fields in $\R^3$
	for which the associated Dirac operator has a non-trivial kernel. Using Clifford analysis, we also
	obtain criteria on the magnetic link for the non-existence of zero modes.
\end{abstract}


\maketitle
\tableofcontents

\section{Introduction}
In two previous papers \cites{dirac_s3_paper1,dirac_s3_paper2} we have introduced 
Dirac operators on $\S^3$ and $\R^3$ with magnetic fields supported on links, 
and we investigated the spectral properties of these operators. 
These rather singular magnetic fields, which we call magnetic links, are described
by an oriented link $\gamma=\cup_{k=1}^K\gamma_k$ together with a collection 
of fluxes $(2\pi\alpha_k)_{k=1}^K$ carried by its connected components.
Each knot $\gamma_k$ can be thought of
as a field line carrying flux $2\pi\alpha_k$. 

When fixing the oriented link $\gamma$, the associated Dirac operators exhibit a 
$2\pi$-periodicity for each flux $2\pi\alpha_k$, and \cite{dirac_s3_paper2} 
was devoted to the study of the spectral flow of loops of 
Dirac operators (on $\S^3$) when tuning the fluxes $2\pi\alpha_k$'s in the 
corresponding parameter space. 
We recall that the spectral flow is the net number of eigenvalues 
(counting multiplicity) crossing the spectral point $0$ from negative to positive along the path. 
For many types of links we were able to show that
the spectral flow is non-zero, which implies that the kernel of
the corresponding Dirac operator is non-trivial, or in other words that the operator has zero modes.

In this paper we show that these singular Dirac operators are limits 
(in an appropriate sense) of Dirac operators with smooth magnetic fields.
The approximation consists in smearing out
the field lines $\gamma_k$'s while keeping the fluxes $0\le 2\pi\alpha_k<1$ fixed. 
In many cases, the spectral flow of paths in the 
singular case equals the spectral flow of the corresponding smooth versions.
Thus we obtain new examples of zero modes in $\S^3$ for Dirac operators
with smooth fields and finite field energy.
In particular, these smooth magnetic fields have compact support in the vicinity of the oriented link.
Note that by conformal invariance of the kernel, the same result holds in $\R^3$ as well.

\subsection{Overview}
The importance of zero modes for Dirac operators (or Pauli operators) became apparent 
in the study of stability of (non-relativistic) matter in the presence of 
magnetic fields \cite{FroLiebLoss86,LiebLoss86}.
In \cite{FroLiebLoss86}, the authors considered the Hamiltonians
$$
	H_{\bA,Z}:=\big(\sigma\cdot(-i\nabla+\bA)\big)^2-\frac{Z}{|\bx|}+\eps\int |\bB|^2
$$ 
on $L^2(\R^3)^2$. They showed that $\inf_{\bA\in \dot{H}^1}\inf \spec(H_{\bA,Z})=-\infty$ 
if the value of the nuclear charge $Z$ is above a critical $Z_c$. The value of $Z_c$ is given by
$$
	Z_c:=\inf \Big\{\big(\cip{\psi}{\tfrac{1}{|\,\cdot\,|}\psi}_{L^2}\big)^{-1}\eps\int |\bB|^2\Big\},
$$ 
where the infimum runs over solutions to the zero-mode problem
\begin{equation}\label{eq:zero_mode}
	\left\{
	\begin{array}{l}
		(\bA,\psi)\in \dot{H}^1(\R^3,\R^3)\times H^1(\R^3,\C^2),\\
		\sigma\cdot(-i\nabla+\bA)\psi=0,\\
		\int|\psi|^2=1.\\
	\end{array}
	\right.
\end{equation}
Explicit smooth solutions to \eqref{eq:zero_mode} were then 
discovered in \cite{LossYau86}, giving a finite value for $Z_c$.
The problem of stability of matter however
remained open and was only solved much later in 
\cite{Fefferman95,LiebLossSol95,Fefferman96}.

Finding a solution to \eqref{eq:zero_mode} on $\R^3$, or more generally on a 3-manifold, 
is a rather difficult problem. 
As mentioned before, the first solutions were exhibited in \cite{LossYau86}.
These were extended in 
\cite{Adametal99, Adametal00_1, Adametal00_2},
and further examples were exhibited in \cite{Elton00}.
The geometry behind the construction in \cite{LossYau86}
was exploited in \cite{ErdSol01} and in \cite{Adametal00_3}; through Hopf maps one can 
pull back smooth magnetic fields from $\S^2$ to $\S^3$. Then it was shown that the kernel of the Dirac operator 
can have an arbitrary dimension. In \cite{ErdSol01} the authors obtained a full characterization 
of the spectrum of Dirac operators on $\S^3$ constructed from the usual Hopf map.
Combining this result with the invariance of the kernel of the Dirac operator
under the stereographic projection (or in fact general conformal maps) 
gives an explicit description of the kernel of the Dirac operator on
$\R^3$. Some of the zero modes obtained though Hopf maps were related
to projections of non-Abelian fields \cite{JackPi00} and to projections 
to $\S^3$ of $4$D-Dirac operators \cite{DunneMin,Min}, with a straightforward extension to higher odd dimensions.
The degeneracy of these zero modes were related to Hopf indices and magnetic helicity:
we refer the interested reader to \cite{Adametal00_3,JackPi00,DunneMin,Min} and the references therein. 
Using other fibrations of $\S^3$, new zero modes were obtained in \cite{Adametal05}. 
As we can see in explicit formulas \cite[Section 3]{LossYau86} and \cite{JackPi00,AleyTolka},
zero modes are intimately related to such fibrations: a real and divergence-free vector field $U$
with unit $L^1(\R^3)^3$-norm gives rise, at least on a formal level, to a solution $(\psi,\bA)$ to the zero mode equation,
$U$ coinciding with $(\cip{\psi}{\sigma_j \psi}_{\C^2})_{1\le j\le 3}$. Note however that the corresponding magnetic field
has a complicated formula, and it is hard to read from it the key features which produce the zero mode.

In \cite{MR1807254}, the authors adopt another strategy and describe the complement of the set
of magnetic potentials giving rise to zero modes. In particular, they show that the set of $\bB\in L^{3/2}(\R^3)^2$
such that the kernel of the Dirac operator is trivial contains an open dense subset of $L^{3/2}(\R^3)^3$ -- 
this indicates that zero modes are ``sparse". 
In \cite{Elton02}, the set of magnetic potentials (in the appropriate space)
whose Dirac operator has a kernel of dimension $m$ is studied, and for $m\in\{1,2\}$ it is shown to be a 
manifold of co-dimension $m^2$.
In \cite{BengVdB15}, the authors give a criterion for the existence of zero modes for Pauli
operators with fastly decaying fields, and in \cite{MR1874252} it is shown that the dimension of the kernel of the Dirac operator
is bounded by the $L^{3/2}(\R^3)^3$-norm of the magnetic field.

In two dimensions, the situation is somewhat different due to the theorem of Aharonov and Casher 
\cite{AhaCash79}, which gives an explicit description of the kernel of the 
Dirac operator in $\R^2$ in the smooth case.
Then this result has been extended to the sphere $\S^2$ 
(see \cite{Schroe_op_Cycon_et_al}*{Section~6.4} and references therein).
The main tool in dimension two is complex analysis; the kernel is 
described in terms of holomorphic or anti-holomorphic functions
and its dimension is related to the integral part of the total flux. 

After the discovery of the Aharonov-Bohm effect \cite{AhaBohm59},
point-like magnetic fields (or more generally magnetic fields which are Borel measures)
have become widely studied from an operator point of view.
In this singular case, there are many extensions for both the Dirac operator and the Pauli operator, 
and the question arises to decide which one
is the physical correct one. Among other features, the system exhibits a 
$2\pi$-periodicity for the flux carried by
a Dirac point-magnetic field (also called an Aharonov-Bohm (A-B) solenoid).
For the Pauli operator in $\R^2$, two choices are proposed in \cite{GeylGrish02, ErdVug02},
exploiting the idea from \cite{AhaCash79} of factorizing the terms in the expression of the energy. 
These two are compared in \cite{Persson_on_Pauli}. 
The extensions of the Pauli operator for two A-B solenoids in $\R^2$ are explicitly 
computed in \cite{GeylStov04}.
In \cite{Persson_dirac_2d}, the author studies the extension of the Dirac operator
in the case of a smooth magnetic field 
together with a finite sum of A-B solenoids, and for each extension investigates
the dimension of the kernel.

From a physical point of view, it is reasonable to think of an A-B solenoid as the limit of 
spread-out magnetic fields, whose supports
collapse to a point.
This leads us to study the operators which are limits of the corresponding 
family of elliptic operators. For a single A-B solenoid,
the Pauli operator is studied in \cite{BorgPule03}.
The Dirac operator is studied in \cite{Tamura}, and it turns out that only one extension 
arises as a limit in the norm-resolvent sense of Dirac operators
with (smooth) spread out magnetic field. This last result will be crucial for the present paper.

We end this overview by recalling some results on the zero mode space in even dimension. 
In \cite{RozenShiro06} the authors show that in dimension two, under some assumptions 
(including the flux being infinite), the kernel of the Pauli operator is infinite dimensional.
In higher even dimensions, some results have been obtained, we refer the reader 
to \cite{shigekawa91,Persson_even_dim}.

\subsection{Main results}
In this paper we study the Dirac operator with a magnetic field supported on 
links introduced in \cite{dirac_s3_paper1}.
More precisely, given a smooth, oriented link $\gamma=\cup_{k=1}^K \gamma_k \subset \S^3$ 
and a set of fluxes $(2\pi\alpha_k)_{k=1}^{K}\in (0,1)^K$,
the one-current $\bB:=\sum_{k=1}^K2\pi\alpha_k[\gamma_k]$ 
can be seen as a magnetic field on $\S^3$. 
We choose as magnetic potential the two-current
$$
	\bA:=\sum_{k=1}^K 2\pi\alpha_k [S_k],
$$
where $S_k$ is a Seifert surface for $\gamma_k$ 
such that $S_k$ and $\gamma_{k'}$ are transverse in $\S^3$ for $k\neq k'$
(a Seifert surface for a link is a compact
oriented surface whose boundary is the given link).
The transversality property is generic and as shown in \cite{dirac_s3_paper1}
can always be assumed. We recall that a $n$-current is a continuous
linear form on the set of smooth $n$-forms with compact support \cite{Rham55}*{Chapter~3}. Taking such a one-form $\omega^1$
and two-form $\omega^2$, $[\gamma_k]$ and $[S_k]$, seen as singular vector and $2$-vector fields respectively,
act as follows:
\[
([\gamma];\omega^1):=\int_{\gamma_k}\omega^1\quad\&\quad([S];\omega^2):=\int_{S_k}\omega^2.
\]
Here $\bA$ is a magnetic potential
of $\bB$ in the sense that $\partial \bA=\bB$, which is the corresponding condition
to $\d \alpha=\beta$ for a couple $(\alpha,\beta)$ of smooth one- and two-forms. This follows from Stokes' formula: 
\[(\partial [S_k];\omega^1):=([S_k];\d \omega^1)=([\gamma_k];\omega^1).\]

Now we quickly sketch the construction of the main
operator, we refer the reader to \cite{dirac_s3_paper1} for the details.
Recall that the spin structures on $\S^3$ are trivial
and that the Dirac operator acts on 
$L^2$-sections of the trivial bundle $\Psi_0:=\S^3\times\C^2$. 
The spin structure endows it with an isometry $\bsigma$, the Clifford map, from the cotangent 
bundle $\rT^*\S^3$ to the bundle of the traceless Hermitian elements
of $\mathrm{End}(\Psi_0)$. The effect of the formal term $\bsigma(\bA)$ 
in the Dirac operator is only seen in the domain of the operator and imposes a phase jump 
$e^{-2i\pi\alpha_k}$ across $S_k$, which leads to the 
$2\pi$-periodicity mentioned before.
We start with the minimal operator $\cD_{\bA}^{(\min)}$,
which acts like the free Dirac operator on 
$\Omega_{\uS}:=\S^3\setminus(\cup_k S_k)$
and has domain
\begin{multline*}
	H_{\bA}^1(\S^3)^2:=\big\{ \psi\in H^1(\Omega_{\uS})^2,\\
	\ \psi_{|_{(S_k)_+}}=e^{-2i\pi\alpha_k}\psi_{|_{(S_k)_-}}\in 
	H_{\loc}^{1/2}(S_k\setminus (\cup_{k'\neq k}S_{k'}\cup \gamma_k))^2\big\}.
\end{multline*}
Its self-adjoint extensions are investigated in \cite{dirac_s3_paper1};
the domain of each is characterized by the behavior of the 
wave functions of its domain in the vicinity of the knots $\gamma_k$.
We are interested in one particular extension, denoted by $\cD_{\bA}^{(-)}$, 
characterized by the property that the singular part of the wave function
``aligns against the magnetic field" in the following sense.
Given the Seifert frame $(\bT,\bS,\bN)$ of a link 
(that is the Darboux frame defined on the Seifert surface \cite{Spivakvol4}*{Part 7.E}), 
there exist two smooth sections $\xi_{\pm}$, defined (up to a common phase) 
in the neighborhood of each knot $\gamma_k$, 
characterized by
$$
	\left\{
	\begin{array}{ccl}
		\bsigma(\bT^{\flat})\xi_{\pm} &=& \pm\xi_{\pm},\\
		\omega(\bS)+i\omega(\bN)&=&\cip{\xi_-}{\bsigma(\omega)\xi_+},
		\ \forall\,\omega\in\Omega^1(B_\eps[\gamma_k]).
	\end{array}
	\right.
$$
We recall that $\flat$ and $\sharp$ denote the musical isomorphisms transforming 
vectors into 1-forms and 1-forms into vectors respectively. Here, the Seifert frame has been extended
by parallel transportation along the geodesics defined by the distance $\rho_{\gamma_k}$ to $\gamma_k$
in $B_\eps[\gamma_k]$. The real $\eps>0$ is chosen small enough such that we have: 
$B_\eps[\gamma]=\cup_{1\le k\le K} B_\eps[\gamma_k].$

The operator $\cD_{\bA}^{(-)}$ acts like the free Dirac operator on
$\Omega_{\uS}$ and has domain
\begin{multline}\label{def_domain_D}
	\dom\big(\cD_{\bA}^{(-)}\big) 
	:=\big\{\psi \in \dom\big( (\cD_{\bA}^{(\min)})^{*}\big):\\
	\quad \langle \xi_{+},\chi_{\delta,\gamma_k}\psi \rangle \xi_{+} 
	\in \dom\big( \cD_{\bA}^{(\min)}\big), 1 \leq k \leq K\big\},
\end{multline}
where $\chi_{\delta,\gamma_k}\in\sD(\S^3,[0,1])$ is a smooth function with support in $B_\delta[\gamma_k]$
which equals $1$ on $B_{\delta/2}[\gamma_k]$, where $\delta>0$ is a small parameter. We emphasize that the definition does not depend on the choice of
$\chi_{\delta,\gamma_k}$.

The first important result of this paper is Theorem~\ref{thm:conv_dirac}, 
where we prove that the operator $\cD_{\bA}^{(-)}$ is the limit in the norm-resolvent sense 
of a one-parameter family of Dirac operators with smooth magnetic fields.
This establishes the appropriateness of our definition of the singular Dirac operators $\cD_{\bA}^{(-)}$.
This smooth approximation is the adaptation to the three-dimensional situation 
of the smearing of Dirac points considered in \cite{BorgPule03,Tamura}.
By introducing appropriate coordinates in the tubular neighborhood of each
knot, we obtain a trivial Seifert fibration (see Figure~\ref{Seifert_fibration}).
These fibers are chosen to be the field lines of the smooth approximations,
and the smooth gauge is defined through their Seifert surfaces.

Then we use results on the spectral flow from \cite{dirac_s3_paper2} to compute the 
spectral flow of continuous paths with values in the
set of the smooth approximations. The continuity refers to the gap topology 
or topologies introduced in \cite{Wahl08} 
strong enough to ensure the homotopy invariance of the spectral flow. 
Consider the singular gauges $\bA(\ua)$ where we emphasize the dependence
in the flux $\ua\in (0,1)^K$ and let us write $\big(\bA_{\delta}(\ua)\big)_{\delta>0}$ 
the corresponding one-parameter family of smooth magnetic potential converging to $\bA$ (say as currents).
Theorem~\ref{thm:main_zero_existence} states that for a path of fluxes $\ua(t):[0,1]\to [0,1)^K$,
and provided that the endpoints $\cD_{\bA(\ua(0))}^{(-)},\cD_{\bA(\ua(1))}^{(-)}$ are invertible, 
the following equality of spectral flows holds for $\delta>0$ small enough:
$$
	\Sf\big[(\cD_{\bA_{\delta}(\ua(t))})_{0\le t\le 1} \big]
	=\Sf\big[(\cD_{\bA(\ua(t))}^{(-)})_{0\le t\le 1} \big].
$$
Then in Section~\ref{sec:new_results} we give examples when the above spectral flow is non-zero. 
This implies that for $\delta>0$ small enough there exists $t_1(\delta)\in (0,1)$ 
such that $\cD_{\bA_{\delta}(t_1)}$ has a non-trivial kernel. By conformal invariance of the kernel and elliptic regularity
(up to a stereographic projection)
this gives new examples of zero modes in $\R^3$
for certain smooth and compactly supported magnetic fields (see Corollary~\ref{coro:zero_mode}).

We conclude our results on zero modes with Section~\ref{sec:crit_inexistence}, 
where we use Clifford analysis to obtain some criteria for the non-existence of zero modes
for the singular Dirac operator $\cD_{\bA}^{(-)}$, in the case of small fluxes: $0\le \alpha_k<2^{-1}$. 
Theorem~\ref{thm:non_existence} gives a condition on the fluxes $2\pi\alpha_k$ and on the (assumed disjoint) Seifert surfaces 
in terms of the Cauchy operator of a bigger (oriented) surface $\Sig$ containing all of them. 
Some more quantitative statements are given in Corollaries~\ref{coro:inex1}~$\&$~\ref{coro:inex2}.

More precisely: the Cauchy operator associated with $\Sig$ is
the bounded operator $\sC\in\mathcal{B}(L^2(\Sig)^2)$, involved in the Plemelj-Sokhotski jump formula for harmonic spinors
on each component of $\R^3\setminus \Sig$. It is self-adjoint if and only if $\Sig$ is a plane or a sphere. 
The condition of Theorem~\ref{thm:non_existence} is
that $\norm{\sC-\sC^*}_{\mathcal{B}(L^2(\Sig)^2)}$ is bounded from above by a certain function of the fluxes, 
which depends on the relative orientation of the $S_k$'s in $\Sig$.
Then Corollary~\ref{coro:inex1} deals with the case where $\Sig$ is a plane or a sphere. Corollary~\ref{coro:inex2} 
studies the case of a magnetic knot and gives a bound on the flux below which the kernel of the Dirac operator is trivial.
In particular the kernel of the Dirac operator is trivial with a magnetic knot embedded in $\Sig$
\begin{itemize}
	\item if $\Sig$ a sphere or a plane and the flux $2\pi\alpha$ is in $(0,\pi)$,
	\item or in general if the flux $2\pi\alpha$ is in $(0,\pi)$ and satisfies:
	\[
	\norm{\sC-\sC^*}_{\mathcal{B}(L^2(\Sig)^2)}<\mathrm{cot}(\pi\alpha).
	\]
\end{itemize}

\subsubsection*{\bf{Convention:}}
For the remainder of this paper, a link $\gamma \subset \S^3$ 
will always mean a \emph{smooth}, \emph{oriented} submanifold of $\S^3$ 
which is diffeomorphic to the \emph{disjoint} union of finitely many copies of $\S^1$. 
We will often write $\gamma = \bigcup_k \gamma_k$, where the $\gamma_k$ are 
\emph{oriented} knots (the connected components of $\gamma$).

The stereographic projections (say with respect to two antipodal points)
define an atlas of $\S^3$. The canonical volume forms of the two charts $\R^3$
give the orientation on $\S^3$ (by pullback through the stereographic projection).

Furthermore, when the symbol $\nabla$ is used, we mean the Levi-Civita connection if it acts on a vector field,
and we mean the trivial (induced) connection on the (trivial) $Spin^c$ bundle on $\S^3$ if it acts on a spinor. 
Superscripts $\R^2,\R^3$ are added to denote the flat connections in the 
corresponding Euclidean spaces.

\section{Approximation by smooth magnetic fields}

\subsection{Approximate Cartesian coordinates}
Let $\gamma\subset \S^3$ be a knot with Seifert surface $S$. We assume that $S\subset \Sig$,
where $\Sig\subset \S^3$ is a closed, oriented surface. On $\gamma$, the triplet $(\bT(s),\bS(s),\bN(s))$ denotes
the (positively oriented) Darboux frame of $\gamma$ relative to $S$, where $\gamma$ is identified with its arclength-parametrization 
$s\in\R/\ell\Z \mapsto \gamma(s)\in\S^3$. This frame is henceforth referred to as the Seifert frame.
Our convention is that $\bS(s)$ points toward the interior of $S$.

\subsubsection{Differential structures related to $\Sig$}
In general, when an object is carrying a sub- or superscript $\Sig$, we mean 
that the corresponding object is defined through the differential structure of $\Sig$. 
Namely, $\exp^{\Sig}$ denotes the exponential map relative to $\Sigma$ and
$\nabla^{\Sig}$ stands for the Levi-Civita connection on $\rT\Sigma$. 
Furthermore, $\bN_{\Sig}$ denotes the oriented unit normal to $\Sig$ in $\S^3$,
which is extended to a tubular neighborhood of $\Sig$ through parallel transport along the 
$\S^3$-geodesics $\exp_{\bp}(x\bN_{\Sig}(\bp))$ for $\bp\in\Sig$.
In a $\Sig$-tubular neighborhood of $\gamma$, $\bT_{\Sig}$ and $\bS_{\Sig}$
denote the $\nabla^{\Sig}$-parallel transport of $\bT(s)$ and $\bS(s)$ respectively 
along the $\Sig$-geodesics 
\begin{equation}\label{def:c(s,u)}
	x\mapsto c(s,x):=\exp_{\gamma(s)}^{\Sig}(x \bS(s)).
\end{equation}

\begin{rem}
	The knot $\gamma:=\partial S$ is identified with its arc length parametrization.
	For $u_1\in(-\eps,\eps)$ (with $\eps>0$ small enough), the vectors 
	$\bT_{\Sig}(c(s,u_1))$ and $\bS_{\Sig}(c(s,u_1))$ 
	are the parallel transports of themselves 
	along the $\S^3$-geodesics $\big(\exp_{c(s,u_1)}(u_2\bN_{\Sig}(c(s,u_1)))\big)_{u_2}$. 
	Then $(\bT_{\Sig},\bS_{\Sig},\bN_{\Sig})$ 
	smoothly extends the Seifert frame $(\bT(s),\bS(s),\bN(s))$ to
	a tubular neighborhood of $\gamma$. This extension is different from the one considered
	in \cite{dirac_s3_paper1,dirac_s3_paper2} (explained in the introduction), but they agree up to a 
	$\underset{\rho_\gamma\to 0}{\mathcal{O}}(\rho_\gamma)$ where 
	$\rho_\gamma$ is the geodesic distance to $\gamma$. The latter frame was extented along $\S^3$-geodesics. 

	Throughout the paper $(\eta_+,\eta_-)$ denotes two smooth sections on $B_\eps[\gamma]\times\C^2$ 
	of unit length associated with the frame $(\bT_{\Sig},\bS_{\Sig},\bN_{\Sig})$ 
	and defined up to an overall phase by
	\begin{equation}\label{eq:def_sections}
		\left\{
		\begin{array}{l}
			\eta_{\pm}\in\ker(\sigma(\bT_{\Sig}^{\flat})\mp 1),\\
			\omega(\bS_{\Sig}^{\flat})+i\omega(\bN_{\Sig}^{\flat})=\cip{\eta_-}{\sigma(\omega)\eta_+}, 
			\quad \forall\,\omega\in\Omega^{1}(B_\eps[\gamma]).
		\end{array}
		\right.
	\end{equation}
	We emphasize that these sections are different from the sections $(\xi_+,\xi_-)$ 
	considered in \cite{dirac_s3_paper1,dirac_s3_paper2}, but
	on $\gamma$, the sections $ \xi_{\pm}(\gamma(s))$ and $ \eta_{\pm}(\gamma(s))$
	coincide (up to a common phase $e^{i\phi(s)}$). 

	The estimate on the discrepancy between
	$(\bT_{\Sig},\bS_{\Sig},\bN_{\Sig})$ and the frame $(\bT,\bS,\bN)$ defining $\xi_{\pm}$
	implies that the operator $\wt{\cD}_{\bA}^{(-)}$, defined like $\cD_{\bA}^{(-)}$ with $(\eta_+,\eta_-)$
	replacing $(\xi_+,\xi_-)$ in the definition \eqref{def_domain_D}, coincides with $\cD_{\bA}^{(-)}$.
\end{rem}

\subsubsection{Coordinates}\label{sec:def_coord}
In a tubular neighborhood $B_\eps[\Sig]$ of the surface $\Sig$, there are two intrinsic functions:
the projection $\pi_{\Sig}(\bp)$ onto $\Sig$ and the distance from $\Sig$. Through the exponential map
we define a function $u_2:B_\eps[\Sig]\to (-\eps,\eps)$ by
\[
	\bp=\exp_{\pi_{\Sig}(\bp)}\big(u_2(\bp)\bN_{\Sig}(\pi_{\Sig}(\bp)) \big), \quad \forall\bp\in B_\eps[\Sig].
\]
Of course $|u_2(\bp)|$ is nothing else than the distance from $\Sig$, but we can specify from which side of the surface the point $\bp$ lies. 
Thus, in any open set $U\subset \Sig$, giving coordinates in $U$ 
(for instance the projection in $\C^2$ onto $\rT_{\bp}\Sig$ for some $\bp\in\Sig$),
we obtain a set of coordinates on the foliation $(\exp(u_2 \bN_{\Sig})\cdot U\big)_{u_2}$.
In a tubular neighborhood of the knot $B_\eps[\gamma]$ we choose the coordinates $(s,u_1)$
defined through the chart of $\Sig$:
\[
(s,u_1)\mapsto \exp_{\gamma(s)}^{\Sig}\big(u_1\bS(s)\big).
\]
The coordinate $s$ corresponds to the parameter of the $\Sig$-projection onto $\gamma$ and $|u_1|$ is the
$\Sig$-distance to $\gamma$.

For $\eps$ small enough, we obtain a chart $F_\gamma$ defined by
\begin{equation}\label{eq:cart_coord}
	\begin{array}{ccl}
		(\R/\ell\Z)\times \mathbb{D}_\eps &\overset{F_{\gamma}}{\longrightarrow}& \S^3=\{\mathbf{z}\in\C^2,\,|\mathbf{z}|=1\},\\
		(s,\bu)&\mapsto &
		\cos(u_2)\exp_{\gamma(s)}^{\Sig}\big(u_1\bS(s) \big)\\
		&&+\sin(u_2)\bN_{\Sig}\big( \exp_{\gamma(s)}^{\Sig}\big(u_1\bS(s) \big)\big),
	\end{array}
\end{equation}
where $\mathbb{D}_\eps$ denotes the disk in $\R^2$ of center $0$ and radius $\eps$.
The calculation is for vectors in $\C^2$ and we see $u_2$ as an angle for vectors in $\C^2$.
For $w_2$ with $|w_2|$ small enough, we write $\Sig_{w_2}$ for the ``suspended" 
surface $\exp\big( w_2\bN_{\Sig}\big)\cdot \Sig$, that is the surface $\{ \bp,\ u_2(\bp)=w_2\}$. 
Remark that the family $(\Sig_{w_2})_{w_2\in(-\eps,\eps)}$ defines a foliation of $B_{\eps}[\Sig]$.
The expression of the free Dirac operator in terms of the $(s,\bu)$-coordinates is 
given in Proposition~\ref{prop:expr_free_dirac} below.

\subsection{Smoothened magnetic fields \& potentials}
Let $\alpha$ be in $(0,1)$ and $\gamma\subset \S^3$ be a knot with 
Seifert surface $S$. We assume that $S$
is embedded in a closed oriented surface $\Sig\subset \S^3$;
this is possible since we can assume that the Seifert surface $S$ is bicollared \cite{Rolfsen}*{Chapter 5}.

Now we define smooth approximations for a magnetic knot 
$\bB=2\pi\alpha [\gamma]$ with
singular magnetic potential $\bA=2\pi\alpha [S]$.
To do so, we introduce knots which are in some sense parallel to $\gamma$, 
together with associated Seifert surfaces close to $S$. The knots will be the 
\emph{field lines} for the approximating field,
and will be weighted by a smooth density to obtain a 
smooth magnetic field. 
The convergence to the singular magnetic knot is
obtained by concentrating the density around the knot $\gamma$.
From a geometric point of view, these approximations corresponds 
to setting a probability density on a (trivial)
\emph{Seifert fibration} of a tubular neighborhood of $\gamma$.

We recall some notations introduced \cite{dirac_s3_paper1,dirac_s3_paper2}; 
the set of knots and set of Seifert surfaces of knots
are denoted by $\sK$ and $\sS_{\sK}$ respectively. Furthermore, 
for $K\in\N$, $\sS_{\sK}^{(K)}$ denotes the set
\[
	\sS_{\sK}^{(K)}
	:=\big\{\uS\in \prod_{k=1}^K \sS_{\sK},\ S_k\mathrm{\ and\ }\partial S_{k'}\mathrm{\ are\ transverse\ in\ }\S^3,\,k\neq k'\big\}.
\]

Let $0<\delta<\eps/2$. For and $u_1\in (-\delta,\delta)$,  we define $S_{(u_1,0)}\subset \Sig$ as
\begin{equation*}
	\left\{
		\begin{array}{lc}
			S_{(u_1,0)}:=S\setminus\Big\{\exp^{\Sig}_{\gamma(s)}\big(r\bS(s)\big),\ s\in\T_{\ell},\ 0\le r<u_1\Big\},
			&\mathrm{if}\ u_1> 0,\\
			S_{(u_1,0)}:=S\cup\Big\{\exp^{\Sig}_{\gamma(s)}\big(-r\bS(s)\big),\ s\in\T_{\ell},\ 0\le r<|u_1|\Big\},
			&\mathrm{if}\ u_1\le 0.
		\end{array}
	\right.
\end{equation*}
Furthermore, for $u_2\in(-\delta,\delta)$ we define $S_{(u_1,u_2)}\subset \Sig_{u_2}$ 
as the ``suspension" of $S_{(u_1,0)}$ at height $u_2$:
\begin{equation}
	S_{(u_1,u_2)}:=\Big\{ \exp_{\bp}(u_2\bN_{\Sig}(\bp)),\ \bp\in S_{(u_1,0)}\Big\}.
\end{equation}
The surfaces $S_{\bu}$ are oriented like $S$ and we define
\begin{equation}
	\gamma_{\bu}:=\partial S_{\bu}.
\end{equation}

\begin{figure}[!ht]
	\begin{tikzpicture}[scale=3]
		\filldraw[fill=gray!15, draw=black, thick]
		 plot[domain=-0.8:0.8] ({\x},{0.3*sqrt(1-(\x)^2)})--
		 plot[domain=9/50:-9/50] ({10/3*sqrt(9/100-(\x)^2)},{\x})--
		 plot[domain=0.8:-0.8] ({\x},{-0.3*sqrt(1-(\x)^2)})--
		 plot[domain=-9/50:9/50] ({-10/3*sqrt(9/100-(\x)^2)},{\x})--cycle;
		
		\filldraw[fill=gray!30, draw=black, thin]
		 plot[domain=-0.9:0.9] ({0.93*\x},{0.8*0.3*sqrt(1-(\x)^2)})--
		 plot[domain=24*sqrt(19)/1000:-24*sqrt(19)/1000] ({0.93*sqrt(1-625/36*(\x)^2)},{\x})--
		 plot[domain=0.9:-0.9] ({0.93*\x},{-0.8*0.3*sqrt(1-(\x)^2)})--
		 plot[domain=-24*sqrt(19)/1000:24*sqrt(19)/1000] ({-0.93*sqrt(1-625/36*(\x)^2)},{\x})--cycle;

		\filldraw[fill=gray!30, draw=black, thin]
		 plot[domain=-0.9:0.9] ({0.93*\x},{0.7+0.8*0.3*sqrt(1-(\x)^2)})--
		 plot[domain=24*sqrt(19)/1000:-24*sqrt(19)/1000] ({0.93*sqrt(1-625/36*(\x)^2)},{0.7+\x})--
		 plot[domain=0.9:-0.9] ({0.93*\x},{0.7-0.8*0.3*sqrt(1-(\x)^2)})--
		 plot[domain=-24*sqrt(19)/1000:24*sqrt(19)/1000] ({-0.93*sqrt(1-625/36*(\x)^2)},{0.7+\x})--cycle;
		
		\draw (-2.3,-0.5)--(-0.8,1);
		\draw (-2.3,-0.5)--(1.2,-0.5);
		
		\draw[>=latex,very thin,->] (-0.93,0) -- (-0.93,0.7);
		\draw[>=latex,very thin,->] (0.93,0) -- (0.93,0.7);
	
		\node at (-1.2,0) {$S$};
		\draw[>=latex,->] (-1.15,0)--(-1,0); 
		\node at (0,0) {$S_{(u_1,0)}$};
		\node at (0,0.7) {$S_{(u_1,u_2)}$};
	\end{tikzpicture}
	\caption{The suspended surfaces $S_{\bu}$.}
\end{figure}
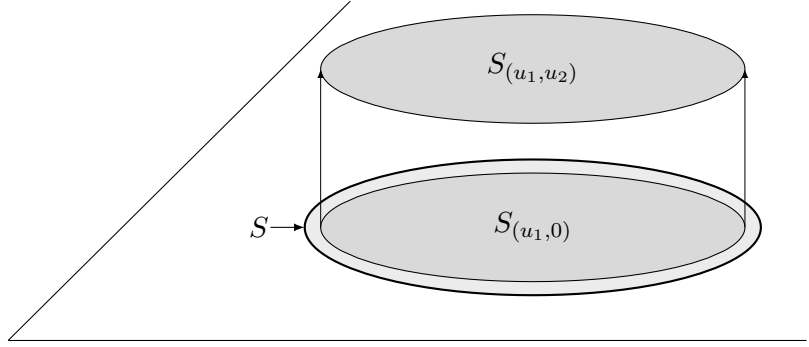

Next, given any smooth probability density $f\in \sD((-\delta,\delta)^2;[0,+\infty))$,
$I_f$ denotes its partial integral
$$
	I_f(u_1,u_2):=\dint_{-\infty}^{u_1} \dint_{-\infty}^{u_2}f(u_1',u_2')\d u_1'\d u_2',
$$
and $m_f$ denotes its marginal
$$
	m_f(u_2):=\dint_\R f(u_1,u_2)du_1.
$$
Then we approximate $\bB$ and $\bA$ by
\begin{equation}\label{eq:def_smooth_knot}
	\left\{
		\begin{array}{l}
			\bB_{f}:=2\pi\alpha\int_{\bu'}f(\bu')[\gamma_{\bu'}]\d\bu',\\\\
			\bA_{f}:=2\pi\alpha\int_{\bu'}f(\bu')[S_{\bu'}]\d\bu'.
		\end{array}
	\right.
\end{equation}
Note that $\supp\,\bB_{f}\subset B_\eps[\gamma]$ and 
that $\supp \bA_{f}\subset B_\delta [S_{(-\delta,0)}]$.
\begin{figure}[!ht]
	\resizebox{0.8\textwidth}{!}{
	\begin{tikzpicture}
		
		\draw[thick] plot[domain=-0.8:0.8] ({0.5*sqrt(1-(\x)^2)},{\x});
		\draw[thick] plot[domain=-0.8:0.8] ({-0.5*sqrt(1-(\x)^2)},{\x});
		\draw[thick] plot[domain=-0.4:0.4] ({\x},{sqrt(1-4*(\x)^2)});
		\draw[thick] plot[domain=-0.4:0.4] ({\x},{-sqrt(1-4*(\x)^2)});
		
		\draw plot[domain=0:1.5] ({\x},0.2*\x*\x);
		\draw[thick,dotted] plot[domain=-1:0] ({\x},{0.2*\x*\x});
		
		\foreach \y in {-3,-2,-1,1,2,3}
		{
			\draw[very thin] plot[domain=0:1.5] ({\x},{0.27*\y+0.2*\x*\x});
			\draw[thin,dotted] plot[domain=-1:0] ({\x},{0.27*\y+0.2*\x*\x});
		}
		\node[right, scale=0.6] at (1.5,{0.2*2.25}) {$\gamma$};
		
		\draw (3.5,0) circle [radius = 1.12];
		
		\foreach \x in {-2,-1,0,1,2}
		{
			\foreach \y in {-2,-1,0,1,2}
			{
			\draw [fill] ({0.3*\x+3.5},{0.3*\y}) circle [radius = 0.03];
			}
		}
		
		\foreach \x in {-3,3}
		{
			\foreach \y in {-1,0,1}
			{
			\draw [fill] ({0.3*\x+3.5},{0.3*\y}) circle [radius = 0.03];
			}
		}
		
		\foreach \x in {-1,0,1}
		{
			\foreach \y in {-3,3}
			{
			\draw [fill] ({0.3*\x+3.5},{0.3*\y}) circle [radius = 0.03];
			}
		}
	\end{tikzpicture}
	}
	\caption{Fibration of the tubular neighborhood by the $\gamma_{\bu}$'s}
	\label{Seifert_fibration}
\end{figure}
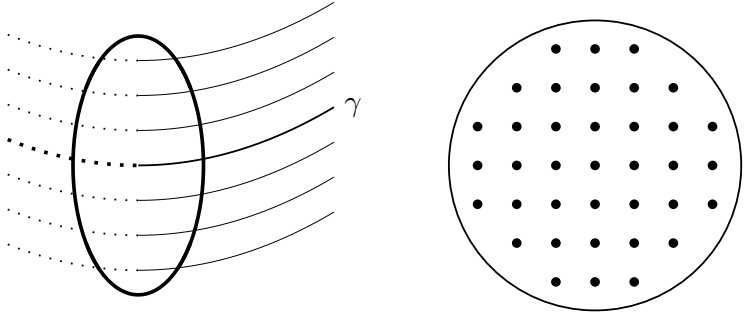

The idea easily carries over to a link. 
Let $K$ be a positive integer and $\uS \in \sS_{\sK}^{(K)}$. 
We can assume that for all $1\le k\le K$, the Seifert surface $S_k$ is embedded in a 
\emph{closed oriented} surface $\Sig_k$. Then the currents 
$$
	\bB=\sum_{k}2\pi\alpha_k[\gamma_k], 
	\quad \bA=\sum_{k}2\pi\alpha_k[S_k]
$$ 
are approximated by
\begin{equation}\label{eq:def_smooth_link}
	\left\{
	\begin{array}{l}
		\bB_{f}:=\sum_{k=1}^K2\pi\alpha_k \int_{\bu'}[(\gamma_k)_{\bu'}]f(\bu')\d\bu',\\\\
		\bA_{f}:=\sum_{k=1}^K2\pi\alpha_k \int_{\bu'}[(S_k)_{\bu'}]f(\bu')\d\bu'.
	\end{array}
	\right.
\end{equation}
By picking $\delta$ small enough, we can assume that we have 
$B_{2\delta}[\gamma]=\cup_k B_{2\delta}[\gamma_k]$.
In that case the support $\supp\,\bB_{f}$ has $K$ connected components, 
each localized around one knot $\gamma_k$.
\begin{proposition}\label{prop:formule_smooth_mag}
	Let $f\in\sD((-\delta,\delta)^2;\R_+)$ be a smooth probability density.
	For $\bA=2\pi\alpha [S]$, $\bA_f$ is a smooth $2$-current. Writing $\alpha_f$ for the associated one-form,
	we have
	\begin{align*}
	\bal_f(\bp)&=2\pi\alpha\, I_f(u_1(\bp),u_2(\bp))(\bN_{\Sig})^{\flat}(\bp),\\
			&=2\pi\alpha\, I_f(u_1(\bp),u_2(\bp))\d u_2(\bp)\in \rT_{\bp}^*\S^3.
	\end{align*}
	Furthermore, $\d \bal_f=\bet_f$ with
	\[
	\bet_f(\bp)=2\pi\alpha\, f(u_1(\bp),\bu_2(\bp))\d u_1(\bp)\wedge \d u_2(\bp),
	\]
	and $\bet_f$ is the $2$-form associated with the one-current $\bB_f$.
	
	Note that the function $I_f(u_1(\bp),u_2(\bp))$ has been canonically extended 
	from $B_{\delta}[\gamma]$ to the bulk $\cup_{u_2\in [-\delta,\delta]} S_{(\delta,u_2)}$ of $B_{\delta}[S]$, since
	\[
		I_f(u_1,u_2)=\dint_{u_1'=-\infty}^{\infty}f(u_1',u_2)\,\d u_1'=m_f(u_2), 
		\quad \forall\,u_1>\delta,\,\forall\,u_2\in\R.
	\]
\end{proposition}

Before proving this result, let us study this gauge in more details. We introduce the phase:
\begin{equation}\label{eq:phase_Phi}
	\Phi_{f}(\bu):=2\pi\alpha\dint_{u=-\infty}^{u_1}\dint_{v=-\infty}^{u_2}f(u,v)\d u\d v.
\end{equation}
We immediately see that $\bal_f$ coincides with
\begin{equation}\label{eq:formule_phase_Phi}
\bal_f=(\partial_{u_2}\Phi_f )\d u_2.
\end{equation}
In the region $u_1>\delta$ of the neighborhood of the Seifert surface, the function $\Phi_f$
does no longer depend on $u_1$, and so $\bal_f$ coincides with the gradient of $\Phi_f$. Indeed, in this region
the phase function coincides with the integrated marginal:
 \begin{equation}\label{eq:def_M_f}
 	M_f(u_2):=\dint_{\R}\d u\dint_{v=-\infty}^{u_2}f(u,v)\d u \d v.
 \end{equation}

\begin{proof}[Proof of Proposition~\ref{prop:formule_smooth_mag}]
	By the definition of the currents $\bA_f$ and $\bB_f$, it is clear that 
	$\partial \bA_f=\bB_f$, so it suffices to prove that
	$\bA_f$ is smooth (as a $2$-vector field). By the musical isomorphisms and the hodge dual we 
	can transform it back to a one-form $\bal_f$, and we will get
	$\d \bal_f=\boldsymbol{\beta}_f$ automatically, where $\boldsymbol{\beta}_f$ is the 2-form associated with $\bB_f$. 
	If we accept the formula for $\bal_f$, then the formula for $\boldsymbol{\beta}_f$ follows immediately.

	Let $\omega\in\Omega^2(\S^3)$. We have
	\begin{align*}
		(2\pi\alpha)^{-1}\big(\bA_f,\omega_i \big)
		&:=\dint_{\bu'\in\R^2}\d\bu' f(\bu')\dint_{S_{\bu'}}\omega_i,\\
		&=\dint_{\bu'\in\R^2}\d\bu' f(\bu')\dint_{S_{\bu'}}g_3^*(\omega_i,\vol_{\Sig_{u_2'}})\vol_{\Sig_{u_2'}}.
	\end{align*}
	Here $g_3^*$ denotes the metric on $2$-forms induced by $g_3$. 
	This quantity equals:
	\begin{multline*}
		\dint_{u_2}\d u_2\dint_{\bp\in S_{(-\delta,u_2)}}g_3^*(\omega_i,\vol_{\Sig_{u_2}})\vol_{\Sig_{u_2}}
		\Big(\dint_{-\delta}^{u_1(\bp)}f(v_1,u_2(\bp))\d v_1\Big)\\
		=\dint_{\bp\in B_\delta[S_{(-\delta,0)}]}\Big(\dint_{-\delta}^{u_1(\bp)}f(v_1,u_2(\bp))\d v_1\Big)
		g_3^*(\omega_i,\vol_{\Sig_{u_2}})\vol_{g_3}.
	\end{multline*}
	Remark that the Hodge dual of $\vol_{\Sig_{u_2}}$ is $(\bN_{\Sig})^{\flat}=\d u_2$
	and that for $\bp\in\Sig_{u_2}$, there holds $\vol_{g_3}(\bp)=\vol_{\Sig_{u_2}}(\bp)\wedge \d u_2(\bp)$.
	This last equality follows from Corollary~\ref{coro:volume_form}.
\end{proof}

We remark that both $\bal_f$ and $\bet_f$ depend on $\bp$ only through $\bu(\bp)$, 
which allows us to use the results of the Dirac operator in two dimensions. 
Observe however that we cannot use the results of \cite{ErdSol01},
because the function 
$$
	\bp\mapsto \bu(\bp)
$$ 
is \emph{not} a Riemannian submersion.

\subsection{Approximation results}
With the previous preparations, we are in a position to state the main theorems of this section.
The proofs of the following theorems are rather involved and require some technical preparations;
they are given in Section~\ref{sec:proof_thm_conv_dirac}. 
In the following two theorems, by a smooth, compactly supported mollifier $\big(\phi_\delta\big)_{\delta>0}$ 
we mean a family
\[
\phi_{\delta}(\bu'):=\frac{1}{\delta^2}\phi_1\big(\frac{\bu'}{\delta}\big),\ \bu'\in\R^2,
\]
where $\phi_1\in \sD((-1,1)^2;[0,+\infty))$ is a probability density with support in $[-1,1]^2$.
\begin{theorem}\label{thm:conv_dirac}
	Fix $K\in\N$, $(\uS,\ua)\in\sS_{\sK}^{(K)}\times (0,1)^K$ and let $\bA=\sum_k 2\pi\alpha_k [S_k]$.
	Let $\big(\phi_\delta\big)_{0<\delta\le 1} \subset \sD((-1,1)^2;[0,+\infty))$ be a smooth mollifier.
	There exists $\delta_0>0$ depending on $S$
	such that for all $0<\delta \leq \delta_0$, the one-form $\bal_{\phi_\delta}$ associated with
	$$
		\bA_{\phi_\delta}=\sum_{k=1}^{K}2\pi\alpha_k \dint_{\bu'}\phi_\delta(\bu')[(S_k)_{\bu'}]\d\bu',
	$$
	is well-defined and smooth. Furthermore the associated family of elliptic operators
	$$
		\big(\cD_{\bA_{\phi_\delta}}\big)_{\delta>0}
		:=\big(\bsigma(-i\nabla+\bal_{\phi_\delta})\big)_{\delta>0}	
	$$
	converges to $\cD_{\bA}^{(-)}$ in the norm-resolvent sense as $\delta$ tends to $0^+$. 
\end{theorem}

Adapting the tools used to prove the above theorem,
we will also show continuity for a homotopy of paths.

\begin{theorem}\label{thm:cont_homot}
	
	Fix $K\in\N$ and $\uS\in \sS_{\sK}^{(K)}$, and let
	$(\phi_\delta)_{0<\delta\le 1}$ be a smooth mollifier in $\sD((-1,1)^2;[0,+\infty))$.
	Let $(\ua(t))_{t\in [0,1]}$ be a continuous family of fluxes
	in $[0,1)^{K}$ ($1$ is excluded). For $0<\delta \leq \delta_0(\uS)$, let $\bA(t)$ and
	$\bA_{\phi_{\delta}}(t)$ be defined as
	$$
		\bA(t):=\sum_{k=1}^K 2\pi\alpha_k(t)[S_k], 
		\quad \bA_{\phi_{\delta}}(t)
		:= \sum_{k=1}^{K}2\pi\alpha_k(t) \dint_{\bu'}\phi_\delta(\bu')[(S_k)_{\bu'}]\d\bu'.
	$$
	Writing $\bal_{\phi_\delta}(t)$ for the associated one-form, we introduce the map:
	\[
	  \begin{array}{rcc}
	      [0,1]\times[0,\delta_0]&\longrightarrow& \Sd,\\
	      (t,\delta) &\mapsto& \cD_{\delta}(t):=\bsigma(-i\nabla+\bal_{\phi_\delta}(t)),
	  \end{array}
	\]
	with the convention $\cD_0(t):=\cD_{\bA(t)}$. This map is continuous in the norm-resolvent sense.
\end{theorem}

\begin{rem}[On the extensions $\cD_{\bA}^{(\pm\cdots \pm)}$]
	In \cite{dirac_s3_paper1} we also introduced the self-adjoint extension $\cD_{\bA}^{(\underline{e})}$,
	with $\underline{e}\in\{+,-\}^K$. The wave functions of its domain are 
	aligned with the magnetic field close
	to the knots $\gamma_k$ with $e_k=1$ and are aligned against 
	it close to the $\gamma_k$ with $e_k=-1$.
	It was shown that it coincides with the operator $\cD_{\bA'}^{(-)}$ where $\bA'$ is the singular gauge
	\[
		\bA':=\sum_{k:\ e_k=1}2\pi(1-\alpha_k)[-S_k]+\sum_{k:\ e_k=-1}2\pi\alpha_k[S_k]
	\]
	and $-S_k$ denotes the Seifert surface $S_k$ with opposite orientation. 
	Let us consider the case of a single Seifert surface $S$; 
	if the operator $\cD_{2\pi\alpha [S]}^{(-)}$ is approximated by 
	the family of elliptic operators $(\cD_{\bA_{\phi_{\delta}}})_{\delta>0}$, then 
	\[
		\cD_{2\pi(1-\alpha)[S]}^{(+)}=\cD_{2\pi\alpha[-S]}^{(-)}
	\]
	is approximated by the family $(\cD_{-\bA_{\phi_{\delta}}})_{\delta>0}$.
	So from the point of view of the smooth approximation procedure, the natural family of fluxes is not $[0,1)^K$ 
	but rather the whole set $(-1,1)^K$, for a given orientation of the knots.
\end{rem}

\begin{rem}[Choice of approximations]
 We have chosen here to approximate the singular magnetic fields, by convoluting in a cross-section.
 By doing so, we can spot the field lines for the simple reason that we are prescribing them. 
 This makes the analysis easier. We could have instead chosen to convolute directly on $\R^3$.
 The same kind of results should hold, at the price of more technicalities.
 
 Let us note the formula in that case. Consider a knot $\gamma$ with Seifert surface $S$ 
 and a smooth, compactly supported mollifier $(\phi_{\delta}(x)\!:=\!\delta^{-3}\phi(x\delta^{-1}))_{\delta>0}$.
 The families $(\phi_{\delta}*[\gamma])_{\delta>0}$ and $(\phi_{\delta}*[S])_{\delta>0}$ define 
 smooth approximations of the currents $[\gamma],[S]$ respectively. A simple calculation gives the formulas which replace \eqref{eq:def_smooth_link}:
 \[
    \phi_{\delta}*[\gamma]=\int\phi_{\delta}(y)[\gamma+y]\d y\quad\&\quad \phi_{\delta}*[S]=\int \phi_{\delta}(y)[S+y]\d y,
 \]
 where $\gamma+y$ and $S+y$ denote the shifted knot and surface.
\end{rem}

\subsection{New results on zero modes on $\S^3$ and $\R^3$}\label{sec:new_results}
Theorem~\ref{thm:cont_homot} raises a natural question, namely whether a
non-trivial spectral flow for the singular Dirac operators $\cD_{\bA(t)}^{(-)}$ implies 
a non-trivial spectral flow for their smooth version $\cD_{\bA_{\phi_\delta}(t)}$ (for $\delta>0$
small enough). This is confirmed in Theorem~\ref{thm:main_zero_existence} below, 
whose proof is given at the end of this section.

We first recall topologies introduced
in \cite{Wahl08}, which were used in \cite{dirac_s3_paper2} 
to study the spectral flow.

\begin{rem}[Topology on $\Sd$]\label{rem:wahl_topology}
	Let $\Lambda\in\sD(\R;\R_+)$ be a bump function
	(centered at $0$), in the sense that $\Lambda$ is even and on its support $[-a,a]\subset \R$, 
	$\Lambda'(x)>0$ for $x\in(-a,0)$ and $\Lambda'(x)<0$ for $x\in(0,a)$.  

	The topology $\sT_{\Lambda}\subset 2^{\Sd}$ is defined 
	as the smallest topology such that for any $\psi\in L^2(\S^3)^2$, the following maps are continuous:
	$$
	\left\{
	\begin{array}{rcl}
		D\in \Sd&\mapsto &(D\pm i)^{-1}\psi\in L^2(\S^3)^2,\\
		D\in \Sd&\mapsto &\Lambda(D)\in\cB\big(L^2(\S^3)^2\big).
	\end{array}
	\right.
	$$
	For two bump functions $\Lambda,\wt{\Lambda}$ with $\supp\,\wt{\Lambda}\subset(\supp\,\Lambda)^{\circ}$,
	we can factorize $\wt{\Lambda}=f\circ\Lambda$ with $f$ continuous, hence functional calculus
	gives the inclusion $\sT_{\wt{\Lambda}}\subset \sT_{\Lambda}$.
	
	\medskip
	
	In \cite{dirac_s3_paper2}, we introduced the notion of bump-continuity: a map $c:\cL\to \Sd$ from a topological space $\cL$
	to $\Sd$ is said to be bump-continuous
	at $x\in\cL$, if there exists a bump function $\Lambda$ centered at $0$ and an open neighborhhood $U$ of $x$
	such that $c_{|_{U}}$ is $\sT_{\Lambda}$-continuous. It is said to be bump-continuous if it is bump-continuous at all points $x\in\cL$. 
\end{rem}

Let us also recall the topology of the norm-resolvent convergence (which coincides with the gap topology: see Theorem~\ref{thm:norm_resv_conv}).

\begin{theorem}\label{thm:main_zero_existence}
	Let $K \in \N$ and $\uS \in \sS_{\sK}^{(K)}$. For a continuous path 
	$\ua(t):[0,1]\to [0,1)^K$, let $\bA(t)$ be given by
	$$
		\bA(t) := \sum_{k=1}^K2\pi\alpha_k(t)[S_k].
	$$
	Furthermore, let $\bA_{\phi_\delta}(t)$ be the smooth approximation~\eqref{eq:def_smooth_link} of 
	$\bA(t)$ for a mollifier $(\phi_{\delta})_{\delta >0}$.
	If the path $\big(\cD_{\bA(t)}^{(-)}\big)_{t\in [0,1]}$ satisfies
	\begin{equation}\label{eq:cond_lift}
			\ker\big(\cD_{\bA(0)}^{(-)}\big)=\ker\big(\cD_{\bA(1)}^{(-)}\big)=\{0\},
	\end{equation}
	there exists $\delta_0>0$ such that for any $0<\delta\le \delta_0$,
	$$
		\Sf\big[ (\cD_{\bA(t)}^{(-)})_{t\in [0,1]}\big] = \Sf\big[(\cD_{\bA_{\phi_\delta}(t)})_{t\in [0,1]}\big].	 
	$$
\end{theorem}

Then the above theorem can be used to show the existence of new zero modes on $\S^3$, and in particular
also $\R^3$. Recall that given a Dirac operator $\cD_{\bA_{\phi_{\delta}}}$ with a smooth magnetic field on $\S^3$, 
one can define a corresponding operator 
on $\R^3$ with the help of the stereographic projection, 
see \cite{ErdSol01}*{Theorem~8.6}.

\begin{corollary}[Zero modes on $\R^3$]\label{coro:zero_mode}
	Take $\big(\cD_{\bA(t)}^{(-)}\big)_{t\in [0,1]}$ as in Theorem~\ref{thm:main_zero_existence}
	with
	$$
		\Sf\big[ (\cD_{\bA(t)}^{(-)})_{t\in [0,1]}\big] \geq 1,
	$$
	and let $(\cD_{\bA_{\phi_\delta}(t)})_{t\in [0,1]}$ be its smooth approximation.
	Then there exists a $t_0 \in [0,1]$ such that
	the Dirac operator $\cD_{\bA_{\phi_\delta}(t_0)}^{\R^3}$ on $\R^3$
	satisfies
	$$
		\dim \ker \big(\cD_{\bA_{\phi_\delta}(t_0)}^{\R^3}\big) \geq 1.
	$$
\end{corollary}

\begin{rem}[Invertibility of the endpoints]
	The free Dirac operator $\sigma(-i\nabla)$ on $\S^3$ has no zero modes
	(say, by Lichnerowicz' formula). So by bump-continuity \cite{dirac_s3_paper2}*{Theorems~11-13-14} 
	(the definition of bump-continuity is given in Remark~\ref{rem:wahl_topology}),
	any path in the bulk of the torus of fluxes $(0,1)^K$ which starts and ends very close to the point $\underline{0}$
	in the torus $\Tf^K$ has invertible endpoints $\cD_{\bA(0)}^{(-)}$ and $\cD_{\bA(1)}^{(-)}$,
	and Theorem~\ref{thm:main_zero_existence} is applicable. 
\end{rem}

Now we give a generic example for which the above holds.

\begin{example}
 Let $K \in \N$ and $\uS \in \sS_{\sK}^{(K)}$.  From \cite{dirac_s3_paper2}*{Theorems~11-14}, we know that up to changing the knots $\gamma_k$
 into other $\gamma_{k'}$ within their corresponding isotopy class and without changing the family $(\link(\gamma_{k_1},\gamma_{k_2}))_{1\le k_1<k_2\le K}$
 of linking numbers, we can assume that the map
 \[
  \ua\in\Tf^K\mapsto \cD_{\bA(\ua)},\ \bA(\ua)=\sum_{k=1}^K2\pi\alpha_k[S_k]
 \]
is $\sT_{\Lambda}$-continuous at the point $(0,\cdots,0)$ for some bump function $\Lambda$. 
It reduces to changing the \emph{Writhe} of each knot, see Remark~\ref{rem:writhe} below. We can also assume \cite{dirac_s3_paper2}*{Theorems~13-14-19}
that there exists a loop $t\in(\R/\Z)\mapsto \ua(t)$ such that $\ua(0)=0$, with $\ua(t)$ in the bulk $(0,1)^K$ for $t\neq 0$
such that the spectral flow of the $\sT_{\Lambda}$-continuous family $(\cD_{\bA(\ua(t))})_{t\in (\R/\Z)}$ is non-trivial.

Such an assumption can be made for instance for a loop where all the fluxes $\alpha_k$'s equal $0$ but $\alpha_{k_0}$ tuned from $0$ to $1$ \cite{dirac_s3_paper2}*{Theorem~19~and~Corollary~23}.
By $\sT_{\Lambda}$-continuity, there exists $\eps>0$ such that for $t\in \S^1$ with $\textrm{dist}_{\R/\Z}(t,0)\le \eps$, the operator $\cD_{\bA(\ua(t))}$ is invertible.
In particular the open path 
\[
  (\cD_{\bA(\ua(t))})_{t\in [\eps,1-\eps]}
\]
has invertible endpoints and has the same spectral flow as the above loop. 
Then we can apply Theorem~\ref{thm:main_zero_existence} to this open path.

\end{example}

\begin{rem}\label{rem:writhe}[Writhe of a knot]
	Let $\gamma$ be a knot embedded in $\R^3$ (seen as the stereographic projection of $\S^3\setminus\{\bp\}$). 
	Then the writhe $\mathrm{Wr}(\gamma)$ of $\gamma$ is the \emph{real} number obtained
	by taking $\gamma_1=\gamma_2=\gamma$ in the Gauss' formula for the linking number of $\gamma_1$ and $\gamma_2$:
	\[
		\mathrm{Wr}(\gamma):=\frac{1}{4\pi}\int_{\gamma}\int_{\gamma}
			\left\langle  \d \mathbf{r}_1\wedge \d\mathbf{r}_2,\,\frac{\mathbf{r}_1-\mathbf{r}_2}{|\mathbf{r}_1-\mathbf{r}_2|^3}\right\rangle_{\R^3}.
	\]
	By the C\u{a}lug\u{a}reanu-White-Fuller Theorem (see \cite{White69}), it
	coincides with $-\tfrac{1}{2\pi}$ times the total relative torsion $\int_{\gamma}\tau_S$ associated with a Seifert surface $S$ of $\gamma$. 
	Using the latter formula as a definition, the writhe is a conformal invariant: 
	if we conformally change the metric in the vicinity of $\gamma$ (\emph{a fortiori} if we take the $\S^3$-metric),
	the writhe remains unchanged. We refer the reader to \cite{dirac_s3_paper2}*{Appendix} and \cite{DeTurGlu08} for more details.
\end{rem}

\begin{proof}[Proof of Theorem~\ref{thm:main_zero_existence}]
	We consider the homotopy of Theorem~\ref{thm:cont_homot};
	for any $t$ and any $k$ we have $0\le \alpha_k(t)<1$,
	since $1$ is excluded from the set of admissible fluxes 
	(we cannot increases a flux $2\pi\alpha$ up to $2\pi$),
	and the path $(\ua(t))_t$ has to be open in the sense that $\ua(0)\neq \ua(1)$ in 
	$\Tf^K$, otherwise the spectral flow is trivial
	(because any closed path in $[0,1)^K$ has zero spectral flow).
	Therefore we are led to use the homotopy property for open paths \cite{Wahl08}*{Section 2}:
	Let $\Lambda$ be some bump function. If the family
	$$
		(\cD_{(\delta,t)})_{(\delta,t)\in [0,\delta_0]\times [0,1]}
	$$
	is $\sT_{\Lambda}$-continuous in $\Sd$ such that
	for all $\delta\in [0,\delta_0]$, the operators $\cD_{(\delta,0)}$ and $\cD_{(\delta,1)}$ are invertible, 
	then
	$$
		\Sf\big[(\cD_{(0,t)})_{t\in [0,1]} \big]=\Sf\big[(\cD_{(\delta_0,t)})_{t\in [0,1]} \big].
	$$
	In our case, 
	$\cD_{\bA(0)}^{(-)}$ and $\cD_{\bA(1)}^{(-)}$ are invertible, so
	by $\sT_{\Lambda}$-continuity there exists $\delta_0=\delta_0(\uS)>0$ such that 
	$\cD_{\bA_{\phi_\delta}(0)}$ and $\cD_{\bA_{\phi_\delta}(1)}$
	are invertible for any $0<\delta<\delta_0$,
	and the claim follows.
\end{proof}

\subsection{The Pauli operator}\label{sec:Pauli_op}
For $K\in\N$ and $(\uS,\ua)\in \sS_{\sK}^{(K)}$, let $\bA:=\sum_{k=1}^K2\pi\alpha_k[S_k]$ 
be the associated singular magnetic potential. Theorem~\ref{thm:conv_dirac} suggests that
$\cD_{\bA}^{(-)}$ is the correct Dirac operator, hence the Pauli operator $T_{\bA}$, as the square of the Dirac operator,
is naturally defined as 
\[T_{\bA}:=\big( \cD_{\bA}^{(-)}\big)^2.
\]

We can also consider $T_{\bA}$ from the point of view of its 
associated quadratic form $q_{\bA}$ with form domain $\dom(\cD_{\bA}^{(-)})$.
We get the inclusions:
\[
q_{\min}\subset q_{\bA}\subset q_{\max},
\]
where $q_{\min}$ and $q_{\max}$ are the quadratic forms associated with $(\cD_{\bA}^{(\min)})^*\cD_{\bA}^{(\min)}$ and
 $\cD_{\bA}^{(\min)}(\cD_{\bA}^{(\min)})^*$ with form domains $\dom(\cD_{\bA}^{(\min)})$ and $\dom\,(\cD_{\bA}^{(\min)})^*$
 respectively.
We emphasize that $\dom(\cD_{\bA}^{(\min)})$ coincide with $H_{\bA}^1(\S^3)^2$ (see the introduction),
where $H_{\bA}^1(\S^3)$ can be seen as the form domain of the spinless magnetic Schr\"odinger operator.
In order to define the Pauli operator, a careless application of the quadratic form method  
would have led to $\cD_{\bA}^{(\min)}(\cD_{\bA}^{(\min)})^*$. By our discussion above, this is the wrong definition. 
In particular, this operator has the following characteristic, proved in Appendix~\ref{sec:proof_infinite}.
\begin{proposition}\label{prop:infinite}
 The kernel $\ker\,\cD_{\bA}^{(\min)}(\cD_{\bA}^{(\min)})^*=\ker\,(\cD_{\bA}^{(\min)})^*$ is infinite dimensional.
\end{proposition}
The reader can find in the proof of \cite{dirac_s3_paper1}*{Theorem~29} an element of this kernel which is not a zero mode
for $\bA=2\pi \alpha[D_0]$, where $\partial D_0=\gamma_0$ is a circle and $\alpha\in (\tfrac{1}{2},1)$. In fact in this simple case, 
the whole kernel of the maximal operator can be computed explicitly. 
It is indeed infinite dimensional while the kernel of $\cD_{2\pi\alpha[D_0]}^{(-)}$
is trivial for all $\alpha\in\Tf$ (this is the statement of \cite{dirac_s3_paper1}*{Theorem~29}).

\section{Some criteria for the non-existence of zero modes}\label{sec:crit_inexistence}
From \cite{dirac_s3_paper1}, we know that given a link $\gamma\subset\S^3$, if the 
fluxes carried by the knots are too small,
then the kernel of the corresponding Dirac operators is trivial. 
In this section, we give a more quantitative statement of this fact
by using Clifford analysis in the spirit of \cite{Arrizetal14,HarmAnDir}, 
and using results of \cite{HardSpacSingInt}.
The results presented in this section are derived directly on $\R^3$, 
yet it is clear that one could derive similar results\footnote{Elliptic regularity and conformal invariance of the kernel 
imply a one-to-one correspondence between the kernels of the Dirac operators for the two metrics.} on $\S^3$.
Furthermore, in \cite{dirac_s3_paper1} we gave the construction
for Dirac operators with magnetic links on $\S^3$, but one
could proceed in a very similar way to define them on $\R^3$; many of the proofs
actually become a lot simpler due to the flat geometry of $\R^3$; therefore we 
leave this task to the reader.

We begin by introducing some technical tools which will be the basis for deriving the main
results in the following sections.

\subsection{Some results in Clifford analysis}\label{sec:clifford}
Let us recall some results in Clifford analysis,
in particular the Plemelj--Sokhotski jump formulas.
The following discussion is based on results presented in \cite{HarmAnDir}*{Theorem.~4.4}.

We consider a closed set in $\R^3$ and assume that it is a smooth, connected surface $\Sigma$ 
which is either the boundary of a connected bounded open set $\Omega_{\textrm{int}}$ 
and an unbounded open set $\Omega_{\textrm{ext}}$, or a smooth graph. We write $\bn_\Sigma$ for the normal to 
$\Sigma$ pointing towards $\Omega_{\textrm{ext}}$, and call $\Sigma_{\textrm{ext}}$ the boundary of $\Omega_{\textrm{ext}}$
(``different" from the other side $\Sigma_{\textrm{int}}$ which is the boundary of $\Omega_{\textrm{int}}$).
We denote by $R_0$ the pseudo-differential operator $\frac{1}{\bsigma\cdot(-i\nabla^{\R^3})}$ which acts on (a subset
of) tempered distributions $\mathscr{S}'(\R^3)^2$. 
For any element $u\in C^1(\Sigma)^2$, the limits
$$
	\left\{
		\begin{array}{rcl}
			P_{+}u(\bx)&:=&\lim_{\eps\to 0^+}\Big(R_0 \big(i(\bsigma\cdot \bn_\Sigma) u\big) \Big)
			(\bx+\eps \bn_{\Sigma})\\
			P_{-}u(\bx)&:=&-\lim_{\eps\to 0^+}\Big(R_0 \big(i(\bsigma\cdot \bn_\Sigma) u\big) \Big)
			(\bx-\eps \bn_{\Sigma})
		\end{array}
	\right.
$$
exist in $L^2(\Sigma,d\mu_\Sigma)^2$. Moreover $P_{\pm}$ continuously extend 
to $L^2(\Sigma,d\mu_\Sigma)^2$ and define spectral projections onto the positive (resp. negative) spectrum of a 
bounded operator 
$$
	2\sC :L^2(\Sigma,d\mu_\Sigma)^2 \to L^2(\Sigma,d\mu_\Sigma)^2,
$$
satisfying $4\sC^2 = 1$.
In the literature, $\sC$ is often called the \textit{Cauchy operator}
and the operators $P_{\pm}$ are called the \textit{Hardy projections}. Then the well known
Plemelj--Sokhotski formulas can be summarized as
\begin{equation}\label{eq:ple_sokh}
	P_{\pm} = \frac{1}{2}\pm \sC.
\end{equation}
The respective ranges of $P_{\pm}$ are called the Hardy spaces,
$$
	\mathscr{H}^2_{\pm}(\Sigma) := \ran{P_{\pm}},
$$
and we denote the angle between the two closed subspaces $\mathscr{H}^2_{\pm}(\Sigma)$ by
\begin{align*}
	\theta_\Sigma
	:=\arccos\bigg\{\underset{f_{\pm}\in\mathscr{H}^2_{\pm}(\Sigma)\setminus\{0\}}{\sup}\dfrac{\Re
	\langle f_+,f_-\rangle_{L^2(\Sigma)^2}}{\norm{f_+}_{L^2(\Sigma)^2}\norm{f_-}_{L^2(\Sigma)^2}}\bigg\}\in(0,\pi/2].
\end{align*}
We emphasize that $\sC$ is generally \emph{not} self-adjoint with respect to 
$\langle\cdot\,,\cdot\rangle_{L^2(\Sigma)^2}$, but one has the bound \cite{HardSpacSingInt}*{Theorem~2.8}
\begin{equation}\label{bound_defect_s_a_cauchy_op}
	\norm{\sC-\sC^*}_{\mathcal{B}(L^2(\Sigma)^2)} \leq 2\cot(\theta_\Sigma).
\end{equation}
Furthermore, according to \cite{HardSpacSingInt}*{Theorem~4.17}, for a uniformly rectifiable open set 
$\Omega$ with $\partial\Omega=\partial\overline{\Omega}$ 
(see \cite{HardSpacSingInt}*{Appendix} for the definition), the operator 
$\sC$ is self-adjoint \emph{if and only if} $\partial\Omega$ is a plane or a 
sphere\footnote{Note that through the stereographic projection on $\R^3$, 
planes and spheres on $\R^3$ correspond to spheres on $\S^3$.}.

\subsection{Dirac operators with magnetic knots carrying small fluxes}\label{sec:dirac_op_small_fluxes}

Let $\gamma=\cup_{k=1}^K\gamma_k\subset\R^3$ be a link (where each connected component is identified 
with its $g_{\R^3}$-arclength parametrization, after having chosen a base point). 
Let $\ua\in \Tf^K$; if two or more fluxes are the same,
using Seifert algorithm, we can pick a common Seifert surface for the corresponding sub-links. 
Doing so, we obtain a singular gauge
\begin{equation}\label{gauge_common_seif}
	\bA:=\sum_{m=1}^M2\pi\widetilde{\alpha}_m [S_m]
\end{equation}
for $\bB:=\sum_{k=1}^K2\pi\alpha_k [\gamma_k]$, where $\{\widetilde{\alpha}_m,1\le m\le M\}=\{ \alpha_k,1\le k\le K\}$.\\

\paragraph{\bf{Assumption (A)}} 
	For all $1\le m\le M$ and all $1\le k\le K$, the surface $S_m$ 
	is embedded in a closed compact oriented surface or a smooth graph $\Sigma$
	and $\Sigma$ either contains the corresponding surface $S_{m'}$ to $\gamma_k$ or
	is transverse to $\gamma_k$.

We label the different closed surfaces $\Sigma$ and write $\cup_j \Sigma_j$ their union.
Then we have the following proposition, which is proved in Appendix~\ref{sec:proof_square_integrability}.
\begin{proposition}\label{prop:small_fluxes}
	Let $K\in \N$, and $\gamma=\cup_{k=1}^K\gamma_k\subset\R^3$ be a link. 
	Let $\ua\in [0,\tfrac{1}{2})^K$ and $\bA$ be a singular gauge of type \eqref{gauge_common_seif} 
	satisfying Assumption (A) above.
	Then, for any $\psi\in \dom\big( \cD_{\bA}^{(-)}\big)$ and any surface $\Sigma_j$, 
	the traces of $\psi$ on both sides of $\Sigma_j$ exist and are 
	elements of $L^2(\Sigma_j)^2$. 
\end{proposition}

To prove the main theorem of this section, we need to make slightly stronger assumptions:
\begin{enumerate}
	\item All the $\gamma_k$'s are contained in a smooth surface $\Sigma=\partial\Omega_{\textrm{int}}$ where 
	$\Omega_{\textrm{int}}\subset \R^3$ is a bounded open set (or the bottom of a smooth graph).
	
	\item	There exists a family of \emph{disjoint} surfaces $(S_m)_m$ embedded in $\Sigma$
	such that $\partial S_m= \sum_{k \in I_j}\gamma_k$
	with $\alpha_{k_1} = \alpha_{k_2} =: \widetilde{\alpha}_m$ for all $k_1,k_2 \in I_j$.
	
	\item The fluxes $\alpha_k$ are in $(0,\tfrac{1}{2})$.
\end{enumerate}

The surfaces $S_m$ have a relative orientation with respect to $\Sigma$, which we denote by $\eps_m \in \{\pm1\}$. 
Then the $\gamma_k$'s are split into two subsets $G_{\pm}$ 
with $\gamma_k\in G_{\eps_m}$, $k\in I_j$, and we define
$$
	\delta\alpha :=\sup_{\gamma_{k_{\pm}}\in G_{\pm}}(\alpha_{k_+}+\alpha_{k_-})\in(0,1).
$$
\begin{theorem}\label{thm:non_existence}
	Let $\bB = \sum_{k=1}^K2\pi\alpha_k[\gamma_k]$ be a magnetic link satisfying points (1)-(2)-(3) above with gauge 
	\eqref{gauge_common_seif}. 
	Then the associated Dirac operator $\mathcal{D}_{\bA}^{(-)}$ has no zero modes in the following two cases:
	\begin{enumerate}
		\item All the $S_m$'s share the same orientation in $\Sigma$
		and
		\begin{align}\label{eq:thm_assum_1}
			\norm{\sC-\sC^*}_{\cB(L^2(\Sigma)^2)} < \cot(\pi\sup_k\alpha_k).
		\end{align}
		
		\item If the orientations of the $S_m$'s are different, then we require $\delta\alpha<2^{-1}$ and that
		\begin{align}\label{eq:thm_assum_2}
			\norm{\sC-\sC^*}_{\cB(L^2(\Sigma)^2)} < \cos(\pi\delta \alpha).
		\end{align}
	\end{enumerate}
\end{theorem}

\begin{proof}
	Let us assume that there exists a zero mode $\psi\in \dom(\mathcal{D}_{\bA}^{(-)})$.
	A computation shows that $u_0=\sigma\cdot(-i\nabla^{\R^3})\psi\in L^2(\Sigma)^2$ and 
	writing $u:=-i(\boldsymbol{\sigma}\cdot \bn_{\Sigma})u_0$ we have $\psi_{|_{\Sigma_{\mathrm{ext}}}}=P_{+}u\in L^2(\Sigma)^2$ 
	and $\psi_{|_{\Sigma_{\mathrm{int}}}}=-P_{-}u\in L^2(\Sigma)^2$. 
	We see $u_0$ as a distribution in $H^{-1}(\R^3)^2$. Then a computation (Stokes formula) shows that $u_0$ is in $L^2(\Sig)^2$ and coincides with:
	\[
	 -i\boldsymbol{\sigma}\cdot \mathbf{n}_{\Sig}(\psi_{|_{\Sig_{\ext}}}-\psi_{|_{\Sig_{\mathrm{int}}}}).
	\]
	The distribution $u_0$ is supported on $\cup_mS_m$, since 
	the spinor $\psi$ satisfies
	\begin{align}\label{eq:jump_cond}
		\chi_{S_m}\psi_{|_{\Sigma_{\mathrm{int}}}} &= e^{2i\eps_m\widetilde{\alpha}_m\pi}\chi_{S_m}\psi_{|_{\Sigma_{\mathrm{ext}}}},\nn\\
		\chi_{\Sigma\setminus(\cup S_m)}\psi_{|_{\Sigma_{\mathrm{ext}}}}
		&=\chi_{\Sigma\setminus(\cup S_m)}\psi_{|_{\Sigma_{\mathrm{int}}}}.
	\end{align}
	We recall that $\Sig_{\mathrm{int}},\Sig_{\ext}$ denote the two sides of $\Sig$. 
	We may assume that $\norm{P_+u}_{L^2(\Sigma)^2}=\norm{P_-u}_{L^2(\Sigma)^2}=1$, since the two boundary
	values only differ by a phase.
	
	\smallskip
	\noindent\textit{Assumption~(1):}
	Now all the $S_m$'s have the same orientation in $\Sigma$, $\eps_m = \eps$ for all $m$, where $\eps$ is either $1$ or $-1$. 
	The jump conditions \eqref{eq:jump_cond} 
	can be rewritten as a condition over $u$; using \eqref{eq:ple_sokh} we obtain
	\[
	  \begin{array}{rcl}
		\chi_{\Sig\setminus \cup_m S_m}\left(\textstyle{\frac{1}{2}} + \sC\right)u
		&=&-\chi_{\Sig\setminus \cup_m S_m}\left(\textstyle{\frac{1}{2}} - \sC\right)u\\
		\chi_{S_m}\left(\textstyle{\frac{1}{2}} + \sC\right)u 
		&=& -e^{-2\pi i \eps \widetilde{\alpha}_m}\chi_{S_m}\left(\textstyle{\frac{1}{2}} - \sC\right)u,
	  \end{array}
	\]
	which yields
	\[
	  \begin{array}{rcl}
		\chi_{\Sig\setminus \cup_m S_m}u&=&0,\\
		\chi_{S_m} (2\mathscr{C}) u
		&=&i\cot(\eps\pi\widetilde{\alpha}_m) \chi_{S_m}u.
	\end{array}
	\]
	(We recover $\supp u_0\subset \cup_mS_m$ since $u =-i(\bsigma\cdot \mathbf{n}_S)u_0$).
	Let us assume that $\eps = 1$ (the case $\eps = -1$ can be dealt with in a similar way).
	We have
	\begin{align*}
		-\langle 2i(\sC-\sC^*)u,\,u \rangle_{L^2(\Sigma)^2}&=\langle -2i\sC u,\,u \rangle_{L^2(\Sigma)^2}+\langle u,\,-2i\sC u \rangle_{L^2(\Sigma)^2}\\
		&=2 \sum_m\cot(\pi\widetilde{\alpha}_m)\dint_{S_m}|u|^2\\
		&\ge 2\cot(\pi\sup_k \widetilde{\alpha}_m)\norm{u}^2_{L^2(\Sigma)^2}.
	\end{align*}
	Furthermore,
	\begin{align*}
		\left|\langle 2i(\sC-\sC^*)u,\,u \rangle_{L^2(\Sigma)^2}\right| 
		&\leq 2\norm{\sC-\sC^*}_{\cB(L^2(\Sigma)^2)} \norm{u}_{L^2(\Sigma)^2}^2,
	\end{align*}
	so that $\cot(\pi\sup_m \widetilde{\alpha}_m) \leq \norm{\sC-\sC^*}_{\cB(L^2(\Sigma)^2)}$, 
	contradicting \eqref{eq:thm_assum_1}.
	
	\smallskip
	\noindent\textit{Assumption~(2):}
	From \eqref{eq:jump_cond} it follows that
	$$
		\langle \psi_{|_{\Sigma_{\mathrm{int}}}} ,\psi_{|_{\Sigma_{\mathrm{ext}}}}\rangle_{L^2(\Sigma)^2}
		=\sum_m\left[ e^{2i\pi\eps_m\widetilde{\alpha}_m}\dint_{S_m}|\psi_{|_{\Sigma_{\mathrm{ext}}}}|^2\right]
		+\dint_{\Sigma\setminus(\cup_m S_m)}|\psi_{|_{\Sigma_{\mathrm{ext}}}}|^2.
	$$
	Since $\int_{\Sigma}|\psi_{|_{\Sigma_{\mathrm{ext}}}}|^2 = 1$, one can conclude that
	$$
		\langle \psi_{|_{\Sigma_{\mathrm{int}}}} ,\psi_{|_{\Sigma_{\mathrm{ext}}}}\rangle_{L^2(\Sigma)^2}
		 \in \textrm{Conv}\Big(\{1,e^{2i\pi\eps_m \widetilde{\alpha}_m}\}_{1\le m\le M} \Big).
	$$
	Furthermore, with \eqref{eq:ple_sokh} one can easily show
	\begin{align*}
		\langle (\sC-\sC^*)f_+,f_-\rangle_{L^2(\Sigma)^2} 
		= \langle f_+,f_-\rangle_{L^2(\Sigma)^2},
		\quad \forall f_{\pm} \in \mathscr{H}^2_{\pm}(\Sigma),
	\end{align*}
	so that
	\begin{align*}
		\left|\langle \psi_{|_{\Sigma_{\mathrm{ext}}}} ,\psi_{|_{\Sigma_{\mathrm{int}}}}\rangle_{L^2(\Sigma)^2}\right|
		&\leq \left|\langle (\sC-\sC^*) \psi_{|_{\Sigma_{\mathrm{ext}}}} ,\psi_{|_{\Sigma_{\mathrm{int}}}}\rangle_{L^2(\Sigma)^2}\right|\\
		&\leq \norm{\sC-\sC^*}_{\cB(L^2(\Sigma)^2)}.
	\end{align*}
	(Just observe that $f_{\pm}=(\tfrac{1}{2}\pm \sC)f_{\pm}$ which gives $\sC f_{\pm}=\pm\tfrac{1}{2}f_{\pm}$ since $4\sC^2=1$). 
	Since we have assumed that $\delta\alpha<2^{-1}$, we obtain
	$$
		\mathrm{dist}_{\mathbb{C}}\Big(0;  \textrm{Conv}\Big(\{1,e^{2i\pi\eps_m \widetilde{\alpha}_m}\}_{1\le m\le M} \Big)\Big)
		=\cos(\pi\delta\alpha)
		<\norm{\sC-\sC^*}_{\cB(L^2(\Sigma)^2)},
	$$
	which contradicts \eqref{eq:thm_assum_2}.
\end{proof}

Now we state two rather straightforward consequences from the above theorem, which uses the estimate \eqref{bound_defect_s_a_cauchy_op}.
For the first theorem, we recall that 
the Cauchy operator the surface $\Sig$ is self-adjoint in $L^2(\Sigma)^2$ if and only if $\Sigma$ is a plane or a sphere.
We emphasize that any knot embedded in a plane or a sphere is a Jordan curve, hence is isotopic to the circle in $\S^3$ and $\R^3$.

\begin{corollary}\label{coro:inex1}
	Let $(\gamma_k)_{k}$ be any collection of (non-intersecting) knots in a plane or in a sphere $\Sigma$ in $\R^3$ such that their interiors $S_k$ 
	in $\Sigma$ all share the same orientation with respect to $\Sigma$. For any  $\bA=\sum_k 2\pi\alpha_k [S_k]$, with $\alpha_k \in (0,1/2)$,
	the corresponding Dirac operator $\mathcal{D}_{\bA}^{(-)}$ has no zero modes.
\end{corollary}

If we choose to remove the restrictions on the geometry of the link, 
we obtain the following corollary. Below $E_\Sig(\gamma)$ denotes the set of 
oriented smooth surfaces containing a Seifert surface for $\gamma$ and which are either closed and compact 
or a graph of a smooth function. (As a Seifert surface is bicollared, the set $E_\Sig(\gamma)$ is never empty).
\begin{corollary}\label{coro:inex2}
	Let $\gamma$ be an arbitrary link in $\R^3$. We introduce $\eps(\gamma,\sC)\ge 0$ as
	\[
	 \eps(\gamma,\sC):=\inf_{\Sig\in E_\Sig(\gamma)}\norm{\sC_{\Sig}-\sC_{\Sig}^*}_{\mathcal{B}(L^2(\Sigma)^2)}
	\] 
	where $\sC_{\Sig}$ is the Cauchy operator associated with $\Sig$. 
	Let $0<\alpha(\gamma)\le \tfrac{1}{2}$ be:
	\[
	 \alpha(\gamma):=\min\Big(\frac{1}{\pi}\mathrm{arccot}(\eps(\gamma,\sC)), \frac{1}{2}\Big).
	\]
	Then for all $0<\alpha'<\alpha(\gamma)$ and any Seifert surface $S$ for $\gamma$ the kernel of $\cD_{2\pi\alpha'[S]}^{(-)}$ is trivial.
\end{corollary}
\begin{rem}
	We can use the estimate~\ref{bound_defect_s_a_cauchy_op} to give a (larger, hence weaker) bound depending on the angle between Hardy spaces.
	Let $\theta(\gamma)\in(0,\tfrac{\pi}{2}]$ be the supremum of $\theta_{\Sig}$ over the set $E_\Sig(\gamma)$. Then if
	\[
	0<\alpha<\widetilde{\alpha}(\gamma):=\frac{1}{\pi}\mathrm{arccot}\big(\frac{1}{2}\cot(\theta(\gamma)) \big),
	\]
	then the kernel of $\cD_{2\pi\alpha'[S]}^{(-)}$ is trivial, where $S$ is any Seifert surface for $\gamma$.
\end{rem}

\section{Technical results \& preparations for the proofs of Theorems~\ref{thm:conv_dirac}~\&~\ref{thm:cont_homot}}\label{sec:tech_result}

The first part of this section recalls the definition of different curvatures, needed later.
The second part of this section is devoted to the proof of some technical results on the 
approximate Cartesian coordinates $(s,\bu)$
which will be used in the proof of Theorem~\ref{thm:conv_dirac}, especially the expression of the free Dirac operator in these coordinates 
(Proposition ~\ref{prop:expr_free_dirac}).
The last part contains convergence results for two-dimensional Dirac operators with singular magnetic fields, 
and the extension to a model case for singular three-dimensional Dirac operators.
The latter will be of great importance for the proof of the main theorems.

\subsection{Curvatures on $\Sig$}

Recall $\Sig\subset \S^3$ is a smooth closed oriented surface on $\S^3$.
The curvature tensor \cite{Spivakvol2}*{Chapter 6}, the mean curvature, 
the extrinsic curvature (in $\S^3$) and the intrinsic curvature of $\Sig$ 
are denoted by 
$$
	R_{\Sig},H_{\Sig},K_{\Sig}^{\ext},K_{\Sig}^{\Int}.
$$	
We recall that the shape operator of $\Sig$ (in $\S^3$) is the tensor defined by:
\[
	\begin{array}{rcl}
		\rT\Sig&\longrightarrow& \rT\Sig\\
		(\bp,X(\bp))&\mapsto& \big(\bp,-\nabla_{X(\bp)}\bN_{\Sig}(\bp)\big).
	\end{array}
\]
The extrinsic curvature $K_{\Sig}^{\ext}$ and the mean curvature $H_{\Sig}$ 
correspond to the determinant and to half the trace of the shape operator. 
The intrinsic curvature $K_{\Sig}^{\Int}$ equals the number:
\[
	K_{\Sig}^{\Int}(\bp)=g_3\big(R_{\Sig}(\be_1,\be_2)\be_2,\be_1\big)(\bp)
\]
where $(\be_1(\bp),\be_2(\bp))$ is an orthonormal basis of $\rT_{\bp}\Sig$.
The two curvatures are related by \cite{Spivakvol4}*{Chapter 7}
\[
	K_{\Sig}^{\Int}=K_{\Sig}^{\ext}+K_{\S^3}^{0},
\]
where $K_{\S^3}^{0}$ is the (constant) scalar curvature of $\S^3$.
We write:
\begin{equation}\label{eq:def_a_i_b_j}
	\left\{
		\begin{array}{rcl}
			\nabla_{\bT_{\Sig}}\bN_{\Sig}&=:&a_0\bT_{\Sig}+b_0\bS_{\Sig},\\
			\nabla_{\bS_{\Sig}}\bN_{\Sig}&=:&a_1\bT_{\Sig}+b_1\bS_{\Sig}.
		\end{array}
	\right.
\end{equation}

\subsection{The Coordinates $(s,\bu)$ and the frame $(\bT_{\Sig},\bS_{\Sig},\bN_{\Sig})$}

The main goal of this section is to write down the free Dirac operator in the local coordinates $(s,\bu)$
(Proposition~\ref{prop:expr_free_dirac}).

\begin{lemma}\label{lem:deriv_N}
	Let $(v,w)\in U\subset \R^2\mapsto \phi(v,w)\in \Sig $ be a local chart, then we have
	$$
	\left\{
		\begin{array}{l}
		(\partial_v\bN_{\Sig})(\phi(v,w)) = \big(\nabla_{\partial_v \phi}\bN_{\Sig}\big)(\phi(v,w)),\\\\
		(\partial_w\bN_{\Sig})(\phi(v,w)) = \big(\nabla_{\partial_w \phi}\bN_{\Sig}\big)(\phi(v,w)).
		\end{array}
	\right.
	$$
\end{lemma}
\begin{proof}
	We prove the result only for $v$, the one for $w$ follows by symmetry.
	It suffices to check that in $\C^2$ we have
	$$
		\cip{\partial_v\bN_{\Sig}(\phi(v,w))}{\phi(v,w)}_{\C^2}=0,
	$$
	where $\phi(v,w)\in\S^3\subset\C^2$.
	To do so, we freeze $w$ and differentiate the map
	$$
		v\mapsto \cip{\bN_{\Sig}(\phi(v,w))}{\phi(v,w)}_{\C^2}\equiv 0,
	$$
	which gives
	\begin{align*}
		0&=\cip{\partial_v\bN_{\Sig}(\phi(v,w))}{\phi(v,w)}_{\C^2}+\cip{\bN_{\Sig}(\phi(v,w))}{\partial_v \phi(v,w)}_{\C^2}\\
		&=\cip{\partial_v\bN_{\Sig}(\phi(v,w))}{\phi(v,w)}_{\C^2}+g_3(\bN_{\Sig},\partial_v \bp)(v,w),\\
		&=\cip{\partial_v\bN_{\Sig}(\phi(v,w))}{\phi(v,w)}_{\C^2}.
	\end{align*}
\end{proof}
Now we give the formula for the pushforward of vector fields through $F_{\gamma}$ defined in \eqref{eq:cart_coord}.
\begin{proposition}\label{prop:push_forward}
	Fix $(s,\bu)\in\dom(F_\gamma)$, and let $y_T(s,\cdot)$ be the solution to
	\begin{equation}\label{eq:system_ODE}
		\left\{
		\begin{array}{l}
			(\partial_{u_1}^2y_{T})(s,u_1)+K_{\Sig}^{\Int}(c(s,u_1))y_{T}(s,u_1)=0,\\
			y_{T}(s,0)=1\quad\&\quad (\partial_{u_1}y_T)(s,0)=-\kappa_g(s).
		\end{array}
		\right.
	\end{equation}
	Then one has
	$$
		\left\{
		\begin{array}{rcl}
			(F_\gamma)_*(\partial_s)(F_\gamma(s,\bu))&
			=&y_T\big[\cos(u_2)\bT_{\Sig}+\sin(u_2)\nabla_{\bT_{\Sig}}\bN_{\Sig} \big](c(s,u_1)),\\
			(F_\gamma)_*(\partial_{u_1})(F_\gamma(s,\bu))&=&(\cos(u_2)\bS_{\Sig}+\sin(u_2)\nabla_{\bS_{\Sig}}\bN_{\Sig} )(c(s,u_1)),\\
			(F_\gamma)_*(\partial_{u_2})(F_\gamma(s,\bu))&=&-\sin(u_2)c(s,u_1)+\cos(u_2)\bN_{\Sig}(c(s,u_1)).
		\end{array}
		\right.
	$$
	More generally, let $\phi:(v,w)\mapsto \phi(v,w)$ be a local chart of $\Sig$. 
	Writing $F_{\phi}$ for the corresponding local chart on $\S^3$,
	$$
		F_{\phi}:(v,w,u_2)\mapsto \cos(u_2)\phi(v,w)+\sin(u_2)\bN_{\Sig}(\phi(v,w)),
	$$
	we have
	\[
		\left\{
		\begin{array}{ccl}
			\big(F_{\phi}\big)_*(\partial_v)&=&\cos(u_2)\partial_v \phi+\sin(u_2)\nabla_{\partial_v \phi}\bN_{\Sig}(\phi),\\
			\big(F_{\phi}\big)_*(\partial_w)&=&\cos(u_2)\partial_w \phi+\sin(u_2)\nabla_{\partial_w \phi}\bN_{\Sig}(\phi),\\
			\big(F_{\phi}\big)_*(\partial_{u_2})&=&-\sin(u_2) \phi+\cos(u_2)\bN_{\Sig}(\phi).
		\end{array}
		\right.
	\]
\end{proposition}

\begin{rem}
	Remark that by definition,
	\[
		(F_\gamma)_*(\partial_{u_2})(F_\gamma(s,\bu))=\bN_{\Sig}(F_\gamma(s,\bu)).
	\]
\end{rem}
\begin{proof}[Proof of Proposition~\ref{prop:push_forward}]
	We only deal with $F_\gamma$, the case of $F_\phi$ is left to the reader.
	The computation of $(F_\gamma)_*(\partial_{u_2})(F_\gamma(s,\bu))$ is obvious.
	
	By definition we have
	$$
		(\partial_{u_1}c)(s,u)=\bS_{\Sig}(c(s,u)),
	$$
	and the computation of $(F_\gamma)_*(\partial_{u_1})(F_\gamma(s,\bu))$ follows easily.

	Observe that $c(s,u)$ defines a one-parameter family $(c(s,\cdot))_{s\in\T_\ell}$ of $\Sig$-geodesics.
	The vector field $J_{s_0}(u_1):=\partial_{s}c(s_0,u_1)$ constitutes a Jacobi 
	field (along $c(s_0,\cdot)$), and it satisfies the Jacobi equation \cite{Spivakvol4}*{Chapter~8}
	\[
		\frac{\mathrm{D}_{\Sig}^2}{(\d u_1)^2}J+R_{\Sig}(J,\bS_{\Sig})\bS_{\Sig}=0.
	\]
	Note that $J_{s_0}(0)$ is equal to $\bT(s)$ and since
	the Levi-Civita connection is torsion-free, we obtain
	\begin{align*}
		\frac{\mathrm{D}_{\Sig}}{\d u_1}J_{s_0}(0)
		&=\left.\Big(\frac{\mathrm{D}_{\Sig}}{\d u_1} \frac{\partial}{\partial s} c\Big)\right|_{(s,u_1)=(s_0,0)}
 		=\left.\Big(\frac{\mathrm{D}_{\Sig}}{\d s} \frac{\partial}{\partial u_1} c\Big)\right|_{(s,u_1)=(s_0,0)}\\
		&=\left.\Big(\frac{\mathrm{D}_{\Sig}}{\d s}\Big)\right|_{s=s_0} \bS(\gamma(s))
		=-\kappa_g(s_0)\bT(s_0).
	\end{align*}

	Thus we obtain $g_3(J_s(u_1),\bS_{\Sig}(c(s,u_1)))=0$, and the other scalar product
	$y_T(s,u_1):=g_3(J_s(u_1),\bT_{\Sig}(c(s,u_1)))$ satisfies \eqref{eq:system_ODE}. 
	This implies the formulas for $(F_\gamma)_*(\partial_s)(F_\gamma(s,\bu))$.
\end{proof}

By writing $(\bT_{\Sig},\bS_{\Sig})$ in terms of $(F_\gamma)_*(\partial_s),(F_\gamma)_*(\partial_{u_1})$,
we obtain the following proposition.
\begin{proposition}[Free Dirac operator]\label{prop:expr_free_dirac}
	If we write $X_{s}$, $X_{u_1}$, $X_{u_2}$ for the push forward 
	$(F_\gamma)_*(\partial_s),(F_\gamma)_*(\partial_{u_1})$ and 
	$(F_\gamma)_*(\partial_{u_2})$ respectively, we have the representation
	\begin{align}
		\bsigma(-i\nabla)
		&=-i\bsigma(\bT_{\Sig}^{\flat})X_s-i\bsigma(\bS_{\Sig}^{\flat})X_{u_1}-i\bsigma(\bN_{\Sig}^{\flat})X_{u_2}\nn\\
		&\quad+L_{\err}^{(s)}X_s+L_{\err}^{(u_1)}X_{u_1},
	\end{align}
	where
	\begin{equation}\label{eq:def_L_err}
	\left\{
		\begin{array}{l}
		L_{\err}^{(s)}=\frac{-i}{y_T\Cg}\bigg[\Big(\cos(u_2)-y_T\Cg+b_1\sin(u_2)\Big) \bsigma(\bT_{\Sig}^{\flat})
		-a_1\sin(u_2)\bsigma(\bS_{\Sig}^{\flat})\bigg],\\
		L_{\err}^{(u_1)}=\frac{-i}{\Cg}\bigg[-b_0\sin(u_2)\bsigma(\bT_{\Sig}^{\flat})+\Big(\cos(u_2)
		-\Cg+a_0\sin(u_2)\Big)\bsigma(\bS_{\Sig}^{\flat})\bigg].
		\end{array}
	\right.
	\end{equation}
	The functions $a_0,a_1,b_0,b_1$ are defined in \eqref{eq:def_a_i_b_j}, and $\Cg$ in \eqref{eq:def_Cg}.
\end{proposition}

For each fixed $u_2$ (with $|u_2|$ small enough), the function $\phi_{(u_2)}:=F_{\phi}(\cdot, u_2)$ defines a local chart
of the sheaf $\Sig_{u_2}$. Thanks to Proposition~\ref{prop:push_forward}, we obtain the following corollary.


\begin{corollary}[Computation of the volume form]\label{coro:volume_form}
	Let $\phi$ be a local chart of $\Sig$, $F_{\phi}$ and $\phi_{(u_2)}$ the 
	associated local charts of $B_{\eps}[\Sig]$ and $\Sig_{u_2}$. 
	Then we have
	\begin{align*}
		F_{\phi}^*\vol_{g_3}
		&=\big(F_{\phi_{(u_2)}}^* \vol_{\Sig_{u_2}}\big)\wedge \d u_2\\
		&= C_\phi(v,w,u_2)\big(F_{\phi}^* \vol_{\Sig}\big)\wedge \d u_2,
	\end{align*}
	with
	\begin{equation*}
		C_\phi(v,w,u_2)=\cos(u_2)^2-\sin(2u_2)H_{\Sig}(\phi(v,w))
		+\sin(u_2)^2K_{\Sig}^{\ext}(\phi(v,w)) .
	\end{equation*}
\end{corollary}

Now we apply this result to $F_\gamma$. 
As above, we define $C_\gamma(s,u_1,u_2)$ by the following formula:
\begin{equation}\label{eq:def_Cg}
	\Cg(s,u_1,u_2):=\cos(u_2)^2-\sin(2u_2)H_{\Sig}(c(s,u_1))
	+\sin(u_2)^2K_{\Sig}^{\ext}(c(s,u_1)) .
\end{equation}
From Proposition~\ref{prop:push_forward}, it is clear that we have
\[
	y_T(s,u_1)\Cg(s,u_1,u_2)\d s\wedge \d u_1=\bT_{\Sig}^{\flat}\wedge \bS_{\Sig}^{\flat},
\]
where $s,u_1$ and $u_2$ are seen as coordinates, and that
\begin{align*}
	\vol_{g_3}&=\bT_{\Sig}^{\flat}\wedge \bS_{\Sig}^{\flat}\wedge\bN_{\Sig}^{\flat}
	=y_T(s,u_1)\Cg(s,u_1,u_2)\d s\wedge \d u_1\wedge \d u_2.
\end{align*}

\subsection{Singular two-dimensional Dirac operators and model case}
\ \newline
\noindent Proposition~\ref{prop:formule_smooth_mag} leads us to study the elliptic operators 
$$
	\cD_{\bA_{f}^{\R^2}}:=\sigma\cdot (-i\nabla^{\R^2}+\bA_{f}^{\R^2})
$$ 
on $L^2(\R^2)^2$, where $f\in\sD([-\delta,\delta]^2,[0,+\infty))$ is a probability density and
$$
	\bA_{f}^{\R^2}(\bu)=2\pi\alpha I_f(\bu)\begin{pmatrix} 0\\ 1\end{pmatrix}\in\R^2
$$
with associated magnetic field $\bB_{f}^{\R^2}(\bu)=2\pi\alpha f(\bu)$.

Because it will be important for our convergence results, we quickly 
address the different choices for the magnetic gauge potential.

We will use the identification $\R^2\simeq \C$, and
$u_1+iu_2=z=r e^{i\theta}$, and will write indifferently $\C$ or $\R^2$. 
We also emphasize that in this paper $\R_+$ denotes $[0,+\infty)$.

\subsubsection{Gauge choice}
Particular to two dimensions, we can also choose a \emph{scalar} gauge, which
is defined as follows. 
The scalar $h_{f}$ is the function
$$
	h_f:=\frac{\log|\cdot|}{2\pi}*\bB_{f}^{\R^2},
$$
and $\bal_{f,h} :=\overline{\alpha}_f\d z+\alpha_f\d\overline{z}$ is the magnetic gauge
potential
$$
	\alpha_f:=i\overline{\partial_z}h_f=\frac{i}{2}\big(\partial_{u_1}+i\partial_{u_2}\big)h_f.
$$
(To keep things simple, we identify vectors and one-forms in $\R^2$.)

On the other hand, as on $\S^3$, we may use singular magnetic gauge potentials
to study our operators.
Let $\log$ be the extension of the logarithm on $\C\setminus \R_+$,
cut along $\theta=0\!\!\mod 2\pi$. 
For $\theta_0\in\R/(2\pi\Z)$, let $L_{\theta_0}$ be the ray
\begin{equation}\label{eq:ray}
	L_{\theta_0}:=\R_+\begin{pmatrix} \cos(\theta_0)\\ \sin(\theta_0)\end{pmatrix}
\end{equation}
(with \emph{positive} orientation). We can write
\begin{equation*}
	\bA_{f}^{\R^2}=2\pi\alpha \dint_{\bu\in \R^2}f(\bu)[\bu+L_{0}]\d\bu,
\end{equation*}
where $[\bu+L_{0}]$ denotes the one-current defined by the oriented half-line $\bu+L_0$.
By applying it to smooth vector fields, we immediately see that it coincides with the smooth vector
(as a vector-valued distribution):
\begin{equation}\label{eq:decomp_A_f}
	\bA_{f}^{\R^2}(\bu)=2\pi\alpha\dint_{-\infty}^{u_1}f(u,u_2)\d u\begin{pmatrix} 0\\ 1\end{pmatrix}.
\end{equation}

The two gauges are related as follows. By direct computation the scalar gauge for the point-like magnetic field
$2\pi\alpha\delta_0$ is $\alpha\nabla^{\R^2}\theta$. By decomposing $\bB_{f}^{\R^2}$ as an integral over Dirac points:
\[
 \bB_{f}^{\R^2}=\int_{\bu}2\pi\alpha f(\bu)[\delta_{\bu}]\d \bu,
\]
we get by linearity
\begin{equation}\label{eq:pre_comput}
	\bA_{f,h}(\bv_0)=\alpha\dint_{\bu} f(\bu)\nabla_{\bv}^{\R^2} \theta_{\bu}(\bv_0)\d\bu,
\end{equation}
where  $\theta_{\bu}\in[0,2\pi)$ is the angle defined by $\theta_{\bu}(\bv):=\theta(\bv-\bu)$:
$$
	e^{i\theta_{\bu}(\bv)}:=\frac{\bv-\bu}{|\bv-\bu|}.
$$
The gradient of $\theta_{\bu}$ is:
\begin{align*}
	\nabla^{\R^2}\theta_{\bu}(\bv)
	=\nabla_{\bv}^{\R^2}\theta(\bv-\bu)=\frac{1}{(v_1-u_1)^2
	+(v_2-u_2)^2}\begin{pmatrix}-(v_2-u_2)\\ v_1-u_1 \end{pmatrix}.
\end{align*}
A change of variables in the integral in \eqref{eq:pre_comput} yields
\begin{align}\label{eq:comput_scalar_gauge}
	\bA_{f,h}(\bv)&=\alpha\dint_{\bu}f(\bv-\bu)\nabla^{\R^2}\theta(\bu)\d\bu,
\end{align}
and the function $\bA_{f,h}(\bv)$ is smooth by the theorem of differentiation under the 
integral sign. Hence we obtain the following proposition.
\begin{proposition}\label{prop:gauge_transf_formula}
	The two magnetic gauge potentials $\bA_{f}^{\R^2}$, $\bA_{f,h}$ 
	in \eqref{eq:decomp_A_f} and \eqref{eq:comput_scalar_gauge} are related via
	\[
		\bA_{f,h}(\bv)-\bA_{f}^{\R^2}(\bv)
		=\alpha\nabla^{\R^2}\Big(\dint_{\bu}\d\bu f(\bu)\theta_{\bu}(\bv)\Big)=:\alpha\nabla^{\R^2} \zeta_f(\bv),
	\]
	where $\zeta_f(\bv)$ is the smooth function
	\begin{equation}\label{eq:sing_gauge}
		\zeta_f(\bv)=\int_{\bu}f(\bv-\bu)\theta(\bu)\d\bu.
	\end{equation}
	Furthermore, for a smooth and compactly supported mollifier:
	 \[\phi_{\delta}(\bu)=\delta^{-2}\phi_1(\bu/\delta),\ \phi_1\in\sD(\R^2,[0,+\infty))\ \&\ \delta>0,
	 \]
	the gauge transformation $\exp(i\alpha\zeta_{\phi_\delta})$ converges to $\exp(i\alpha\theta)$ (cut along $\R_+$),
	pointwise on $\C\setminus \R_+$ and uniformly in $\C\setminus B_\eps[\R_+]$ for any $\eps>0$. 
	But
	\begin{equation}\label{eq:discrepancy_gauge}
		\liminf_{\delta\to 0^+}\norm{\exp(i\alpha\theta)-\exp(i\zeta_{\phi_\delta})}_{L^\infty}>0.
	\end{equation}
\end{proposition}
The discrepancy~\eqref{eq:discrepancy_gauge} is easy to see, since $\exp(i\zeta_{\phi_\delta})$ is continuous while $\exp(i\alpha\theta)$
is not. We recall that the angle function $\theta$ in the proposition is cut along $\theta=0$.

\subsubsection{Behavior under scaling of the density}
	Pick a fixed probability density $\phi$ in $\sD(\R^2)$, and consider the associated mollifier $(\phi_\delta)_{\delta>0}$
	given by the formula:
	$
	\phi_\delta(\bu):=\delta^{-2}\phi\Big(\frac{\bu}{\delta}\Big). 
	$
	Using the integral representation~\eqref{eq:decomp_A_f}, we get:
	\begin{equation}\label{eq:scal_prop}
	\bA_{\phi_\delta}^{\R^2}(\bu)=\frac{1}{\delta}\bA_{\phi}^{\R^2}\Big(\frac{\bu}{\delta}\Big).
	\end{equation}
	The scalar gauge shares the same scaling property: by \eqref{eq:comput_scalar_gauge}, we have
	\[
	\bA_{\phi_\delta,h}(\bu)=\frac{1}{\delta}\bA_{\phi,h}\Big(\frac{\bu}{\delta}\Big).
	\]
	It follows from the equality $\nabla^{\R^2}\theta\big(\frac{\bu}{\delta}\big)=\delta\nabla^{\R^2}\theta(\bu)$.
\subsubsection{Local gauge transformation}\label{sec:local_gauge_away}
Consider the representation \eqref{eq:decomp_A_f} of $\bA_f^{\R^2}$. It has the same form as in Eq~\eqref{eq:formule_phase_Phi}. 
Thus for $u_1>\delta$, that is on the right-hand side
of $\supp\,f$, we have:
\begin{equation}\label{eq:local_gauge_away}
\bA_f(\bu)=2\pi\alpha\nabla^{\R^2}M_f(u_2),\ u_1>\delta,
\end{equation}
where $M_f$ is the (partially) integrated marginal~\eqref{eq:def_M_f}. 
We can gauge away the magnetic potential on this region
with the function $e^{i2\pi\alpha M_f(u_2)}$:
\[
\bsigma\cdot(-i\nabla^{\R^2})e^{iM_f(u_2)}\psi=e^{iM_f(u_2)}\bsigma\cdot(-i\nabla^{\R^2}+2\pi\alpha\nabla^{\R^2}M_f)\psi.
\]

\subsubsection{Convergence for a point-like field \& Tamura's result}
The point-like magnetic field $2\pi\alpha \delta_0$ 
will play a particular role in our analysis, and we quickly recapitulate
some results.

Recall $z = re^{i\theta} \in \C$; by a straight-forward computation one can verify that
\begin{equation*}
	\alpha\nabla^{\R^2}\theta=2\pi\alpha\dint_0^{2\pi}[L_{\theta_0}]\frac{\d\theta_0}{2\pi},
\end{equation*}
where $L_{\theta_0}$ is the ray in the $\theta_0$-direction \eqref{eq:ray}. We recall that it is the 
scalar gauge for the point-like magnetic field in the origin $2\pi\alpha\delta_0$. 
Let $\sigma\cdot(-i\nabla^{\R^2}+\alpha\nabla\theta)$ be a differential operator defined on $\sD(\C\setminus\{0\})^2$. 
We are interested in the self-adjoint  extension\footnote{We refer to \cite{Persson_dirac_2d} 
for the complete discussion on all self-adjoint extensions of this operator.} $\cD_{0,\alpha,h}^{(-)}$ defined through
$$
	\left\{
	\begin{array}{lcl}
		\dom\big(\cD_{0,\alpha,h}^{(-)}\big) &:=& \mathrm{clos}_{\cG}(\sD(\C\setminus\{0\})^2)
		\overset{\perp_{\mathcal{G}}}{\oplus}\C\begin{pmatrix}0\\ K_\alpha(r) \end{pmatrix}\\
		\cD_{0,\alpha,h}^{(-)} \psi &:=& \big[\sigma\cdot(-i\nabla^{\R^2}+\alpha\nabla\theta)\psi\big]_{|_{\C\setminus\{0\}}}, 
		\quad \psi \in \dom\big(\cD_{0,\alpha,h}^{(-)}\big).
	\end{array}
	\right.
$$
It is a natural question whether the operator in the radial gauge is the limit of a smooth elliptic operator.
Let $(\phi_{\delta})_{\delta >0}$ be a mollifier, and denote by
$$
	\cD_{\bA_{\phi_\delta,h}} := \sigma\cdot(-i\nabla^{\R^2}+\bA_{\phi_\delta,h}).
$$
The following theorem is the main result of \cite{Tamura}.

\begin{theorem}[Norm resolvent convergence]\label{thm:tamura}
	Let $(\phi_{\delta})_{\delta >0}$ be a mollifier, then
	$$
		\lim_{\delta \to 0^{+}} \big(\cD_{\bA_{\phi_\delta,h}} + i\big)^{-1}
		= \big(\cD_{0,\alpha,h}^{(-)} + i\big)^{-1} \textrm{ in } \mathcal{B}\big(L^2(\R^2)^2\big),\ 0<\alpha<1.
	$$
\end{theorem}

\begin{rem}
	Observe that there is no norm resolvent convergence for the case $\alpha = 1$. The reason is that an 
         ``approximated singular line" disappears at finite energy as $\delta\to 0^{+}$.
	 Indeed: let $\chi_1\in \sD(\R^2,[0,1])$ be any cut-off function with $\chi_1(\bu)=1$ for $|\bu|\le 2^{-1}$.
	Then consider the wave function 
	\begin{equation}\label{eq:collapsed_quasimode}
		\psi_{\delta}(\bu):=\chi_1(\bu)\begin{pmatrix}0 \\ \exp(-h_{\phi_\delta}(\bu)) \end{pmatrix}\in H^1(\R^2)^2.
	\end{equation}
	Firstly, $\psi_\delta(\bu)$ converges to $|\bu|^{-1}\chi_1$, so the normalized function $\norm{\psi_\delta}_{L^2}^{-1}\psi_\delta$
	concentrates around $0$ and tends weakly to $0$ in $L^2$ (with $\norm{\psi_\delta}_{L^2}$ of order $\sqrt{-\log(\delta)}$). Then we have
	\[
		\frac{1}{\norm{\psi_\delta}_{L^2}^2}\int |\cD_{\bA_{\phi_\delta,h}}\psi_\delta|^2
		=\frac{1}{\norm{\psi_\delta}_{L^2}^2}\int e^{-2h_{\phi_\delta}}|\sigma\cdot \nabla^{\R^2}\chi_1|^2\underset{\delta \to 0}{\longrightarrow}0.
	\]
	If we had norm-resolvent convergence, the limit would be the operator $D:=e^{-i\theta}\sigma\cdot(-i\nabla^{\R^2})e^{i\theta}$, which is unitarily equivalent to
	the free Dirac operator. Let us show that this convergence cannot hold.

 	If we pick any non-negative cut-off function $\chi_2$ (with $\chi_1\chi_2=\chi_1$ to simplify), then the following operator is compact:
 	\[
 	 B:=(D+i)^{-1}\chi_2= e^{-i\theta}\big(\sigma\cdot(-i\nabla^{\R^2})+i\big)^{-1}\chi_2e^{i\theta}.
 	\]
	Along a sequence $\delta_n\to 0$, Formula \eqref{eq:collapsed_quasimode} 
	provides us -- up to normalization -- with an element $\psi_n$ in the domain of $D_n=\cD_{\bA_{\phi_{\delta_n},h}}$ such that
	$\psi_n=\chi_2\psi_n$ concentrates around $0$ with
	\[
	 \int |\psi_n|^2=1, \int |D_n\psi_n|^2\to 0\quad\&\quad  \psi_n\rightharpoonup 0.
	\]
	Since $\chi_2\psi_n=\psi_n$, we have $\norm{\big((D_n+i)^{-1}\chi_2+i\big)\psi_n}_{L^2}\to 0$.
	The norm-resolvent convergence would give $\norm{\big((D+i)^{-1}\chi_2+i\big)\psi_n}_{L^2}\to 0$, 
	which would contradict the inequality 
	\[\norm{B}_{\mathcal{B}(L^2(\C)^2)}<1.
	\]
	This inequality can be proved in the following way. Remark that we have the equalities
	$\norm{B}_{\mathcal{B}(L^2(\C)^2)}^2=\norm{B^*B}_{\mathcal{B}(L^2(\C)^2)}$ and
	$B^*B = \chi_2 (D^2+1)^{-1}\chi_2$.
	The operator $\chi_2 (D^2+1)^{-1}\chi_2$ is non-negative and compact, hence its highest eigenvalue is its norm.
	Let $\psi\in L^2(\C)^2$ be the corresponding normalized eigenfunction, the following holds:
	\[
	 \cip{\psi}{\chi_2 (D^2+1)^{-1}\chi_2\psi}_{L^2}<\norm{\chi_2\psi}_{L^2}^2\le 1.
	\]
\end{rem}

\begin{rem}[norm-resolvent convergence and gap topology]
Tamura's main theorem states a convergence in the norm-resolvent sense. By parts b. and c. of \cite{Kato}*{Theorem~IV.2.23}, 
it coincides with the convergence in the gap topology. Indeed, this theorem implies the following result.
\end{rem}
\begin{theorem}\label{thm:norm_resv_conv}
	Let $\cH$ be a Hilbert space, let $(\cD_n)$ be a sequence of self-adjoint operators on 
	$\cH$ and let $\cD$ be a self-adjoint operator. The following two propositions are equivalent.
	\begin{enumerate}
	\item The sequence $(\cD_n)$ converges to $\cD$ in the norm-resolvent sense.
	\item The projections $P_{n}\in \cB(\cH\times \cH)$ onto the graphs of the $\cD_n$'s converge to
	the projection $P$ onto the graph of $\cD$ in $\cB(\cH\times \cH)$.
	\end{enumerate}
\end{theorem}

The equivalent result for the strong-resolvent convergence is
\cite{dirac_s3_paper1}*{Lemma~28}.

\begin{lemma}\label{lem:cond_equiv_strong}
Under the same assumption of Theorem~\ref{thm:norm_resv_conv}, we have the equivalence of the three propositions.
\begin{enumerate}
	\item The sequence $(\cD_n)$ converges to $\cD$ in the strong-resolvent sense
	\item The projections $P_{n}\in \cB(\cH\times \cH)$ onto the graphs of the $\cD_n$'s converge to
	the projection $P$ onto the graph of $\cD$ in the strong operator topology.
	\item For all $\Psi=(\psi,D\psi)\in \cG(D)$, there exists a sequence $(\Psi_n)$
	with $\Psi_n=(\psi_n,D_n\psi_n)\in \cG(D_n)$ such that $\Psi_n\to \Psi$ in $\cH^2$.
\end{enumerate}
\end{lemma}

That the strong convergence of the projections imply the third characterization is trivial, the converse 
is easy to show once one notices that for a self-adjoint operator $D$ with graph-projection $P_{0}$, the orthogonal
complement of its graph is:
\[
	\cG(D)^{\perp}=\ker P_{0}=\big\{(-D\psi,\psi),\quad \psi\in\dom(D)\big\}.
\]
The third proposition of Lemma \eqref{lem:cond_equiv_strong} implies $\slim\,P_n P_0=P_0$ and 
$\slim (1-P_n)(1-P_0)=1-P_0$, where $P_n$ is the orthogonal projection onto $\cG(D_n)$. 
We deduce that
\[
	\slim\,P_n=\slim\big(P_n P_0-(1-P_n)(1-P_0)+(1-P_0)\big)=P_0.
\]

\subsubsection{Convergence in the singular gauge}
Now we consider the operator in the singular gauge; we set
$$
	\cD_{0,\alpha}^{(-)} := e^{i\alpha\theta} \cD_{0,\alpha,h}^{(-)}e^{-i\alpha\theta}
$$
with domain
\begin{equation}\label{eq:decomp_dom}
	\dom\big(\cD_{0,\alpha}^{(-)}\big)=\dom\big(\cD_{0,\alpha}^{(\min)}\big)
	\overset{\perp_{\cG}}{\oplus}\C\begin{pmatrix}0\\ K_\alpha(r)e^{i\alpha\theta} \end{pmatrix},
\end{equation}
where
$$
	\dom\big(\cD_{0,\alpha}^{(\min)}\big)
	=\big\{ \psi\in H^1(\R^2\setminus L_0)^2: \left.\psi\right|_{(L_0)_+}
	=e^{-2i\pi\alpha}\left.\psi\right|_{(L_0)_-}\big\}.
$$
One may again ask in which sense 
$$
	\cD_{\bA_{\phi_\delta}^{\R^2}} \to \cD_{0,\alpha}^{(-)}, \quad (\delta \to 0^{+}).
$$
To answer this question, we can use Theorem~\ref{thm:tamura} and 
study the gauge transformation between 
$\cD_{0,\alpha}^{(-)}$ and $\cD_{0,\alpha,h}^{(-)}$. 
Indeed, one readily shows that
\[
	\sigma\cdot(-i\nabla^{\R^2}+A)e^{i\phi}\psi=e^{i\phi}\sigma\cdot(-i\nabla^{\R^2}+A+\nabla^{\R^2}\phi)\psi.
\]
So if we write $P_{\alpha}$, $P_{\alpha,h}$, $P_{\delta}$ and $P_{\delta,h}$ the projections onto the graph norm
of $\cD_{0,\alpha}^{(-)},\cD_{0,\alpha,h}^{(-)}$, $\cD_{\bA_{\phi_\delta}^{\R^2}}$ and $\cD_{\bA_{\phi_\delta,h}^{\R^2}}$
respectively, we have:
\begin{equation}\label{eq:def_projection}
	P_\alpha=e^{i\alpha\theta}P_h e^{-i\alpha\theta}\quad\&\quad 
		P_{\delta}=e^{i\zeta_{\phi_\delta}}P_{\delta,h}e^{-i\zeta_{\phi_\delta}}.
\end{equation}
	Since $\lim_{\delta\to 0}\norm{P_h-P_{\delta,h}}=0$ by Theorem~\ref{thm:tamura}, we have
	the equivalence
	\begin{equation*}\label{eq:equiv_norm_res_conv}
		\lim_{\delta\to 0}\norm{P_\alpha-P_{\delta}}=0\quad \iff\quad \lim_{\delta\to 0}\norm{[e^{i(\alpha\theta-\zeta_{\phi_{\delta}})}, P_h]}=0.
	\end{equation*}
	Observe the following equalities:
	\begin{equation*}
		[e^{i(\alpha\theta-\zeta_{\phi_{\delta}})}, P_h]=[(e^{i(\alpha\theta-\zeta_{\phi_{\delta}})}-1), P_h],
	\end{equation*}
	and:
	$$
		\norm{[e^{i(\alpha\theta-\zeta_{\phi_{\delta}})}, P_h]}=\norm{e^{-i\zeta_{\phi_{\delta}}} P_\alpha e^{i\zeta_{\phi_{\delta}}}-P_{h}}.
	$$
	\begin{rem}
	Because of \eqref{eq:discrepancy_gauge} we have $\liminf_{\delta\to 0}\norm{(e^{i(\alpha\theta-\zeta_{\phi_{\delta}})}-1)P_h}>0$,
	but this does not disprove norm-resolvent convergence of $\cD_{\bA_{\phi_\delta}^{\R^2}}$ to $\cD_{0,\alpha}^{(-)}$ as $\delta\to 0$.
	In fact we also expect the norm-resolvent convergence in that case, but we do not study this question in this paper.
	\end{rem}

Now we can state two theorems, whose proofs are given at the end of this section, 
concerning the (norm) resolvent convergence of the operator $\cD_{\bA_{\phi_{\delta}}^{\R^2}}$.

\begin{theorem}\label{thm:conv_2D_dirac}
	For a mollifier $(\phi_{\delta})_{\delta>0}$,
	we have
	$$
		\slim_{\delta \to 0^{+}}\big(\cD_{\bA_{\phi_{\delta}}^{\R^2}} + i\big)^{-1}
		= \big(\cD_{0,\alpha}^{(-)} + i\big)^{-1}.
	$$
\end{theorem}

Thus we have strong resolvent convergence.

\begin{theorem}[Analysis of a graph-norm bounded sequence]\label{thm:defect}
	Let $\delta_0>0$, and for $0<\delta\le \delta_0$, let $\cD_{\delta}^{\R^2}:=\cD_{\bA_{\phi_{\delta}}^{\R^2}}$. 
	Let
	\[
	\big[\Psi_{\delta}=(\psi_\delta,\cD_{\delta}^{\R^2}\psi_{\delta})\big]_{0<\delta<\delta_0}
	\]
	be a normalized family of elements in the graphs $\cG(\cD_{\delta}^{\R^2})$.
	
	Then, along any decreasing sequence $(\delta_n)_{n\in\mathbb{N}}$, $\delta_n\to 0$, and up to an extraction
	of a subsequence the following holds.
	\begin{enumerate}
		\item The sequence $(\psi_{\delta_n})_{n\in\mathbb{N}}$
	converges in $L^2_{\loc}(\R^2)^2$ to an element $\psi_0\in\dom(\cD_{0,\alpha}^{(-)})$.
		\item The sequence $(\Psi_{\delta_n})_{n\in\mathbb{N}}$ $L^2$-weakly converges to
		\[
			(\psi_0,\cD_{0,\alpha}^{(-)}\psi_0).
		\]
	\end{enumerate}
\end{theorem}

\subsubsection{Description of the domain $\dom(\cD_{0,\alpha}^{(-)})$}
We recall that we have the decomposition \eqref{eq:decomp_dom}. Moreover, 
by the Green Theorem we easily deduce that
\[
	\norm{\cD_{0,\alpha}^{(\min)}\psi_0}_{L^2(\C)^2}^2=\norm{\nabla^{\R^2}\psi_0}_{L^2(\C\setminus\R_+)^2}^2,
	\quad \forall \psi_0\in \dom(\cD_{0,\alpha}^{(\min)}). 
\]
Then, we decompose $e^{-i\alpha\theta}\dom(\cD_{0,\alpha}^{(-)})$ with respect to the polar decomposition
\[
	L^2(\R^2)^2=L^2(\R_+,r\d r)^2\otimes L^2(\S^1,\d\theta).
\]
We get
\[
	\dom(\cD_{0,\alpha}^{(\min)})\subset \bigoplus_{m\in\Z} e^{i(m+\alpha)\theta}H_0^1((0,\infty),r\d r)^2.
\]
That is, we decompose $\psi_0\in \dom\big(\cD_{0,\alpha}^{(\min)}\big)$ in modes
$\psi_0=\sum_{m\in\Z}e^{i(m+\alpha\theta)}\phi_m(r)$, so that
\begin{multline}\label{decomp_graph_norm}
	\norm{\psi_0}^2_{L^2(\C)^2}+\norm{\cD_{0,\alpha}^{(\min)}\psi_0}_{L^2(\C)^2}^2\\
	\quad=\sum_{m\in\Z}\int_0^{\infty}\Big(|\phi_m(r)|^2+|(\partial_r\phi_m)(r)|^2
	+\frac{(m+\alpha)^2}{r^2}|\phi_m(r)|^2 \Big)\,r\d r.
\end{multline}

\subsubsection{Proof of Theorems~\ref{thm:conv_2D_dirac}~\&~\ref{thm:defect}}

	\ \medskip
	\paragraph{\textit{Proof of Theorem~\ref{thm:conv_2D_dirac}}}
	We recall that the strong resolvent convergence of a family of self-adjoint operator
	is equivalent to the convergence in the strong operator topology of the orthogonal
	projections onto the corresponding graphs.
	
	Recall: $P_{\alpha}$, $P_{\alpha,h}$, $P_{\delta}$ and $P_{\delta,h}$ 
	denote the projections onto the graphs of $\cD_{0,\alpha}^{(-)},\cD_{0,\alpha,h}^{(-)}$, 
	$\cD_{\bA_{\phi_\delta}^{\R^2}}$ and $\cD_{\bA_{\phi_\delta,h}^{\R^2}}$ respectively (and these projections
	are linked by \eqref{eq:def_projection}).
	We have to show that for all $\Psi\in L^2(\S^3)^2\times L^2(\S^3)^2$
	we have:
	\[
		P_{\delta}\Psi\underset{\delta\to 0}{\longrightarrow}P_{\alpha}\Psi.
	\]
	By Theorem~\ref{thm:tamura}, we have $\norm{P_h-P_{\delta,h}}_{\mathcal{B}}\to 0$.
	By uniform convergence of $\zeta_{\phi_\delta}$ on the complement set
	$\complement_{\C}B_{\eps}(\R_+)$ ($\eps>0$) of any strip along $\R_+$ 
	and dominated convergence, we obtain: 
	\[
		\norm{\big[e^{-i\zeta_{\phi_\delta}}-e^{-i\alpha\theta}\big]\Psi}_{L^2}\underset{\delta\to}{\to}0.
	\]
	Hence the following convergence holds in $L^2$:
	\begin{equation*}
		P_{\delta,h}e^{-i\zeta_{\phi_\delta}}\Psi-P_{h}e^{-i\alpha\theta}\Psi
		=\big[P_{\delta,h}-P_h\big]e^{-i\zeta_{\phi_\delta}}\Psi+P_h\big[e^{-i\zeta_{\phi_\delta}}-e^{-i\alpha\theta}\big]\Psi\to 0.
	\end{equation*}
	By dominated convergence, we get:
	\begin{multline*}
		P_{\delta}\Psi=e^{i\zeta_{\phi_\delta}}P_{\delta,h}e^{-i\zeta_{\phi_\delta}}\Psi
				=e^{i\zeta_{\phi_\delta}}\big[P_{\delta,h}e^{-i\zeta_{\phi_\delta}}-P_{h}e^{-i\alpha\theta}\big]\Psi
				+e^{i\zeta_{\phi_\delta}}P_he^{-i\alpha\theta}\Psi\\
				\to e^{i\alpha\theta}P_he^{-i\alpha\theta}\Psi=P_\alpha\Psi.
	\end{multline*}

	\medskip
	
	\paragraph{\textit{Proof of Theorem~\ref{thm:defect}}}
	
	\ \medskip
	\subparagraph{\textit{A priori results}}
	
	For $0<\delta<\delta_0$, let $\Psi_{\sing,\delta}\in\cG(\cD_{\delta}^{\R^2})$ 
	be the ``singular" element:
	\begin{equation*}
		\Psi_{\sing,\delta}:=(\psi_{\sing,\delta},\cD_{\delta}^{\R^2}\psi_{\sing,\delta})\in H^1(\R^2)^2,
	\end{equation*}
	that is the closest normalized vector of $\cG(\cD_{\bA_{\phi_{\delta}}^{\R^2}})$ to the (normalized)
	element $\Psi_{\sing}$ of the graph $\cG(\cD_{0,\alpha}^{(-)})$:
	\[
	\Psi_{\sing}:=\big(\norm{K_\alpha(|\bu|)}_{L^2}^2+\norm{K_{1-\alpha}(|\bu|)}_{L^2}^2\big)^{-1/2}
		\bigg( \begin{pmatrix}0\\ K_\alpha(r)\end{pmatrix},i\begin{pmatrix} K_{1-\alpha}e^{-i\theta}\\ 0\end{pmatrix}\bigg).
	\]

	By Theorem~\ref{thm:tamura}, we know that
	$e^{-i\zeta_{\phi_\delta}}\Psi_{\sing,\delta}\in\cG(\cD_{\bA_{\phi_\delta},h})$ converges
	to the singular element $e^{-i\alpha\theta}\Psi_{\sing}$ of $\cG(\cD_{0,\alpha,h})$.
	By pointwise convergence of $e^{i\zeta_{\phi_\delta}}$ (to $e^{i\alpha\theta}$) and dominated convergence,
	we get that $\Psi_{\sing,\delta}$ converges to $\Psi_{\sing}$,
	hence
	\[
		\norm{(1-P_\alpha)\Psi_{\sing,\delta}}\underset{\delta\to}{\to} 0.
	\]
	We recall that $P_\alpha$ denotes the orthogonal projection onto the graph of $\cD_{0,\alpha}^{(-)}$. 
	Up to extraction, the orthogonal projection of $\Psi_{\delta_n}$ onto $\C\Psi_{\sing,\delta_n}$ will converge strongly.
	So, up to removing this projection, we can assume that $\Psi_{\delta_n}$
	is orthogonal to that line. 
	
	Up to extracting a subsequence, we can assume that $(\Psi_{\delta_n})_{n\in\mathbb{N}}$ converges weakly in $L^2$
	to some $\Psi_0=(\psi_0, \wt{\psi}_0)$. Our aim is to show that $\Psi_0$ is in the graph of $\cD_{0,\alpha}^{(-)}$,
	and that the convergence of $(\psi_{\delta_n})_{n\in\mathbb{N}}$ holds in $L^2_{\loc}(\C)^2$ up to extracting a subsequence.
	
	We will have to use the singular line for $\cD_{0,\alpha}^{(+)}$:
	\begin{multline}\label{eq:sing_+}
		\wt{\Psi}_{\sing}:=\big(\psi_{\sing,+},i\psi_{\sing,-}\big)\\
				= \big(\norm{K_\alpha(r)}_{L^2}^2+\norm{K_{1-\alpha}}_{L^2}^2\big)^{-1/2}
		\bigg(\begin{pmatrix} K_{1-\alpha}(r)e^{-i\theta}\\ 0\end{pmatrix}, i\begin{pmatrix}0\\ K_\alpha(r)\end{pmatrix}\bigg),
	\end{multline}

	We first show that the convergence holds in $L^2_{\loc}(\C\setminus\{0\})^2$. The idea is to use the form of the magnetic
	potential (see Section~\ref{sec:local_gauge_away}).
	
	\medskip
	\subparagraph{\textit{Localization}}
	We localize the function $\psi_{\delta}$ in three different region:
	\begin{enumerate}
		\item	close to $\{0\}$,
		\item away from $\{0\}$ but close to $\R_+$,
		\item	away from $\R_+$.
	\end{enumerate}
	
	To do so, we pick a radial function $\chi\in\sD(\R^2,[0,1])$ for which
	$\chi(\bu)=0$ for $|\bu|\le 1$ and $\chi(\bu)=0$ for $|\bu|\ge 2$.
	And we introduce a smooth function $g\in C^{\infty}(\R^2,[0,1])$ which localizes around $\R_+$
	$$
		\left\{
		\begin{array}{lr}
			g(\bu)=1, &\mathrm{dist}(\bu,\R_+)\le \frac{1}{2},\\
			g(\bu)=0, &\mathrm{dist}(\bu,\R_+)>1.
		\end{array}
		\right.
	$$
	
	Let $r_0>0$, we consider $\chi_{r_0}(\bu)=\chi(\bu/r_0)$
	 and $g_{r_0}(\bu):=g(\tfrac{\bu}{r_0})$.
	The localization at level $r_0$ is:
	\begin{align*}
	\psi_\delta&=\big(\chi_{r_0}g_{r_0}+(1-\chi_{r_0})g_{r_0}+(1-g_{r_0})\big)\psi_{\delta},\\
			&=\psi_{\delta}^{(0)}+\psi_{\delta}^{(\R_+)}+\psi_{\delta}^{(\mathrm{Rem})}.
	\end{align*}
	
	\subparagraph{\textit{Convergence in $L^2_{\loc}(\C\setminus\{0\})^2$}}
	Along a sequence $\delta_n\to 0$, the countable family $(\psi_{\delta_n}^{(\mathrm{Rem})})_n$
	is $H^1$-bounded, and by Section~\ref{sec:local_gauge_away}, so is the family
	\[
	\exp(-iM_{\phi_{\delta_n}}(u_2))\psi_{\delta_n}^{(\R_+)},\ 
	M_{\phi_{\delta_n}}(u_2)=2\pi\alpha\int_{-\infty}^{+\infty}\d v_1\int_{-\infty}^{u_2}\d v_2\phi_{\delta_n}(v_1,v_2).
	\]
	Thus up to extracting a subsequence, they both converge in $L^2_{\mathrm{loc}}(\C)^2$,
	and by dominated convergence, so does the family $(\psi_{\delta_n}^{(\R_+)})_n$.
	By a standard diagonal argument, this shows that the family $(\psi_{\delta_n})_{n\in\mathbb{N}}$
	converges in $L^2_{\loc}(\C\setminus\{0\})^2$ up to the extraction of a subsequence.
	The limit is necessarily $\psi_0$. 
	
	\medskip
	\subparagraph{\textit{Proof of \textit{$\Psi_0\in\cG(\cD_{0,\alpha}^{(-)})$}}}

	Due to the convergence of the smooth phase,
	$\psi_0$ has the phase jump $e^{-2i\pi\alpha}$ across $\R_+$  
	and in $\sD'(\C\setminus[0,+\infty))^2$ we get
	\[
	\cD_{\delta_n}^{\R^2}\psi_{\delta_n}\rightharpoonup \bsigma\cdot(-i\nabla)\psi_{0}|_{\C\setminus L_0}.
	\]
	Thus we obtain $\psi_0\in\dom(\cD_{0,\alpha}^{(\max)})$. Let us prove that it is in $\cG(\cD_{0,\alpha}^{(-)})$.
	By strong resolvent continuity (third characterization of Lemma~\ref{lem:cond_equiv_strong}),
	any element $\Psi_-=(\psi_{-,\infty},\cD_{0,\alpha}^{(-)}\psi_{-,\infty})\in \cG(\cD_{0,\alpha}^{(-)})$ is the norm limit of a sequence $(\Psi_{-,n})_{n\in\mathbb{N}}$
	with $\Psi_{-,n}=(\psi_{-,n},\cD_{\delta_n}^{\R^2}\psi_{-,n})\in \cG(\cD_{\delta_n}^{\R^2})$. 
	Since the approximating operators are elliptic, their common domain is $H^1(\C)^2$. It follows that we have:
	\[
	\cip{\psi_{\delta_n}}{\cD_{\delta_n}^{\R^2}\psi_{-,n}}_{L^2}=\cip{\cD_{\delta_n}^{\R^2}\psi_{\delta_n}}{\psi_{-,n}}_{L^2}.
	\]
	By weak convergence of $\psi_{\delta_n}$ and $\cD_{\delta_n}^{\R^2}\psi_{\delta_n}$ and strong convergence of $\cD_{\delta_n}^{\R^2}\psi_{-,n}$
	and $\psi_{-,n}$, we obtain as $n\to+\infty$:
	\[
	\cip{\psi_{0}}{\cD_{0,\alpha}^{(-)}\psi_{-,\infty}}_{L^2}=\cip{\sigma\cdot(-i\nabla)\psi_{0}|_{\C\setminus L_0}}{\psi_{-,\infty}}_{L^2}.
	\]
	In other words, $\psi_0\in\dom((\cD_{0,\alpha}^{(-)})^*)=\dom(\cD_{0,\alpha}^{(-)})$.
	Let us show that we have $\Psi_0\in\dom(\cD_{0,\alpha}^{(\min)})$.
	By assumption, $\Psi_{\delta_n}$ is orthogonal
	to the singular line, thus $\Psi_0$ is orthogonal to $\Psi_{\sing}$ by weak convergence.
	Let us show that we have $\Psi_0\in\dom(\cD_{0,\alpha}^{(\min)})$.
	In the scalar gauge the function 
	\[
	\Psi_{\delta_n,h}:=\mathrm{exp}\big[-i\zeta_{\phi_\delta}\big]\Psi_{\delta_n}
	\]
	satisfies the following. Let $\wt{P}_{\sing,h}$ be the projection onto $\C e^{i\alpha\theta}\wt{\Psi}_{\sing}$,
	where $\wt{\Psi}_{\sing}\in\cG(\cD_{0,\alpha}^{(+)})$ is defined in \eqref{eq:sing_+}, 
	then we have by Theorem~\ref{thm:tamura}:
	\[
	\limsup_{n\to+\infty}\norm{\wt{P}_{\sing,h} \Psi_{\delta_n,h}}\le \norm{\wt{P}_{\sing,h}P_\alpha}=0.
	\]
	By pointwise convergence of $\mathrm{exp}\big[i\zeta_{\phi_\delta}\big]$, the function
	\[
	\mathrm{exp}\big[-i\zeta_{\phi_\delta}\big]e^{i\alpha\theta}\wt{\Psi}_{\sing}
	\]
	converges strongly in $L^2$ which implies that $\Psi_0$ is orthogonal to $e^{i\alpha\theta}\wt{\Psi}_{\sing}$,
	hence $\Psi_0\in\dom(\cD_{0,\alpha}^{(\min)})$.
	By assumption $\Psi_{\delta}$ does not converge in norm, so there has to be some loss of mass: 
	$\norm{|\Psi_0}_{L^2}<\limsup_{n\to \infty}\norm{|\Psi_{\delta_n}}_{L^2}$.
	Besides the possible loss of mass at infinity,
	the mass of $\psi_{\delta_n}$ can only concentrate around $0$, that is 
	$
		\lim_{\eps\to 0}\limsup_{n\to+\infty}\int_{D_{\C}(0,\eps)}|\psi_{\delta_n}|^2>0
	$
	can occur.
	\medskip

	\subparagraph{\textit{Convergence in $L^2_{\loc}(\C)^2$}}
	We have to show that there is no mass of $(\psi_{\delta_n})_{n\in\mathbb{N}}$ concentrating around $0$. 
	The $L^2$-norm is insensitive to gauge transformation. So if there is concentration of mass around $0$, then
	the same holds for the sequence $(e^{-i\zeta_{\phi_{\delta_n}}}\psi_{\delta_n})$. This cannot happen due to the norm-resolvent
	convergence of $\cD_{\bA_{\phi_{\delta_n}},h}$ to $\cD_{0,\alpha,h}^{(-)}$ (by Theorem~\ref{thm:tamura}).
 	We can proceed as follows. Let $\chi_1\in\sD(\C,[0,1])$ with $\chi(x)=1$ for $|x|\le 1$
	and let $w_n:=\chi_1 e^{-i\zeta_{\phi_{\delta_n}}}\psi_{\delta_n}$, we show that $w_n$ converges in norm. 
	Introducing $\chi_2(x):=\chi_1(\tfrac{x}{2})$, we have $\chi_2w_n=w_n$ and for all $n$, the following
	operator is compact:
	\begin{equation*}
	    K_n:=\chi_2 (\cD_{\bA_{\phi_{\delta_n}},h}+i)^{-1}=\chi_2(\sigma\cdot(-i\nabla^{\R^2})+i)^{-1}\big[1+\sigma\cdot\bA_{\phi_{\delta_n},h}(\cD_{\bA_{\phi_{\delta_n}},h}+i)^{-1}\big].
	\end{equation*}
	Norm-resolvent convergence implies: $K_n\underset{n\to+\infty}{\longrightarrow}K_0:=\chi_2(\cD_{0,\alpha,h}^{(-)}+i)^{-1}$ in operator norm. In particular $K_0$ is compact.
	We have weak-$L^2$ convergence of $w_n$ and $\cD_{\bA_{\phi_{\delta_n}},h} w_n$ to 
	$e^{-i\alpha\theta}\chi_1\psi_0$ and $e^{-i\alpha\theta}\cD_{0,\alpha}^{(-)}(\chi_1\psi_0)$ respectively. 
	Putting everything together, it follows that
	\[
	 w_n=\chi_2w_n=K_n (\cD_{\bA_{\phi_{\delta_n}},h} w_n)+iK_nw_n
	\]
	converges strongly in $L^2(\C)^2.$

\subsubsection{Extension to $\cD_{\T_{\ell_k},\alpha_k}^{(-)}$ on $L^2(\T_{\ell_k}\times\R^2)^2$}
Let us consider one knot $\gamma_k$ with length $\ell_k>0$ and a flux $0<\alpha<1$. 
In \cite{dirac_s3_paper1}, we introduced the operator $\cD_{\T_{\ell_k},\alpha_k}^{(-)}$ as a model operator, which can be seen
as acting on $L^2$-spinors over the normal bundle of $\gamma_k$ in $\S^3$ (with a flat metric). Formally it corresponds to the operator
$-i\tfrac{\d}{\d s}\otimes\sigma_3+\mathrm{id}\otimes\cD_{0,\alpha}^{(-)}$ acting on $L^2(\T_{\ell_k})\otimes L^2(\R^2)^2$ identified with $L^2(\T_{\ell_k}\times \R^2)^2$.

Here and below, $\norm{\cdot}_{\T_{\ell_k}}$ denotes the graph norm of $\cD_{\T_{\ell_k},\alpha}^{(-)}$.
We refer the reader to \cite{dirac_s3_paper1}*{Section 3.2.2} for more details, including the proof of the following Lemma.

\begin{lemma}\label{lem:smart}
	the operator $-i\tfrac{\d}{\d s}\otimes\sigma_3+\mathrm{id}\otimes\cD_{0,\alpha}^{(-)}=\cD_{\T_{\ell_k},\alpha}^{(-)}$ is self-adjoint on its natural domain in 
	$L^2(\T_{\ell_k})\otimes L^2(\R^2)^2$,
	and for all $\psi\in \dom\big( \cD_{\T_{\ell_k},\alpha}^{(-)}\big)$ we have
	\begin{equation}\label{eq:egalite}
		\dint\Big| \cD_{\T_{\ell_k},\alpha}^{(-)}\psi \Big|^2=\dint |\partial_s \psi|^2+\dint |\cD_{0,\alpha}^{(-)}\psi|^2.
	\end{equation}
\end{lemma}

By using a Hilbert basis corresponding to the decomposition of the domain $\dom(\cD_{\T_{\ell_k},\alpha}^{(-)})$ 
into $L^2(\T_{\ell_k})$-Fourier modes with values $\dom(\cD_{0,\alpha}^{(-)})$, 
and using the convergence of $(\cD_{\bA_{\phi_\delta}^{\R^2}})_{\delta}$ to $\cD_{0,\alpha}^{(-)}$, 
we obtain the following result, which extends Theorems~\ref{thm:conv_2D_dirac}~and~\ref{thm:defect}.
	
\begin{proposition}\label{prop:strg_res_pour_DT}
	The family of elliptic operators
	of 
	$$
		\Big(-i\sigma_3\partial_s+\cD_{\bA_{\phi_\delta}^{\R^2}}\Big)_{\delta>0}
	$$
	on $L^2(\T_{\ell_k}\times\R^2)^2$ converges to $\cD_{\T_{\ell_k},\alpha}^{(-)}$ as $\delta\to 0^+$ 
	in the strong resolvent sense. 
	Furthermore the statement of Theorem~\ref{thm:defect} also applies to $\cD_{\T_{\ell_k}}^{(-)}$ and its approximations
	$\big(-i\sigma_3\partial_s+\cD_{\bA_{\phi_\delta}^{\R^2}}\big)_{\delta>0}$.
 \end{proposition}

\section{Proof of Theorems~\ref{thm:conv_dirac}~\&~\ref{thm:cont_homot}}\label{sec:proof_thm_conv_dirac}

\subsection{Strategy to prove the continuity in the gap topology}\label{sec:strategy}

We will prove the convergence stepwise and show successively.
	\begin{enumerate}
		\item Convergence in the strong resolvent sense,
		$$
			\slim_{\delta \to 0^{+}}(\cD_{\bA_{\phi_\delta}}+i)^{-1}
			= (\cD_{\bA}^{(-)} + i)^{-1}.
		$$ 
		\item For any smooth function $\Lambda\in\sD(\R)$,
		$$
			\lim_{\delta \to 0^{+}}\Lambda(\cD_{\bA_{\phi_\delta}})
			= \Lambda(\cD_{\bA}^{(-)})
			\textrm{ in } \mathcal{B}\big(L^2(\S^3)^2\big).
		$$
	\end{enumerate}
	
As the bounded functions $x\in\R\mapsto (x\pm i)^{-1}$ tend to $0$ at infinity, 
this will give the norm-resolvent convergence by functional calculus. We recall Theorem~\ref{thm:norm_resv_conv}:
norm-resolvent convergence and convergence in the gap topology (that is convergence in operator norm
of the orthogonal projections onto the graphs) coincide.

\begin{rem}[Quick recall of the geometrical objects]\label{rem:quick_recall}
	We recall that for a knot with Seifert surface $S$ embedded 
	in a closed and oriented surface $\Sig$, the coordinates $(s,\bu)$ relative to 
	$S$ are defined in \eqref{eq:cart_coord}
	The number $|u_2|$ corresponds to the distance to $\Sig$, $|u_1|$ corresponds to the $\Sig$-distance to $\gamma$
	of the projection onto $\Sig$, which $\Sig$-projects itself to $\gamma(s)$. The sign of $u_2,u_1$ corresponds to the side
	of the points $\bp,\mathrm{proj}_{\Sig}\bp$ with respect to $\Sig$ and $\gamma$ respectively.
	
	We extend the Seifert frame $(\bT,\bS,\bN)$ on $B_\delta[\gamma]$ stepwise.
	 First $\bT_{\Sig},\bS_{\Sig}$ are the $\Sig$-parallel transportations of $\bT,\bS$ 
	 along $\Sig$-geodesics orthogonal to $\gamma$, and $\bN$ is the normal to $\Sig$. 
	 Off $\Sig$, the vectors $\bT,\bS,\bN$ are the parallel transportations of $\bT_{\Sig},\bS_{\Sig},\bN$
	 along $\S^3$-geodesics orthogonal to $\Sig$.
	 
	 The Seifert frame defines the normalized sections $(\eta_+,\eta_-)$ of the spinor bundle 
	 on $B_\delta[\gamma]$ by \eqref{eq:def_sections} (defined up to a common phase $e^{i\phi(s)}$).
	 Pointwise, $\eta_{\pm}$ are eigenfunctions of the matrix $\sigma\cdot\bT$ (with eigenvalue $\pm1$).
	 
	 The expression of the free Dirac operator in the coordinates $(s,\bu)$ is written in Proposition~\eqref{prop:expr_free_dirac}.
	 It involves a model operator and the ``error terms" $L_{\err}^{(s)}X_s+L_{\err}^{(u_1)}X_{u_1}$ \eqref{eq:def_L_err},
	 where $X_s,X_{u_1}$ represent the (pushforwards of) the vector fields $\tfrac{\partial}{\partial s}$ and $\tfrac{\partial}{\partial u_1}$.
	 The model operator corresponds to the free Dirac operator for the flat metric on $\T_{\ell_{\gamma}}\times\R^2$, and the matrices 
	 $L_{\err}^{(s)},L_{\err}^{(u_1)}$ show the discrepancy between the flat metric and the (pullback of the) metric on $\S^3$.
\end{rem}

\subsection{Localization and jump phase functions}\label{sec:techn_tools_loc_phase_fun}
We begin by introducing some technical tools to help in the analysis of $\dom(\cD_{\bA}^{(-)})$.
These utilities are in some sense similar to the ones used in \cite{dirac_s3_paper1}, yet they 
appear in a different guise since we work in the coordinates $(s,\bu)$, whereas in \cite{dirac_s3_paper1}
we use the chart:
\begin{equation}\label{eq:exp_chart}
(s,\rho,\theta)\in\T_{\ell}\times [0,\eps)\times\R/2\pi\Z\mapsto 
\mathrm{exp}_{\gamma(s)}^{\S^3}[\rho(\cos(\theta)\bS(s)+\sin(\theta)\bN(s))].
\end{equation}

\subsubsection{Localization functions}\label{loc_fun}
We pick $\delta_0>0$ small enough such that the tubular neighborhood 
$B_{2\delta_0}[\gamma]$ has $K$ connected components,
and such that on each $B_{\delta_0}[\gamma_k]$, the coordinates 
$(s,\bu)$ relative to $S$ are well defined (see Section \ref{sec:def_coord}). We recall that we have:
\[
	u_{1,k}(\bp)+iu_{2,k}(\bp)=:r_k(\bp)\exp(i\theta_k(\bp)), \quad \bp\in B_{\delta_0}[\gamma_k].
\]

For simplicity, we will drop the dependence in $k$ of the coordinates.

Next, we chose a cut-off function $\chi_1\in \sD(\R,[0,1])$ with
$\chi(x)=1$ for $|x|\le 2^{-1}$ and $\chi(x)=0$ for $|x|\ge 1$.
We define the localization function at level $0<\delta \le \delta_0$ by
\begin{equation}\label{eq:chi_loc_curve}
	\chi_{\delta,\gamma_k}: 
	\begin{array}{ccl}
		B_{\delta}[\gamma_k] &\longrightarrow& [0,+\infty)\\
		\bp 
		&\mapsto& \chi_1\left(r(\bp)\delta^{-1}\right).
	\end{array}
\end{equation}
Then we pick $0<\delta<\delta_0$ such that
$$
	y_T(s,u_1)C_{\gamma_k}(s,u_1,u_2)(\bp)\ge 3/4, \quad \forall \bp\in B_\delta[\gamma_k],
$$
and the pullback of the volume form $\vol_{g_3}$ becomes comparable 
to the canonical volume form $ds\wedge du_1\wedge du_2$ on $\T_k\times \D_{\delta}$.
The function $y_T(s,u_1)$ is defined in Proposition~\ref{prop:push_forward}, the function $C_{\gamma_k}(s,u_1,u_2)$
in Corollary~\ref{coro:volume_form}, and the volume form of $\S^3$ in the coordinates $(s,\bu)$ is
\[
y_T(s,u_1)C_{\gamma_k}(s,u_1,u_2)\d s\wedge\d u_1\wedge \d u_2. 
\]

Furthermore, set
$$
	\chi_{\delta,S_k}(\bp) 
	:= \chi\left(4\tfrac{u_2(\bp)}{\delta}\right)\left(1-\chi_{\delta,\gamma_{k}}(\bp)\right),
$$
which has support close to the Seifert surface $S_k$, but away from the knot $\gamma_k$.
Then the remainder is defined as
$$
	\chi_{\delta,R_k}(\bp) := 1-\chi_{\delta,S_k}(\bp) - \chi_{\delta,\gamma_k}(\bp),
$$
which constitutes a partition of unity subordinate to $S_k$.
The partition of unity for the entire link is given by
\begin{align}\label{def:part_unity}
	1 &= \prod_{k=1}^K\left(\chi_{\delta,\gamma_k}(\bp) + \chi_{\delta,S_k}(\bp)
	+\chi_{\delta,R_k}(\bp) \right)\nn\\
	&= \sum_{\underline{j} \in \{1,2,3\}^K}\chi_{\delta,\underline{j}}(\bp)
	=\sum_{k=1}^K\chi_{\delta,\gamma_k}(\bp)+\sum_{\underline{j} \in \{2,3\}^K}\chi_{\delta,\underline{j}}(\bp),
\end{align}
where $\chi_{\delta,\underline{j}}$ is the product $\prod_{k=1}^K\chi_{\delta,X_k}$, with 
$$
	X_k = \left\{\begin{array}{lr}
	\gamma_k, &j_k = 1,\\
	S_{k}, &j_k = 2,\\
	R_k, &j_k = 3.
	\end{array}\right.
$$

\subsubsection{Jump phase functions}
Let $\psi\in\dom(\cD_{\bA}^{(-)})$. We study the phase jumps of 
$\chi_{\delta,\underline{j}}\psi$ for any tuple $\underline{j}$.

\subsubsection*{Away from the knots:}
	This corresponds to the functions coming from the last term in \eqref{def:part_unity}.
	Fix $\underline{j} \in \{2,3\}^K$, and write $\{1,\dots,K\} = J_2 \cup J_3$ where $j_k=2, k \in J_2$ and 
	$j_k=3, k \in J_3$.
	If $J_2$ is empty, then the function $(\chi_{\delta,\underline{j}}\psi)_n$ has no phase jump and is in $H^1(\S^3)^2$.
	
	Assume now that $J_2$ contains at least one element. 
	For each $1 \leq k \leq K$, the Seifert surface $S_k$ splits $\supp \chi_{\delta,S_{k}}$ into two 
	disconnected regions $O_{k,\pm}^{(n)}$, which correspond to the regions above and below the Seifert surface.
	We remove all the phase-jumps with the help of the function
	\begin{equation}\label{def:n_phase_remove_smooth_bulk}
		E_{\delta,\underline{j}}:=\prod_{k \in J_2}\exp\left[-2\pi i \alpha_k\mathds{1}_{O_{k,+}}\right].
	\end{equation}
	The function $\overline{E}_{\delta,\underline{j}}\chi_{\delta,\underline{j}}\psi$ does not contain any phase jump and
	is again in $H^1(\S^3)^2$.

\subsubsection*{Close to the knots:}
Let us say that $\gamma_k$ intersects the $S_{k'}$ at different points.
We study the corresponding phase jumps one Seifert surface at a time.
Consider an intersecting Seifert surface $S_{k'}$ which intersects $\gamma_k$ at
$0\le s_{0,k'}<s_{1,k'}<\ldots<s_{M_{k,k'}-1,k'}<\ell_k$, inducing a phase jump $e^{ib_{m,k'}}$
across $s_{m,k'}$ (with $b_{m,k'}=\pm2\pi\alpha_{k'}$).
We have:
\[
	\sum_{m=0}^{M_{k,k'}-1}b_{m,k'}=-2\pi\alpha_{k'}\link(\gamma_{k'},\gamma_k).
\]
Now we define for all $k'$ a phase function $E_{k,k'}(\bp)$ that contains the phase jump due to the intersecting $S_{k'}$
in $B_{\eps}[\gamma_k]$. Below the coordinate $s(\bp)$ corresponds to the chart $F_{\gamma_k}$.
Two cases can occur:

\subparagraph{\textit{Case 1: $\gamma_{k'}\cap \gamma_k=\{s_{0,k'}\}$}}
	Then we define
	\begin{equation}\label{Eq:E_k_k'_1_intersection}
		E_{k,k'}(\bp):=\exp\Big(-i \frac{b_{0,k'}}{\ell_k}(s(\bp)-s_{0,k'})\Big),\ s_{0,k'}\le s\le s_{0,k'}+\ell_k,
	\end{equation}
	the function $E_{k,k'}$ has the correct phase jump across $s(\bp)=s_{0,k'}\mod \ell_k$.
\subparagraph{\textit{Case 2: $\big|S_{k'}\cap \gamma_k\big|=M_{k,k'}\ge 2$}} 
	Then we have
	\[
		B_{\eps}[\gamma_k] \cap \complement S_{k'}
		=: R_{01}^{(k')} \cup R_{12}^{(k')} \cup \dots \cup R_{(M_{k,k'}-1)M_{k,k'}}^{(k')},
	\]
	which represents the decomposition of the tubular neighborhood around $\gamma_k$ cut into $M_{k,k'}$ sections.
	When passing from $R_{(m-1)m}^{(k')}$ to $R_{m(m+1)}^{(k')}$, we are going through $S_{k'}$
	which induces the phase jump $e^{ib_{m,k'}}$.
	Writing 
	$$
		R_{0m}^{(k')}:=R_{01}^{(k')} \cup R_{12}^{(k')} \cup \dots \cup R_{(m-1)m}^{(k')},
	$$
	we set
	\begin{equation}\label{Eq:E_k_k'_many_intersections}
		E_{k,k'}(\bp):=\exp\Big(-i \sum_{m=1}^{M_{k,k'}}\frac{b_{m,k'}}{\ell_k}s(\bp)
		+i \sum_{m=1}^{M_{k,k'} - 1}b_{m,k'} \mathds{1}_{R_{0m}^{(k')}}(\bp)\Big),\ 0\le s\le \ell_k.
	\end{equation}
		
Accompanying these decompositions, for $\bp \in B_{\eps}[\gamma] \cap \Omega_{\uS}$ 
we define	the function
\begin{equation}\label{def:curve_phasej_rem}
	E_k(\bp) := \prod_{k'\neq k}E_{k,k'}(\bp),
\end{equation}
with the convention that $E_{k,k'}\equiv 1$ when $\gamma_k\cap S_{k'}=\emptyset$ 
We recall that $\Omega_{\uS}$ denotes:
\[
\Omega_{\uS}:=\S^3\setminus(\cup_{1\le k\le K} S_k).
\]

\begin{proposition}\label{prop:phase_jump_removal}
	Let $\bA_{(k)}=2\pi\alpha_k [S_k]$, and let $\chi_{\delta,\gamma_k}$, $E_{k}$ 
	be defined as in \eqref{eq:chi_loc_curve} and \eqref{def:curve_phasej_rem}.
	If $\psi$ is in $\dom(\cD_{\bA}^{(-)})$ resp. in $\dom(\cD_{\bA}^{(\min)})$,
	then $\overline{E}_{k}\chi_{\delta,\gamma_k}\psi$ is in $\dom(\cD_{\bA_{(k)}}^{(-)})$ 
	resp. in $\dom(\cD_{\bA_{(k)}}^{(\min)})$.
\end{proposition}
		
Furthermore we have these important results, where $\norm{\cdot}_{\T_{\ell_k}}$ and
$\norm{\cdot}_{\bA_{(k)}}$ denote the graph norms of $\cD_{\T_{\ell_k},\alpha}^{(-)}$ and 
$\cD_{\bA_{(k)}}^{(-)}$ respectively.
\begin{proposition}\label{prop_link_with_model_case}
	With the same notation as in Proposition~\ref{prop:phase_jump_removal}, 
	for $\psi\in \dom(\cD_{\bA_{(k)}}^{(-)})$ let $f=(f_+,f_-)^T$
	be the spinor in $L^2(\T_{\ell_k}\times\R^2)^2$ given by
	\[
		\overline{E}_k\chi_{\delta,\gamma_k}\psi(\bp)
		=f_+((s,\bu)(\bp))\eta_+(\bp)+f_-((s,\bu)(\bp))\eta_-(\bp).
	\]
	We have $f\in \dom(\cD_{\T_{\ell_k},\alpha}^{(-)})$ and there 
	exists a constant $C_0=C_0(\delta,\gamma_k)>0$ such that
	\[
		\norm{f}_{\T_{\ell_k}}\le C_0(\delta,\gamma_k)\norm{\psi}_{\bA_{(k)}}.
	\]
	Furthermore, for $0<\delta\le \delta_0(\gamma_k)$ small enough
	there exists $0<\eps_0(\delta,\gamma_k)<1$ and $C_1(\delta,\gamma_k)>0$ such that
	\[
		\norm{f}_{\T_{\ell_k}}\ge (1-\eps_0(\delta,\gamma_k))\norm{\chi_{\delta,\gamma_k}\psi}_{\bA_{(k)}}
		-C_1(\delta,\gamma_k)\norm{\chi_{\delta,\gamma_k}\psi}_{L^2}.
	\]
\end{proposition}
\begin{proof}[Sketch of the proof]
	We do not prove it in full details: the method is the same as \cite{dirac_s3_paper1}*{Lemma~27},
	which was written for the coordinates \eqref{eq:exp_chart}. It is relatively easy to deduce the above results 
	taking the following things into account:
	\begin{enumerate}
		\item The connection form of $\nabla$ with respect to $(\eta_+,\eta_-)$ 
		is uniformly bounded in the tubular neighborhood $B_{2\delta}[\gamma_k]$, and so is
	 	the derivatives of $E_k$ (away from the jumps set).
		
		\item The level $\delta>0$ has been chosen small enough such that
		on $\T_{\ell_k}\times\D_{\delta}$ the volume forms 
		$F_{\gamma_k}^*\vol_{g_3}$ and $ds\wedge du_1\wedge du_2$ are comparable.
		
		\item On $B_{\delta}[\gamma_k]$, the error matrices $L_{\err}^{(s)}$ and $L_{\err}^{(u_1)}$ \eqref{eq:def_L_err} 
		can be written as $u_2\wt{L}_{\err}^{(\star)}$, $\star\in\{s,u_1\}$, where both the matrices $\wt{L}_{\err}^{(\star)}$ are 
		bounded by a constant $C=C(\delta,\gamma_k)$.
		
		\item For $\psi$ in $\dom(\cD_{\bA}^{(\max)})$, we have 
		$u_2\chi_{\delta,\gamma_k}\psi\in\dom(\cD_{\bA}^{(\min)})$, and the energy estimate
		\begin{align*}
			\norm{\cD_{\bA}^{(-)}(u_2\chi_{\delta,\gamma_k}\psi)}_{L^2}\le \norm{\chi_{\delta,\gamma_k}\psi}_{L^2}
			+\delta\norm{\cD_{\bA}^{(-)}(\chi_{\delta,\gamma_k}\psi)}_{L^2}.
		\end{align*}
		We use the fact that for elements in $\dom(\cD_{\bA}^{(\min)})$, Lichnerowicz' formula (for quadratic forms) holds
		(see \cite{dirac_s3_paper1}*{Proposition~5}):
		\begin{equation}\label{eq:Lich_for_min}
		\forall\,\psi\in\dom(\cD_{\bA}^{(\min)}),\ \int |\cD_{\bA}^{(\min)}\psi|^2=\int|(\nabla\psi)|_{|_{\Omega_{\uS}}}|^2+\frac{3}{2}\int |\psi|^2,
		\end{equation}
		to control the $L^2$-norm of the term 
		$
		\big(L_{\err}^{(s)}X_s+L_{\err}^{(u_1)}X_{u_1}\big)\chi_{\delta,\gamma_k}\psi.
		$
	\end{enumerate}
\end{proof}

\subsection{Proof of Theorem~\ref{thm:conv_dirac}, Part 1: Convergence in the strong resolvent sense}
To keep the notational burden at a minimum, we abbreviate
$$
	\cD_{\delta} := \cD_{\bA_{\phi_\delta}}.
$$

\subsubsection{Reduction to the approximation of $\cG(\cD_{\bA}^{(-)})$}
First observe that by Lemma~\ref{lem:cond_equiv_strong}, 
the strong resolvent convergence is equivalent to the following condition:
\begin{multline}\label{cond_equiv_strong}
	\forall\, (\psi,\cD_{\bA}^{(-)}\psi)\in \cG(\cD_{\bA}^{(-)})
	\,\exists\, (\psi_{\delta},\cD_{\delta}\psi_{\delta})\in \cG(\cD_{\delta}):\\
	\lim_{\delta \to 0^{+}}(\psi_{\delta},\cD_{\delta}\psi_{\delta}) =
	(\psi,\cD_{\bA}^{(-)}\psi) \textrm{ in } L^2\times L^2.
\end{multline}

\paragraph{\textit{Level of the localization}}  Now fix $\psi\in \dom(\cD_{\bA}^{(-)})$. We use the partition of unity described in Section~\ref{loc_fun}.
	We fix an appropriate level $0<\delta_1<\delta_0$ and localize at level $\delta_1$.
	Then we run a parameter $0<\delta<\delta_1$ for $\bal_{\phi_{\delta}}$ and study the limit $\delta\to 0^+$.

\subsubsection{Smooth phase jump functions}

\medskip

\paragraph{\textit{Local gauge transformations}}
	Here, we emphasize that the coordinate $u_2=u_{2,k}$ is defined relatively to $\Sig_{k}$,
	and is, up to the sign, the $\S^3$-distance to $\Sig_k$, the sign corresponding to the side of the point
	with respect to the oriented surface. Recall the form of the approximations of $2\pi\alpha_k[S_k]$ (Proposition~\ref{prop:formule_smooth_mag}):
	\begin{align*}
		\bal_{\phi_{\delta},k}(\bp)&=\Big[\int_{-\infty}^{u_1(\bp)}\int_{-\infty}^{u_2(\bp)}\phi_{\delta}(u')\d u'\Big]\d u_2(\bp),\\
							&=\Big[\int_{-\infty}^{u_1(\bp)/\delta}\int_{-\infty}^{u_2(\bp)/\delta}\phi_{1}(u')\d u'\Big]\d u_2(\bp),
	\end{align*}
	with the convention that $\bal_{\phi_{\delta},k}(\bp)=0$ on $\S^3\setminus B_{\sqrt{2}\delta}[S_k]$ ($\supp\,\phi_1\subset [-1,1]^2$). 
	For $\bp$ in the bulk of $B_\delta[S_k]$ (away from the knot $\gamma_k$), the one-form 
	$\bal_{\phi_{\delta},k}(\bp)$ takes the form
	\[
		\bal_{\phi_{\delta},k}(\bp)=2\pi\alpha_k m_{\phi_{\delta}}(u_2(\bp))\d u_2(\bp).
	\]
	We can (only locally) gauge it away with the help of the integrated marginal
	$$
		J_{\phi_\delta}(u_2):=\dint_{-\infty}^{u_2} m_{\phi_\delta}(u)\d u,
	$$
	since $\d J_{\phi_\delta}(u_2)=m_{\phi_{\delta}}(u_2)\d u_2$. 
	
	Observe that if $u_2<-\delta$, then we have $J_{\phi_\delta}(u_2)=0$, 
	whereas if $u_2>\delta$, then $J_{\phi_\delta}(u_2)=1$.
	Thus the function $\exp(-2i\pi\alpha_j J_{\phi_{\delta}}(u_2(\bp)))$ 
	spreads the phase jump across $S_k$ (away from its boundary)
	in the layer $-\delta <u_2(\bp)<\delta$.
	
\medskip	

\paragraph{\textit{Away from the knots}}
	This corresponds to the functions coming from the last term in \eqref{def:part_unity}.
	Fix $\underline{j} \in \{2,3\}^K$, and write $\{1,\dots,K\} = J_2 \cup J_3$ where $j_k=2, k \in J_2$ and 
	$j_k=3, k \in J_3$.
	For $J_2$ non-empty, we define the smooth function
	\begin{equation}\label{def:n_phase_remove_smooth_knot}
		E_{\delta_1,\delta,\underline{j}}(\bp):=\prod_{k \in J_2}\exp\left[-2\pi i \alpha_kJ_{\phi_\delta}(u_{2,k}(\bp))\right],
	\end{equation}
	which approximates the function $E_{\delta,\underline{j}}$ \eqref{def:n_phase_remove_smooth_bulk}.
	It is clear that away from the Seifert surfaces it converges pointwise (in the limit $\delta \to 0^{+}$),
	and that the convergence is uniform in any set of the form
	\[
		\big\{\bp\in\supp\,\chi_{\delta_1,\underline{j}}: |\bu(\bp)|\ge \delta_{2}\big\},
	\]
	where $0<\delta_2<\delta_1$ is given. Furthermore, writing 
	\[
		\wt{\psi}_{\delta_1,\delta,\underline{j}}
		:=E_{\delta_1,\delta,\underline{j}}\,\overline{E}_{\delta,\underline{j}}\, \chi_{\delta_1,\underline{j}}\psi,
	\] 
	we have $\wt{\psi}_{\delta_1,\delta,\underline{j}}\in H^1(\S^3)^2$, 
	\[
		\cD_\delta\wt{\psi}_{\delta_1,\delta,\underline{j}}
		=E_{\delta_1,\delta,\underline{j}}\, \overline{E}_{\delta,\underline{j}}\,
		\cD_{\bA}^{(-)}\chi_{\delta_1,\underline{j}}\psi.
	\]
	By dominated convergence we obtain
	\[
		\lim_{\delta\to 0^+}\big(\wt{\psi}_{\delta_1,\delta,\underline{j}},\cD_\delta\wt{\psi}_{\delta_1,\delta,\underline{j}} \big)
		=\big(\chi_{\delta_1,\underline{j}}\psi,\cD_{\bA}^{(-)}\chi_{\delta_1,\underline{j}}\psi \big)
		\quad\mathrm{in}\ L^2(\S^3)^2\times L^2(\S^3)^2.
	\]

\paragraph{\textit{Close to the knots: removal of the phase jumps}}
	We have to approximate the phase jump functions $E_{k}$ by smooth functions $E_{k,\delta}$.
	Due to transversality, by picking $\delta_1>0$ small enough we can assume that on 
	$B_{\delta_1}[\gamma_k]$, the Seifert surfaces $S_{k'}$
	are transverse to all the $(\gamma_k)_{\bu}$ for $|\bu|\le \delta_1$. 
	As $\delta\to 0^+$, the intersection of 
	$B_{\delta_1}[\gamma_k]$ with $B_{\delta}[\Sig_{k'}]$ has 
	$M_{k,k'}$ connected components, 
	each localized around one $\gamma_k(s_{m,k'})$. 
	
	On the component of the intersection close to $s_{m,k'}$, the surface $S_{k'}$
	cuts the component into two parts $T_{m,k'}^{\pm}$, inducing a phase jump $e^{ib_{m,k'}}$ across $S_{k'}$. 
	We write $\wt{E}_{k,k'}$ for the restriction of $E_{k,k'}$ to $B_{\delta_1}[\gamma_k]\cap B_{\delta}[\Sig_{k'}]$ (defined in \eqref{Eq:E_k_k'_many_intersections},
	and corresponds to the phase jumps due to $S_{k'}$ on $B_{\delta_1}[\gamma_k]$). 
	For $\bp\in \overline{T}_{m,k'}^{+}\cup \overline{T}_{m,k'}^{-}$ we define
	\[
	\wt{E}_{k,k'}(\bp):=
	\left\{
		\begin{array}{lr}
			E_{k,k'}(\bp)\exp(-ib_{m,k'}\mathds{1}_{T_{m,k'}^{+}}(\bp)), & \textrm{ if }b_{m,k'}=-2\pi\alpha_{k'},\\
			E_{k,k'}(\bp)\exp(ib_{m,k'}\mathds{1}_{T_{m,k'}^{-}}(\bp)), & \textrm{ if }b_{m,k'}=2\pi\alpha_{k'}.
		\end{array}
	\right.
	\]
	Here we used the convention that the tangent vector $\bT_k$ of $\gamma_k$ on the intersection $\gamma_k\cap S_{k'}$ 
	points toward $T_{m,k'}^+$. We define $E_{k,k',\delta}$ as follows.
	$$
		E_{k,k',\delta}(\bp):=
		\left\{\!\!
			\begin{array}{cl}
			\wt{E}_{k,k'}(\bp)\exp(-2i\pi\alpha_{k'}J_{\phi_{\delta}}(u_{2,k'}(\bp))),
			&\textrm{on }B_{\delta_1}[\gamma_k]\cap B_{\delta}[\Sig_{k'}]\\
			E_{k,k'}(\bp), &\textrm{on }B_{\delta_1}[\gamma_k]\cap  \big(B_{\delta}[\Sig_{k'}])^c
			\end{array}
		\right.
	$$
	and
	$$
		E_{k,\delta}(\bp):=\prod_{k'\neq k}E_{k,k',\delta}(\bp).
	$$
	This last function $E_{k,\delta}$ satisfies on $B_{\delta_1}[\gamma_k]$:
	\begin{equation}\label{eq:rid_of_others}
	\forall\,\bp\in B_{\delta_1}[\gamma_k],\ \overline{E}_{k,\delta}\d E_{k,\delta}(\bp)=\sum_{k'\neq k}\bal_{\phi_{\delta},k'}(\bp),
	\end{equation}
	where $\bal_{\phi_{\delta},k'}$ is the approximation of $2\pi\alpha_{k'}[S_{k'}]$ (Proposition~\ref{prop:formule_smooth_mag})
	(with the convention that $\bal_{\phi_{\delta},k'}(\bp)=0$ on $\S^3\setminus B_{\delta}[S_{k'}]$).
\subsubsection{Close to the knots: choice of the approximations}
Recall Proposition~\ref{prop_link_with_model_case}. Let
$f\in \dom(\cD_{\T_{\ell_k},\alpha}^{(-)})$ be the spinor associated with 
$\overline{E}_k \chi_{\delta_1,\gamma_k}\psi\in\dom(\cD_{2\pi[S_k]}^{(-)})$.
We use Proposition~\ref{prop:strg_res_pour_DT}: 
let $(f_{\delta})$ be the family 
in $H^1(\T_{\ell_k}\times\R^2)^2$ such that
\begin{equation}\label{eq:choice_f}
	\big(f_{\delta},\big(-i\sigma_3\partial_s +\cD_{\bA_{\phi_\delta}^{\R^2}}\big)f_\delta \big)
	\to \big(f,\cD_{\T_{\ell_k},\alpha}^{(-)}f \big).
\end{equation}
Up to localizing, we can assume that $\supp\,f_{\delta}\subset\{|\bu|\le 2\delta_1 \}$.
The approximation of $\chi_{\delta_1,\gamma_k}\psi$ is given by:
\[
	\wt{\psi}_{\delta_1,\delta,k}
	:=E_{k,\delta}((f_{\delta})_+ \eta_++(f_{\delta})_-\eta_-)\in H^1(\S^3)^2.
\]

Using the expression of the free Dirac operators in the $(s,\bu)$-coordinates (Proposition~\ref{prop:expr_free_dirac}), 
let us show that
\[
	\big(\wt{\psi}_{\delta_1,\delta,k},\cD_\delta\wt{\psi}_{\delta_1,\delta,k}\big)\to 
	\big(\chi_{\delta_1,\gamma_k}\psi,\cD_{\bA}^{(-)}\chi_{\delta_1,\gamma_k}\psi\big).
\]
By \eqref{eq:rid_of_others}, we have the exact cancellation on $B_{\delta_1}[\gamma_k]$:
\[
-i\overline{E}_{k,\delta}\sigma(\d E_{k,\delta})+\sigma(\bal_{\phi_{\delta}}-\bal_{\phi_{\delta},k})=0,
\]
hence we can assume without loss of generality that $\bal_{\phi_{\delta}}=\bal_{\phi_{\delta},k}$ (and $E_{k,\delta}=1$).
Locally around $\gamma_k$ we have
\[
	\begin{pmatrix}
	\cip{\eta_+}{\sigma(-i\nabla)\psi}\\ 
	\cip{\eta_-}{\sigma(-i\nabla)\psi}
	\end{pmatrix}
	= \Big(\underbrace{-i\sigma_1 \partial_{u_1}-i\sigma_{2} \partial_{u_2}-i\sigma_3 \partial_s}_{\cT} +\cE_0+\cE_1\Big)
	\begin{pmatrix} f_+\\ f_-\end{pmatrix}\big|_{\Omega_{\uS}},
\]
where $\cE_0$ is the matrix due to the connection form and 
$$
	\cE_1=L_{\err}^{(s)}\partial_s+L_{\err}^{(u_1)}\partial_{u_1}.
$$
The matrices $L_{\err}^{(s)},L_{\err}^{(u_1)}$ are defined in \eqref{eq:def_L_err}, 
where here $\sigma(\bT^{\flat})$ and $\sigma(\bS^{\flat})$ denote $\sigma_3$ and $\sigma_1$ respectively.
Let us write $\cT[\psi],\cE_0[\psi]$ and $\cE_1[\psi]$ for the corresponding elements in 
$L^2(B_{\delta_1}[\gamma_k])^2$ (we have made the abuse of notations $\cT[\psi]=\cT[\psi]_{|_{\Omega_{\uS}}}$ 
and $\cE_1[\psi]=\cE_1[\psi]_{|_{\Omega_{\uS}}}$).
From Proposition~\ref{prop:strg_res_pour_DT} we get that
\begin{equation}\label{conv_intermediate}
	\cT[E_{k,\delta}\wt{\psi}_{\delta_1,\delta,k}]
	+\sigma(\bal_{\phi_\delta})E_{k,\delta}\wt{\psi}_{\delta_1,\delta,k}
	\underset{\delta\to 0}{\longrightarrow}\cT[E_k\chi_{\delta_1,\gamma_k}\psi]
	\textrm{ in } L^2(\S^3)^2,
\end{equation}
and 
$$
	\lim_{\delta\to 0^+}\norm{\cE_0[E_{k,\delta}\wt{\psi}_{\delta_1,\delta,k}]
	-\cE_0[E_{k}\chi_{\delta_1,\gamma_k}\psi]}_{L^2}=0.
$$

We are left with the error term coming from $\cE_1$.
Observe that the graph element corresponding to
\[
	u_{2,k}\chi_{\delta_1,\gamma_k}\psi\in\dom(\cD_{\bA}^{(\min)})\cap H^1(\Omega_{\uS})^2
\]
is the limit of the one corresponding to $(u_{2,k}E_{k,\delta}\wt{\psi}_{\delta_1,\delta,k})_{0<\delta\le \delta_1}$.
The presence of $u_2$ enables us to use an $\eps$-argument with cut-off function localizing outside $S_k$.
For $0<\delta_2\le \delta_1$, we introduce $\xi_{k,\delta_2}$ to be:
\[
\xi_{k,\delta_2}(\bp):=\chi_{1}\big(\frac{|\bu_{2,k}(\bp)|}{\delta_2}\big).
\]
We split $\psi_{\delta,k}:=E_{k,\delta}\wt{\psi}_{\delta_1,\delta,k}$ into two with respect to $1=\xi_{k,\delta_2}+(1-\xi_{k,\delta_2})$.
From the formula for $\bal_{\phi_{\delta},k}$, we have $|\bal_{\phi_{\delta},k}(\bp)|\le 2\pi\alpha_k$.
For $\star\in\{s,u_1\}$, we obtain:
\begin{align*}
\norm{L_{\err}^{(\star)}X_{\star}(\xi_{k,\delta_2} \psi_{\delta,k})}_{L^2}&\le C(\delta_1,\gamma_k)\norm{X_{\star}(u_2\xi_{k,\delta_2} \psi_{\delta,k})}_{L^2},\\
											&\le C(\delta_1,\gamma_k)\norm{\nabla(u_2\xi_{k,\delta_2} \psi_{\delta,k})}_{L^2},\\
											&\le C(\delta_1,\gamma_k)\norm{\sigma(-i\nabla)(u_2\xi_{k,\delta_2} \psi_{\delta,k})}_{L^2},\\
											&\le C(\delta_1,\gamma_k)\big[\delta_2\norm{\xi_{k,\delta_2} \psi_{\delta,k}}_{L^2}\\
											&\quad\quad+\norm{\sigma(-i\nabla+\bal_{\phi_{\delta},k})(u_2\xi_{k,\delta_2} \psi_{\delta,k})}_{L^2}\big].
\end{align*}
The last upper bound is a $o_{\delta_2\to 0}(1)$. 
Similarly, using Lichnerowicz' formula for $\cD_{\bA}^{(\min)}$ (\cite{dirac_s3_paper1}*{Proposition~5}), we obtain:
\[
\norm{\cE_1[\xi_{k,\delta_2}\chi_{\delta_1,\gamma_k}\psi]_{|_{\Omega_{\uS}}}}_{L^2}\le C(\delta_1,\gamma_k)\norm{u_2\xi_{k,\delta_2} \chi_{\delta_1,\gamma_k}\psi}_{\cD_{\bA}^-}=o_{\delta_2\to 0}(1).
\]
Hence it suffices to prove that $\psi_{\delta,k}$ converges to $\chi_{\delta_1,\gamma_k}\psi$
in $H^1_{\loc}(\Omega_{\uS})^2$ to end the proof of the convergence in the strong resolvent sense. We work with the corresponding spinors
\[
g_{k,\delta}:=(1-\xi_{k,\delta_2})f_{k,\delta}\quad\mathrm{and}\quad g_k:=(1-\xi_{k,\delta_2})f,
\]
where we have dropped the dependence in $\delta_2$ to shorten notations.
Since the one-form $\bal_{\phi_{\delta},k}$ is supported on $B_{\sqrt{2}\delta}[S_k]$, for $\delta\ll \delta_2$, 
$-i\sigma_3\partial_s+\cD_{\bA_{\phi_{\delta}}^{\R^2}}$ and $\cD_{\T_{\ell_k},\alpha_k}^{(-)}$ act like the free Dirac
operator on $g_{k,\delta}$ and $g_k$ respectively (outside the phase jump surface for $g_k$), 
and $g_{k,\delta},g_k\in H^1(\complement_{\T_{\ell_k}\times\R^2}\{\theta=0\})^2$. Since we have 
\[
|\partial_{u_2}\xi_{k,\delta_2}|\le \delta_2^{-1}\norm{\chi_1'}_{L^\infty}<+\infty,
\] 
the same convergence as in \eqref{eq:choice_f} holds for $g_{k,\delta}$ and $g_k$. By Stokes formula on $\complement(B_{\delta_2/2}[\T_{\ell_k}]\cup\{\theta=0\})$ we get:
\[
\int_{\complement_{\T_{\ell_k}\times\R^2}\{\theta=0\}} |\sigma\cdot(-i\nabla)(g_{k,\delta}-g_k)|^2=\int_{\complement_{\T_{\ell_k}\times\R^2}\{\theta=0\}} |\nabla(g_{k,\delta}-g_k)|^2\underset{\delta\to 0}{\to} 0,
\]
where $\nabla$ denotes the flat connection on $\T_{\ell_k}\times\R^2$, hence $g_{k,\delta}$ converges to $g_k$ in $H^1(\{\theta\neq 0\})^2$. 
By Proposition~\ref{prop:push_forward} and Corollary~\ref{coro:volume_form}, we can transpose this result to $\S^3$ and $(1-\xi_{k,\delta_2})\psi_{\delta,k}$
converges to $(1-\xi_{k,\delta_2})\wt{\psi}_{\delta_1,\delta,k}$ in $H^1(\Omega_{\uS})^2$.

\subsection{Proof of Theorem~\ref{thm:conv_dirac}, Part 2: Convergence in the gap topology}
We prove that for any function $\Lambda\in\sD(\R)$, the family 
$(\Lambda(\cD_{\delta}))_\delta$ converges to $\Lambda(\cD_{\bA}^{(-)})$.
As said in Section~\ref{sec:strategy}, this implies the convergence in the gap topology by functional calculus.
We proceed by contradiction.

Recall that we have already proved strong resolvent convergence and 
that the operators under investigation have discrete spectrum. 
Thanks to that and using Lemma~\ref{lem:cond_equiv_strong}, we easily show the following equivalence:
\begin{enumerate}
\item $\limsup_{\delta\to 0}\norm{\Lambda(\cD_{\delta})-\Lambda(\cD_{\bA}^{(-)})}_{\mathcal{B}}>0$,
\item There exists a sequence of 
eigenvalues $(\lambda_{\delta_n})$ in $(\supp\Lambda)^{\circ}$ associated with
a sequence $(\delta_n)$ decreasing to $0$  together with a family $(\psi_n)$ of 
wave functions of unit norm such that:
\begin{equation}\label{eq:collapsing_sequence_quasimode}
	\left\{
	\begin{array}{l}
		\lambda_n\to\lambda\in (\supp\Lambda)^{\circ},\\
		\psi_n\in\ker(\cD_{\delta_n}-\lambda_n),\\
		\psi_n \rightharpoonup 0 \textrm{ in } L^2(\S^3)^2.
	\end{array}
	\right.
\end{equation}
\end{enumerate}
Indeed, by strong resolvent convergence (criterion~\ref{cond_equiv_strong}) and the discreteness of the spectrum of the operators,
for every eigenfunction $\psi$ of $\cD_{\bA}^{(-)}$, there exists an eigenfunction $\psi_n\in\dom(\cD_{\delta_n})$
such that $(\psi_n,\cD_{\delta_n}\psi_n)\to (\psi,\cD_{\bA}^{(-)}\psi)$ in $L^2(\S^3)^2\times L^2(\S^3)^2$. If 
$$
	\lim_n\norm{\Lambda(\cD_{\delta_n})-\Lambda(\cD_{\bA}^{(-)})}_{\cB}>0,
$$ 
then there exists an eigenplane $\C \psi_n\subset \ker(\cD_{\delta_n}-\lambda_n)$, $\norm{\psi_n}_{L^2}=1$,
which stays away from $\cG(\cD_{\bA}^{(-)})$ such that
$\lambda:=\lim_n\lambda_n\in (\supp\Lambda)^{\circ}$ and for every eigenfunction $\psi\in\dom(\cD_{\bA}^{(-)})$
with eigenvalue $\lambda\in\supp\,\Lambda$:
\begin{equation}\label{eq:cond_away_from_graph}
	\lim_{n\to+\infty}\cip{\psi_n}{\psi}_{L^2(\S^3)^2}=0.
\end{equation}

By $H^1$-boundedness away from the Seifert surfaces, the family $(\psi_n)$
converges (up to an extraction) in $L^2_{\loc}(\Omega_{\uS})^2$.

The next step is to deal with the part which is away from the knots but close to the Seifert surfaces. Fix $\delta>0$.
For $\underline{j}\in\{2,3\}^K$, consider the localization $\chi_{\delta,\underline{j}}\psi_n$,
which is still graph-norm bounded. We can gauge away the magnetic potential in that region of the sphere
with the help of the smooth phase function $E_{\delta,\delta_n,\underline{j}}$ (see \eqref{def:n_phase_remove_smooth_knot}):
the sequence with running element
\[
\overline{E_{\delta,\delta_n,\underline{j}}}\chi_{\delta,\underline{j}}\psi_n=:\chi_{\delta,\underline{j}}\wt{\psi}_n
\]
is $H^1(\S^3)^2$-bounded. Thus it converges in $L^2(\S^3)^2$ up to the extraction of a subsequence. 
By dominated convergence, this implies that $\chi_{\delta,\underline{j}}\psi_n$ converges in $L^2(\S^3)^2$
to an element in $H^1(\Omega_{\uS})^2$ with the correct phase jumps across the Seifert surfaces 
(we recall that the smeared out phase jump function $E_{\delta,\delta_n,\underline{j}}$ takes values in $\S^1$ and converges almost everywhere
on $\supp\,\chi_{\delta,\underline{j}}$).

By a standard diagonal argument with $\delta$ we obtain that 
the family $(\psi_n)$ converges (up to the extraction of a subsequence) in 
$L^2_{\loc}(\S^3\setminus\gamma)^2$ to an element in 
$$
	\psi_{\infty} \in \ker\big(\cD_{\bA}^{(\max)}-\lambda\big).
$$ 
Furthermore, the sequence $(\psi_n,\cD_{\delta_n}\psi_n)$ converges weakly 
in $L^2(\S^3)^2\times L^2(\S^3)^2$. By using strong resolvent convergence (criterion~\ref{cond_equiv_strong}) and the equality
\[
	\cip{\cD_{\delta_n}\psi_n}{\phi}_{L^2}
	=\cip{\psi_n}{\cD_{\delta_n}\phi}_{L^2},\quad\forall\,\phi\in H^1(\S^3)^2,
\]
we get that the limit $\psi_\infty$ is necessarily in $\dom(\cD_{\bA}^{(-)})$ and in 
$\ker(\cD_{\bA}^{(-)}-\lambda)$. By \eqref{eq:cond_away_from_graph}, we get $\psi_{\infty}\equiv 0$.
In other words, all the $L^2(\S^3)^2$-mass of $\psi_n$ concentrates around the link $\gamma$: we have $\norm{\psi_n}_{L^2}\equiv1$
and
\[
\lim_{\eps\to 0}\liminf_{n\to+\infty}\int_{B_\eps[\gamma]}|\psi_n|^2=1.
\]

Let us show that this leads to a contradiction. 
Up to localizing, we study $\chi_{\delta,\gamma_k}\psi_n$, where $k$ is chosen such that part of the $L^2(\S^3)^2$-mass of $\psi_n$
disappears in $\gamma_k$, that is such that 
\[
\lim_{\eps\to 0}\liminf_{n\to+\infty}\int_{B_\eps[\gamma_k]}|\psi_n|^2>0.
\]
Up to considering a subsequence, we assume that the liminf is a limit.
Furthermore, up to multiplying by the appropriate phase functions,
we study $\overline{E}_{k,\delta}\chi_{\delta_1,\gamma_k}\psi_n\in\dom(\cD_{2\pi\alpha_k[S_k]}^{(-)})$. 
Now we use Proposition~\ref{prop_link_with_model_case} and study the 
associated $f_n\in L^2(\T_{\ell_k}\times\R^2)^2$.
Recall Lemma~\ref{lem:smart}, and decompose $f_n$ into $s$-Fourier modes with values in $\dom\big(\cD_{\bA_{\phi_\delta}^{\R^2}}\big)$.

We can assume that the decomposition is orthogonal with respect to
\[
	\norm{f}_{\delta_n}^2
	:=\int |f|^2+\int \big|\sigma\cdot(-i\nabla_{\bu}+\bA_{\phi_{\delta_n}}(\bu))f\big|^2.
\]
By a straightforward computation one can show that the (squared) graph norm of $f_n$ equals
\[
	\int\Big( |f_n|^2+|\partial_s f_n|^2\Big)+\int \big|\sigma\cdot(-i\nabla_{\bu}
	+\bA_{\phi_{\delta_n}}(\bu))f_n\big|^2.
\]
More precisely, we decompose $f_n$ into modes of $\T_{\ell_k}$:
\[
f_n(s,\bu)=\frac{1}{\sqrt{\ell_k}}\sum_{j\in\T_{\ell_k}^*}e^{ijs}w_{n,j}(\bu),
\]
where we have $w_{n,j}\in \dom\big( \cD_{\bA_{\phi_{\delta_n}}^{\R^2}}\big)$.
Writing $\norm{\cdot}_{\delta_n,\R^2}$ for the graph norm of the $2$d-operator $\cD_{\bA_{\phi_{\delta_n}}}^{\R^2}$,
we have:
\[
	\norm{f_n}_{\delta_n}^2=\sum_{j\in\T_{\ell_k}^*}\norm{w_{n,j}}_{\delta_n,\R^2}^2.
\]
Let us consider one $w_{n,j}(\bu)$ of the decomposition: the corresponding sequence of the graph norms
is bounded. Furthermore, since $(f_n)$ converges weakly to $0$, so does $(w_{n,j})$, and
either $(w_{n,j})$ converges strongly to $0$ or there is a loss of mass by concentration on $0$ since $\psi_{\infty}=0$.

By assumption, there is at least one $w_{n,j}$ that loses mass at $0$ and we consider this sequence. 
Now we use Theorem~\ref{thm:defect}, and obtain a contradiction: the theorem states that a sequence $(g_n)_{n\ge 1}$
with $g_n\in\dom\big( \cD_{\bA_{\phi_{\delta_n}}^{\R^2}}\big)$ such that 
$(\norm{g_n}_{\delta_n,\R^2})_n$ is bounded converges in $L^2_{\loc}(\R^2)^2$, up to the extraction of a subsequence. 
In particular there cannot be any loss of mass at $0$.

\subsection{Proof of Theorem~\ref{thm:cont_homot}}

We consider the homotopy $(t,\delta)\in[0,1]\times [0,\delta_0]\mapsto \cD_{\delta}(t)$.
The restriction to $\delta>0$ is obviously a homotopy of elliptic operators, and more precisely of Dirac operators
with smooth magnetic potential $\bA_{\delta}(t)$, varying continuously with $(t,\delta)\in [0,1]\times (0,\delta_0]$.
Hence this part of the homotopy is obviously continuous in the norm-resolvent sense.

Now we deal with the singular part of the family ($\delta=0$), and then treat the whole homotopy.
The arguments for both cases are similar to the one used in the proof of Theorem~\ref{thm:conv_dirac},
and we leave some details to the reader. 

Since the parameter space $[0,\delta_0]\times [0,1)^K$ is metric, it suffices to check sequential continuity in the gap topology.

In \cite{dirac_s3_paper2} we proved that the family of singular Dirac operators for fluxes $\ua\in (0,1)^K$
is $\sT_{\Lambda}$-continuous for any bump function $\Lambda$ (see Remark~\ref{rem:wahl_topology}). 
For these range of fluxes, it can be easily checked that the function $\Lambda$
can be smooth and compacly supported and not only a bump function, and that the $\sT_{\Lambda}$-continuity extends to 
$\ua\in[0,1)^K$. 
Indeed, by \cite{dirac_s3_paper1}*{Theorem~22}, there does not exist any collapsing sequence of eigenfunctions in the sense of 
\eqref{eq:collapsing_sequence_quasimode} along a given sequence $(\uS^{(n)},\ua^{(n)})\to (\uS,\ua)$ 
when $\ua\in(0,1)^K$  or $\ua\in[0,1)^K$ and $\ua_k^{(n)}\to 0^+$ for the corresponding $k$. 
The cited theorem states that for such a converging sequence 
$(\uS^{(n)},\ua^{(n)})\to (\uS,\ua)$, any sequence $(\psi_n)$ of wave functions such that
\begin{enumerate}
 \item $\psi_n\in\dom(\cD_{\bA_n}^{(-)})$, $\bA_n=\sum_k \alpha_k^{(n)}[S_k^{(n)}]$,
 \item the sequence of graph norms $(\norm{\psi_n}_{\bA_n})_{n\ge 0}$ is bounded,
\end{enumerate}
converges -- up to the extraction of a subsequence -- strongly in $L^2(\S^3)^2$ 
to an element $\psi\in\dom(\cD_{\bA}^{(-)})$, $\bA=\sum_k \alpha_k [S_k]$ and $\cD_{\bA_n}\psi_n$
converges weakly to $\cD_{\bA}\psi$.
This shows that 
the path $t\in [0,1]\mapsto \cD_0(t)$ is norm resolvent continuous, the argument is as in the proof of Theorem~\ref{thm:conv_dirac}.

Now we consider the full homotopy. Let us first establish the strong resolvent continuity. 
We use Lemma~\ref{lem:cond_equiv_strong} (third characterization). 
The (strong-resolvent) continuity for the singular Dirac operators and 
Theorem~\ref{thm:conv_dirac} implies the continuity in the strong resolvent sense.

Now we establish the norm resolvent continuity of the homotopy.
Let us fix $\Lambda\in\sD(\R)$. Just as for Theorem~\ref{thm:conv_dirac},
we prove the continuity of $\Lambda(D(\delta,\ua))$ by contradiction.
The failure of continuity is equivalent to the existence of a sequence 
in the parameter space $(\delta_n,\ua_n)$ converging to
$(0,\ua)$ with $\ua\in [0,1)^K$ such that 
$$
	\limsup_{n\to\infty}\,\norm{\Lambda(D(\delta_n,\ua_n))-\Lambda(D(0,\ua))}_{\cB}>0.
$$
By strong resolvent convergence, this implies the existence of a sequence 
of normalized eigenfunctions $(\psi_n)$ associated with $\lambda_n$
such that $\lambda_n$ converges to $\lambda\in(\supp\,\Lambda)^{\circ}$ and $\psi_n\rightharpoonup 0$
(it is a collasping sequence in the sense of \eqref{eq:collapsing_sequence_quasimode}).
Then $(\psi_n)$ collapses onto the knots $\gamma_k$ and as in the proof of Theorem~\ref{thm:conv_dirac}, 
we obtain the necessary contradiction. We leave the details to the reader.

\appendix

\section{Proof of Proposition~\ref{prop:small_fluxes}}\label{sec:proof_square_integrability}
The proof uses the localization and phase jump functions introduced in 
Section~\ref{sec:techn_tools_loc_phase_fun}.
Those where developed for the problem on 
$\S^3$, but one can define them for the flat metric on $\R^3$ using the 
approximate Cartesian coordinates~\eqref{eq:cart_coord} associated with the flat metric. 

They are defined by
\begin{equation}\label{eq:cart_coord_r3}
	F_{\gamma}^{\R^3}(s',\bu)
	=\exp_{\gamma(s)}^{\Sigma}(u_1 \bs_{\Sigma}(\gamma(s')))
	+u_2\bn_{\Sigma}(\exp_{\gamma(s)}^{\Sigma}(u_1 \bs_{\Sigma}(\gamma(s')))),
\end{equation}
where $s'$ is the arclength parameter of $\gamma$ \emph{in the flat metric} $g_{\R^3}$, the triplet $(\mathbf{t}_{\gamma},\bs_{\Sig},\bn_{\Sigma})$ 
is the Darboux frame $g_{\R^3}$-normalized, and $\exp^{\Sig}$ denotes the exponential map in the metric induced by $g_{\R^3}$. We also denote by $\ell_k'$
the length of $\gamma_k$ in the flat metric.

We leave it to the reader to show that a formula similar to that of Proposition~\ref{prop:expr_free_dirac} 
holds for the free Dirac operator, written in the coordinates $(s_{k}', \bu_k)$ around the knot $\gamma_k$.
We still write $(\eta_+,\eta_-)$ for the sections aligned with or against the $\gamma_k$'s, 
defined in a $g_{\R^3}$-tubular neighborhood
of the link $\gamma$. Their relative phase is fixed with respect to the 
Seifert frame $(\mathbf{t}_{\gamma_k},\bs_{\Sigma_m},\bn_{\Sigma_m})$
for $\gamma_k\subset \partial S_m$.

For $\delta>0$ small enough and $\psi\in\dom(\cD_{\bA}^{(-)})$ (just as in the case of $\S^3$) 
we get the following:
\begin{enumerate}
	\item The phase jump functions $E_{\delta,\underline{j}}$ and $E_k$ are Borel and have values in $\S^1\subset\C$.
	
	\item For any $1\le m\le M$, and $\underline{j}\in \{2,3\}^M$, 
	the wave function $\overline{E}_{\delta,\underline{j}}\chi_{\delta,\underline{j}}\psi$ is in $H^1(\R^3)^2$.
	
	\item For any $1\le k\le K$, the wave function $\overline{E}_{\delta,\gamma_k}\chi_{\delta,\gamma_k}\psi$ satisfies
	$f=(f_+,f_-)^{\mathrm{T}}\in \dom\big(\cD_{\T_{\ell_k'},\alpha_k}^{(-)}\big)$ where $f$ is defined by:
	$$
		(\overline{E}_{\delta,\gamma_k}\chi_{\delta,\gamma_k}\psi)(\bx)
		=:f_+(s_k'(\bx),\bu_k(\bx))\eta_+(\bx)+f_-(s_k'(\bx),\bu_k(\bx))\eta_-(\bx),
	$$
	where $\dist_{\R^3}(\bx,\gamma_k)\le \delta$.
\end{enumerate}

We check the continuity on the dense set (in the graph norm):
\[
H^1(\T_{\ell_{k}'})\otimes\big(e^{i\alpha\theta}C^{\infty}_0(\R^2\setminus\{0\})^2)\oplus A_{\sing,\ell_k',\alpha_k},
\]
where $A_{\sing,\ell_k',\alpha_k}$ denotes the orthogonal complement of $\dom(\cD_{\T_{\ell_k'}}^{(\min)})$ with respect to the graph-norm.
We have \cite{dirac_s3_paper1}*{Lemma~10~and~below}:
\[
A_{\sing,\ell_k',\alpha_k}=\Big\{f_{\sing}(\underline{\lambda}):= \frac{1}{\sqrt{2\pi\ell}}\sum_{j\in\T_{\ell_k'}^*}\lambda_je^{ijs}
\begin{pmatrix} 0\\ e^{i\alpha_k\theta}K_{\alpha_k}(r\langle j\rangle)\end{pmatrix},
\ \underline{\lambda}\in\ell^2(\T_\ell^*)\Big\},
\]
where $K_{\alpha_k}$ denotes the modified Bessel function of the second kind, $\T_{\ell_k'}^*=\R/(\tfrac{2\pi}{\ell_k'}\Z)$ is the dual of $\T_{\ell_k'}$
and $\langle j\rangle:=\sqrt{1+j^2}$.
Let us show that the trace operator (with values in $L^2(\wt{\Sigma})^2$) is well-defined for elements of the form
\begin{equation}\label{eq:form_ans}
f:=\sum_{n=1}^N e^{i\alpha_k\theta} a_n(s)w_n(\bu)+f_{\sing}(\underline{\lambda}),
\end{equation}
where $N\in\mathbb{N}$, $w_n\in C^{\infty}_0(\R^2\setminus\{0\})^2$, $a_n\in H^1(\T_{\ell_{k}'})$ and $f_{\sing}(\underline{\lambda})\in A_{\sing,\ell_k',\alpha_k}$.
A computation gives the following identity:
\[
\norm{f_{\sing}(\underline{\lambda})}_{\T_{\ell}}^2=\bigg(\int_{0}^{+\infty}(K_{\alpha_k}(r)^2+K_{1-\alpha_k}(r)^2)r\d r \bigg)\sum_{j\in\T_{\ell_k'}^*}|\lambda_j|^2.
\]

For short we write $S(a,w)=\sum_{n=1}^Na_n(s)w_n(\bu)$ and $f_=:=e^{i\alpha_k\theta}S(a,w)\in \dom(\cD_{\T_{\ell_k'}}^{(\min)})$.
We also drop the dependence in $k$ and write $\ell$ instead of $\ell_k$.

If $0<\alpha_k<\tfrac{1}{2}$, then the trace of $f$ 
on $\T_{\ell}\times (\R\times\{0\})$ from above (or below) is square integrable 
in 
$$
	L^2(\T_{\ell}\times (\R\times\{0\});\d s'\d u_1)^2.
$$
Indeed the trace of $f_0$ is in $H^{1/2}(\T_{\ell}\times (\R\times\{0\}))^2$ by $H^1$-regularity on $\T_{\ell}\times(\R\times \R_+)$. 
Then the trace of $f_{\sing}(\underline{\lambda})$ is also $L^2$ with:
\begin{align*}
 \underset{{(s,\bu):\,u_2=0}}{\int}|f_{\sing}(\underline{\lambda})|^2&\le \frac{1}{2\pi}\sum_{j\in \T_{\ell}^*}|\lambda_j|^2\int_{\R}|K_{\alpha}(u_1\langle j\rangle)|^2\d u_1,\\
								    &\le \frac{1}{2\pi}\sum_{j\in \T_{\ell}^*}\frac{|\lambda_j|^2}{\langle j\rangle}\int_{\R}|K_{\alpha} (u_1)|^2\d u_1,\\
								    &\le C(\alpha)\norm{f_{\sing}(\underline{\lambda})}_{\T_{\ell}}^2,
\end{align*}
where $C(\alpha)$ is of order $\mathcal{O}_{\alpha\to \tfrac{1}{2}}(\tfrac{1}{1-2\alpha})$. By $\norm{\cdot}_{\T_{\ell}}$-orthogonality of $f_0$ and $f_{\sing}(\underline{\lambda})$,
we get that the trace of $f$ is in well-defined as a continuous linear map from $\dom(\cD_{\T_{\ell}}^{(-)})$ to $L^2(\T_{\ell}\times \R\times\{0\})^2$.

Going back to $\R^3$ (we were localizing around a knot $\gamma_k$), this shows that the trace of 
$\overline{E}_{\delta,\gamma_k}\chi_{\delta,\gamma_k}\psi$ 
on $\Sigma_m$ (with $\gamma_k\subset\partial S_m$) from above (or below) is also 
square integrable, and the same applies for $\chi_{\delta,\gamma_k}\psi$.

Away from $\partial S_m$ let us show that the trace of 
$\psi$ on $\Sigma_m$ is square integrable. 

Away from the other knots
this is due to the fact that $\overline{E}_{\delta,\underline{j}}\psi$ is 
$H^1$ on $B_{\delta \Sigma_m}\setminus B_{\delta/2}[\gamma]$ (away from the knots but close to $\Sigma_m$). 

Close to a knot $\gamma_k$ (transverse to $\Sigma_m$)
this is due to the fact that on $L^2(\T_{\ell_k'}\times\R^2)^2$ the trace on an oriented surface $\wt{\Sigma}$
transverse to the null section $\sigma_0:=\T_{\ell_k'}\times\{0\}$, with $|\wt{\Sigma}\cap\sigma_0|=1$
is continuous in the graph norm of $\cD_{\T_{\ell_k'},\alpha_k}^{(-)}$, and it takes values in $L^2(\wt{\Sigma})^2$.
Indeed: we need to compare square integrability close to the support of the magnetic knots.
By Proposition~\ref{prop_link_with_model_case}, we can work with the model case $\cD_{\T}^{(-)}$
for functions which are localized around a knot $\gamma_{k'}$. We use Proposition~\ref{prop:push_forward}, adapted to $\R^3$
to write down the Dirac operator in Cartesian coordinates~\eqref{eq:cart_coord}. 

We recall that Lichnerowicz' formula applies for functions in $\dom(\cD_{\T_{\ell_{k'}}}^{(\min)})$:
for $f\in e^{i\alpha\theta}C^{\infty}_0(\T_{\ell_{k'}}\times(\R^2\setminus\{0\}))^2$, Stokes formula gives:
$
 \int |\cD_{\T_{\ell_{k'}}}^{(\min)} f|^2=\int |(\nabla f)_{|_{\theta\neq 0}}|^2.
$
The extension to $\dom(\cD_{\T_{\ell_{k}'}}^{(\min)})$ follows by density. 

Using \eqref{eq:egalite}, 
it suffices to study the trace operator for functions localized around the intersection, or equivalently 
we can assume that $\wt{\Sig}$ is localized around the null-section (like a small disk) and check the continuity
on the dense set \eqref{eq:form_ans}.

So let $f=f_0+f_{\sing}(\underline{\lambda})$ with $f_0=e^{i\alpha_{k}\theta}S(a,w)$. 
Again, we drop the dependence in $k$ and write $\ell$ instead of $\ell_{k}'$.

Since $S(a,w):=\sum_{n=1}^N a_n(s)w_n(\bu)\in H^1$, its $L^2(\wt{\Sig})^2$-trace is well-defined with norm smaller than
$C(\wt{\Sig})\norm{S(a,w)}_{H^1}$. The trace of $e^{i\alpha\theta}S(a,w)$ is well-defined as an element of
$L^2(\Sig)^2$ with norm smaller than
\[
C(\wt{\Sig})\norm{S(a,w)}_{H^1(\T_{\ell}\times\R^2 )^2}=C(\wt{\Sig})\norm{e^{i\alpha\theta}S(a,w)}_{\T_{\ell}}.
\]
We have used Stokes formula on $\{(s,\bu)\in \T_{\ell}\times \R^2,\ |\bu|\ge r\}$ with $r$ small enough.
Furthermore, by the assumption of transversality $\wt{\Sig}\pitchfork \T_{\ell}\times\{0\}$ 
the trace of $f_{\sing}(\underline{\lambda})$ is well-defined in $L^2(\wt{\Sig})^2$ with 
squared norm smaller than:
\begin{align*}
	\int_{\wt{\Sig}}|f_{\sing}(\underline{\lambda})|^2&\le \frac{C(\wt{\Sig})}{2\pi\ell}\int_{\R^2}\Big|\sum_{j\in\T_{\ell}^*}|\lambda_j| K_{\alpha}(r\langle j\rangle)\Big|^2,\\
			&\le \frac{C(\wt{\Sig})}{\ell}\sum_{j_1,j_2\in \T_{\ell}^*}|\lambda_{j_1}||\lambda_{j_2}|\int_{0}^{+\infty}
				K_{\alpha}(r\langle j_1\rangle)K_{\alpha}(r\langle j_2\rangle)r\d r.
\end{align*} 
By using Cauchy-Schwarz inequality, then changing variables $r'_k=r\langle j_k\rangle$, $k\in\{1,2\}$ and re-using Cauchy-Schwarz inequality, we obtain
\begin{align*}
	\int_{\wt{\Sig}}|f_{\sing}(\underline{\lambda})|^2&\le \frac{C(\wt{\Sig})}{\ell}\sum_{j_1,j_2\in \T_{\ell}^*}\frac{|\lambda_{j_1}||\lambda_{j_2}|}{\langle j_1\rangle\langle j_2\rangle}
				\int_{0}^{+\infty}K_{\alpha}(r)^2r\d r,\\
			&\le \frac{C(\wt{\Sig})}{\ell}\int_{0}^{+\infty}K_{\alpha}(r)^2r\d r\sum_{j\in\T_{\ell}^*}|\lambda_j|^2\sum_{j\in\T_\ell^*}\langle j\rangle^{-2},\\
			&\le C(\wt{\Sig},\ell)\norm{f_{\sing}}_{\T_{\ell}}^2.
\end{align*}

All in all the $L^2(\wt{\Sig})^2$-trace 
of the sum is well defined with squared norm smaller than
\begin{equation*}
C(\wt{\Sig},\ell)[\norm{e^{i\alpha\theta}S(a,w)}_{\T_{\ell}}^2+\norm{f_{\sing}(\underline{\lambda})}_{\T_{\ell}}^2]
	\le C(\wt{\Sig},\ell)\norm{e^{i\alpha\theta} S(a,w)+f_{\sing}(\underline{\lambda})}_{\T_{\ell}}^2,
\end{equation*}
where we have used the $\norm{\cdot}_{\T_{\ell}}$-orthogonality of $e^{i\alpha\theta} S(a,w)$ and $f_{\sing}(\underline{\lambda})$.
By $\norm{\cdot}_{\T_{\ell}}$-density, the $L^2(\wt{\Sig})^2$-trace is well-defined as a continuous linear map on 
$(\dom(\cD_{\T_{\ell}}), \norm{\cdot}_{\T_{\ell}})$.

\section{Proof of Proposition~\ref{prop:infinite}}\label{sec:proof_infinite}
	We prove that $\ker\,(\cD_{\bA}^{(\min)})^*$ is infinite dimensional,
	or equivalently, that $\big(\ran\,\cD_{\bA}^{(\min)}\big)^{\perp}$ is infinite-dimensional.
	In \cite{dirac_s3_paper1}, we showed that the minimal domain 
	$\dom(\cD_{\bA}^{(\min)})$ corresponds to the elements in 
	$H^1(\Omega_{\uS})^2$ with the correct phase jumps across the surfaces.
	Furthermore, for $\psi\in \dom(\cD_{\bA}^{(\min)})$ we have by Lichnerowicz' formula
	\begin{equation}\label{eq:Dirac_and_Dirichlet_energy}
		\int |\cD_{\bA}^{(\min)}\psi|^2 \geq \int|(\nabla\psi)_{|_{\Omega_{\uS}}}|^2.
	\end{equation}
	Let us show that the boundary values of $\psi$ on each $\widetilde{S}_k$ (on both sides) 
	are continuous in $H^{1/2}_{\loc}(\wt{S}_k)$. (We recall that $\wt{S}_k=S_k\setminus \gamma$ is an open surface).
	Let $\chi\in \sD(\S^3,[0,1])$ with support in $\S^3\setminus \gamma$
	and which equals one in $\{\bp\in\cup_k S_k,\ \dist(\bp,\partial S_k)>\eps\}$ 
	for a given small enough $\eps>0$.
	Then $\chi\psi$ is in $\dom(\cD_{\bA}^{(\min)})$ and its traces on both sides 
	of $S_k$ are controlled by its $H^1(\Omega_{\uS})^2$-norm.
	More generally, we denote by $D$ the operator with domain $H^1(\Omega_{\uS})^2$ that acts as follows:
	\[
		D\psi:=\big(-i\bsigma(\nabla)\psi\big)\big|_{\Omega_{\uS}}
		\in L^2(\Omega_{\uS})^2\hookrightarrow L^2(\S^3)^2.
	\]
	Using \eqref{eq:Dirac_and_Dirichlet_energy} and the Sobolev inequality (in three dimensions)
	we see that for any smooth function $\chi\in \sD(\wt{S}_k)$, there exists $C=C(\uS,\chi)>0$ 
	such that for any $\psi\in\dom(D)$ the following estimate holds:
	\begin{equation}\label{eq:trace_cont}
		\norm{\chi \psi_{|_{(\wt{S}_k)_{\pm}}}}_{H^{1/2}(\wt{S}_k)^2}
		\leq C(\uS)\norm{\chi\psi}_{H^1(\Omega_{\uS})^2}\le C(\uS,\chi)\norm{D\psi}_{L^2}.
	\end{equation}
	So the trace functionals $\psi\mapsto \psi_{|_{(S_k)_{\pm}}}$ are continuous from
	$\big(\dom(\overline{D}),\norm{D(\cdot)}_{L^2}\big)$ to $H^{1/2}_{\mathrm{loc}}(\wt{S}_k)^2$.
	The elements of $\dom\big(\cD_{\bA}^{(\max)}\big)$ satisfy the jump condition across $S_k$,
	which is a $H^{1/2}_{\mathrm{loc}}(\wt{S}_k)^2$-continuous condition.
	And from the definition of $\cD_{\bA}^{(\min)}$, the traces of elements in $\dom(\cD_{\bA}^{(\min)})$
	are in $H^{1/2}(S_k)$.
	
	The estimate \eqref{eq:trace_cont} implies that for any $\psi\in H^1(\S^3)^2$ the distance 
	\[
	\mathrm{dist}_{L^2}\Big(D\psi, \ran\big(\cD_{\bA}^{(\min)}\big)\Big)
	\] is positive whenever $\psi_{|_{S_k}}\neq 0$, since the traces
	of $\psi$ on $(S_{k})_{\pm}$ coincide. 
	Thus the $L^2$-orthogonal complement
	of $\ran\big(\cD_{\bA}^{(\min)}\big)$ in $\clos_{L^2}(\ran D)\subset L^2(\S^3)^2$ is non-trivial.
	Let us show that it is infinite dimensional to end the proof.
	We can pick a sequence $(u_n)_{n\ge 0}$ of non-zero elements in $\sD(\wt{S}_k)^2$
	with disjoint support. We can extend them on $\S^3$ to form a sequence $(\psi_n)_{n\ge 0}$
	of elements in $H^1(\S^3)^2$ and we consider $E=\mathrm{span}_{n\ge 0}(D\psi_n)$.
	By \eqref{eq:trace_cont}, this vector subspace is infinite dimensional and we have
	\[
	 E\oplus \clos_{L^2}\ran\big(\cD_{\bA}^{(\min)}\big)\subset \clos_{L^2}(\ran D)\subset L^2(\S^3)^2.
	\]
	To establish the supplementarity of the two vector subspaces, 
	we can argue by contradiction and use the estimate \eqref{eq:trace_cont}.



\begin{bibdiv}[Bibliography]{}
\begin{biblist}

	\bib{Adametal99}{article}{
	author={Adam, C.},
	author={Muratori, B.},
	author={Nash, C.},
	title={Zero modes of the Dirac operator in three dimensions},
	journal={Phys. Rev. D (3)},
	volume={60},
	date={1999},
	number={12},
	pages={125001, 8},
	}
	
	\bib{Adametal00_1}{article}{
	author={Adam, C.},
	author={Muratori, B.},
	author={Nash, C.},
	title={Degeneracy of zero modes of the Dirac operator in three dimensions},
	journal={Phys. Lett. B},
	volume={485},
	date={2000},
	number={1-3},
	pages={314--318},
	}

	\bib{Adametal00_2}{article}{
	author={Adam, C.},
	author={Muratori, B.},
	author={Nash, C.},
	title={Zero modes in finite range magnetic fields},
	journal={Modern Phys. Lett. A},
	volume={15},
	date={2000},
	number={25},
	pages={1577--1581},
	}
	
	\bib{Adametal00_3}{article}{
	author={Adam, C.},
	author={Muratori, B.},
	author={Nash, C.},
	title={Multiple zero modes of the Dirac operator in three dimensions},
	journal={Phys. Rev. D (3)},
	volume={62},
	date={2000},
	number={8},
	pages={085026, 9},
	}
	
	\bib{Adametal05}{article}{
	author={Adam, C.},
	author={S\'{a}nchez-Guill\'{e}n, J.},
	title={The symmetries of the Dirac-Pauli equation in two and three dimensions},
	journal={J. Math. Phys.},
	volume={46},
	date={2005},
	pages={052304},
	}
	
	
	\bib{AhaBohm59}{article}{
	author={Aharonov, Y.},
	author={Bohm, D.},
	title={Significance of electromagnetic potentials in the quantum theory},
	journal={Phys. Rev. (2)},
	volume={115},
	date={1959},
	}
	
	\bib{AhaCash79}{article}{
	author={Aharonov, Y.},
	author={Casher, A.},
	title={Ground state of a spin-${\frac{1}{2}}$ charged particle in a two-dimensional magnetic field},
	journal={Phys. Rev. A (3)},
	volume={19},
	date={1979},
	number={6},
	pages={2461--2462},
	}
	
	
	
	\bib{AleyTolka}{article}{
   	author={Aleynikov, D.},
  	author={Tolkachev, E.},
   	title={Kustaanheimo-Stiefel transformation and static zero modes of Dirac operator},
   	date={2002},
	url={https://arxiv.org/abs/hep-th/0206211}
	}
	
	
	
	\bib{Arrizetal14}{article}{
   	author={Arrizabalaga, Naiara},
  	author={Mas, Albert},
   	author={Vega, Luis},
   	title={Shell interactions for Dirac operators},
   	journal={J. Math. Pures Appl. (9)},
   	volume={102},
   	date={2014},
   	number={4},
   	pages={617--639},
	}

	\bib{HarmAnDir}{article}{
   	author={Axelson, Andreas},
	author={Grognard, Ren\'{e}},
	author={Hogan, Jeff},
   	title={Harmonic Analysis of Dirac Operators on Lipschitz Domains},
   	journal={Clifford Analysis and its Applications},
   	volume={25},
   	date={2001},
   	pages={231--246},
	}

	\bib{MR1807254}{article}{
	author={Balinsky, A. A.},
	author={Evans, W. D.},
	title={On the zero modes of Pauli operators},
	journal={J. Funct. Anal.},
	volume={179},
	date={2001},
	number={1},
	pages={120--135},
	}

	\bib{MR1874252}{article}{
	author={Balinsky, A. A.},
	author={Evans, W. D.},
	title={On the zero modes of Weyl-Dirac operators and their multiplicity},
	journal={Bull. London Math. Soc.},
	volume={34},
	date={2002},
	number={2},
	pages={236--242},
	}
	
	\bib{BengVdB15}{article}{
	   author={Benguria, R. D.},
	   author={Van Den Bosch, H.},
	   title={A criterion for the existence of zero modes for the Pauli operator
	   with fastly decaying fields},
	   journal={J. Math. Phys.},
	   volume={56},
	   date={2015},
	   number={5},
	   pages={052104, 7},
	}
	
	
	\bib{BorgPule03}{article}{
	author={Borg, J. L.},
	author={Pul{\'e}, J. V.},
	title={Pauli approximations to the self-adjoint extensions of the Aharonov-Bohm Hamiltonian},
	journal={J. Math. Phys.},
	volume={44},
	date={2003},
	number={10},
	pages={4385--4410},
	}
	
	\bib{Schroe_op_Cycon_et_al}{book}{
	author={Cycon, H. L.},
	author={Froese, R. G.},
	author={Kirsch, W.},
	author={Simon, B.},
	title={Schr\"odinger operators with application to quantum mechanics and global geometry},
	series={Texts and Monographs in Physics},
	edition={Springer Study Edition},
	publisher={Springer-Verlag, Berlin},
	date={1987},
	}
	

	\bib{DeTurGlu08}{article}{
	author={DeTurck, D.},
	author={Gluck, H.},
	title={Electrodynamics and the Gauss linking integral on the 3-sphere and in hyperbolic 3-space},
	journal={J. Math. Phys.},
	volume={49},
	date={2008},
	number={2},
	misc={023504},
	}
	
	
	\bib{DunneMin}{article}{
 	 title = {Abelian zero modes in odd dimensions},
	  author = {Dunne, G.},
	  author={Min, H.}
 	 journal = {Phys. Rev. D},
 	 volume = {78},
 	 issue = {6},
 	 pages = {067701},
	  year = {2008},
	}
	
	
	
	\bib{Elton00}{article}{
	author={Elton, Daniel M.},
	title={New examples of zero modes},
	journal={J. Phys. A},
	volume={33},
	date={2000},
	number={41},
	pages={7297--7303},
	}
	
	\bib{Elton02}{article}{
	author={Elton, Daniel M.},
	title={The local structure of zero mode producing magnetic potentials},
	journal={Comm. Math. Phys.},
	volume={229},
	date={2002},
	number={1},
	pages={121--139},
	}

	\bib{ErdSol01}{article}{
   	author={Erd{\H{o}}s, L{\'a}szl{\'o}},
   	author={Solovej, Jan Philip},
   	title={The kernel of Dirac operators on $\mathbb{S}^3$ and $\mathbb{R}^3$},
   	journal={Rev. Math. Phys.},
   	volume={13},
   	date={2001},
   	number={10},
   	pages={1247--1280},
	}
	
	\bib{ErdVug02}{article}{
	author={Erd{\H{o}}s, L{\'a}szl{\'o}},
	author={Vougalter, Vitali},
	title={Pauli operator and Aharonov-Casher theorem for measure valued magnetic fields},
	journal={Comm. Math. Phys.},
	volume={225},
	date={2002},
	number={2},
	pages={399--421},
	}
	
	\bib{Fefferman95}{article}{
	author={Fefferman, Charles},
	title={Stability of Coulomb systems in a magnetic field},
	journal={Proc. Nat. Acad. Sci. U.S.A.},
	volume={92},
	date={1995},
	number={11},
	pages={5006--5007},
	}
	
	\bib{Fefferman96}{article}{
	author={Fefferman, Charles},
	title={On electrons and nuclei in a magnetic field},
	journal={Adv. Math.},
	volume={124},
	date={1996},
	number={1},
	pages={100--153},
	}
	
	\bib{FP30}{article}{
   	author={Frankl, F.},
   	author={Pontrjagin, L.},
   	title={Ein Knotensatz mit Anwendung auf die Dimensionstheorie},
   	language={German},
   	journal={Math. Ann.},
   	volume={102},
   	date={1930},
   	number={1},
   	pages={785--789},
	}
	
	\bib{FroLiebLoss86}{article}{
	author={Fr{\"o}hlich, J{\"u}rg},
	author={Lieb, Elliott H.},
	author={Loss, Michael},
	title={Stability of Coulomb systems with magnetic fields. I. The one-electron atom},
	journal={Comm. Math. Phys.},
	volume={104},
	date={1986},
	number={2},
	pages={251--270},
	}

	\bib{GeylGrish02}{article}{
	author={Geyler, V. A.},
	author={Grishanov, E. N.,},
	title={Zero Modes in a Periodic System of Aharonov-Bohm Solenoids},
	journal={JETP Letters},
	volume={75},
	number={7},
	date={2002},
	pages={354-356},
	}

	\bib{GeylStov04}{article}{
	author={Geyler, V. A.},
	author={{\v{S}}{\soft{t}}ov{\'{\i}}{\v{c}}ek, P.},
	title={On the Pauli operator for the Aharonov-Bohm effect with two solenoids},
	journal={J. Math. Phys.},
	volume={45},
	date={2004},
	number={1},
	pages={51--75},
	}

	

	\bib{HiroOgur01}{article}{
	author={Hirokawa, Masao},
	author={Ogurisu, Osamu},
	title={Ground state of a spin-$1/2$ charged particle in a two-dimensional magnetic field},
	journal={J. Math. Phys.},
	volume={42},
	date={2001},
	number={8},
	pages={3334--3343},
	}

	\bib{HardSpacSingInt}{article}{
   	author={Hofmann, Steve},
	author={Marmolejo-Olea, Emilio},
	author={Mitrea, Marius},
	author={P\'{e}rez-Esteva, Salvador},
	author={Taylor, Michael},
   	title={Hardy spaces, singular integrals and the geometry of euclidean domains of locally finite perimeter},
   	journal={Geom. Funct. Anal.},
   	volume={19},
   	date={2009},
	number={3},
   	pages={842--882},
	}
	
	
	\bib{JackPi00}{article}{
	  title={Creation and evolution of magnetic helicity},
	  author = {Jackiw, R.},
	  author={Pi, S.-Y.},
 	 journal = {Phys. Rev. D},
 	 volume = {61},
 	number = {10},
 	 pages = {105015},
 	 year = {2000},
	}
	
	
	
	\bib{Kato}{book}{
	author={Kato, Tosio},
	title={Pertubation Theory for Linear Operators},
	edition={2},
	publisher={Springer, Berlin, Heidelberg},
	date={1976},
	}
	
	\bib{LiebLossSol95}{article}{
	author={Lieb, Elliott H.},
	author={Loss, Michael},
	author={Solovej, Jan Philip},
	title={Stability of matter in magnetic fields},
	journal={Phys. Rev. Lett.},
	volume={75},
	date={1995},
	number={6},
	pages={985--989},
	}
	
	\bib{LiebLoss86}{article}{
	author={Lieb, Elliott H.},
	author={Loss, Michael},
	title={Stability of Coulomb systems with magnetic fields. II. The
	many-electron atom and the one-electron molecule},
	journal={Comm. Math. Phys.},
	volume={104},
	date={1986},
	number={2},
	pages={271--282},
	}
	
	\bib{LossYau86}{article}{
	author={Loss, Michael},
	author={Yau, Horng-Tzer},
	title={Stabilty of Coulomb systems with magnetic fields. III. Zero energy
	bound states of the Pauli operator},
	journal={Comm. Math. Phys.},
	volume={104},
	date={1986},
	number={2},
	pages={283--290},
	}
	
	\bib{MR2912458}{article}{
	author={Marmolejo-Olea, Emilio},
	author={Mitrea, Irina},
	author={Mitrea, Marius},
	author={Shi, Qiang},
	title={Transmission boundary problems for Dirac operators on Lipschitz
	domains and applications to Maxwell's and Helmholtz's equations},
	journal={Trans. Amer. Math. Soc.},
	volume={364},
	date={2012},
	number={8},
	pages={4369--4424},
	}

	\bib{Min}{article}{
	author = {Min, Hyunsoo},
	year = {2009},
	title = {Fermion Zero Modes in Odd Dimensions},
	journal={J. of Physics A},
	volume = {43},
	number={9},
	}


	\bib{Persson_on_Pauli}{article}{
	author={Persson, Mikael},
	title={On the Aharonov-Casher formula for different self-adjoint
	extensions of the Pauli operator with singular magnetic field},
	journal={Electron. J. Differential Equations},
	date={2005},
	number={55},
	pages={pp.~16},
	note={(electronic)},
	}
		
	\bib{Persson_dirac_2d}{article}{
  	author = {Persson, Mikael},
	title = {On the Dirac and Pauli Operators with Aharonov-Bohm Solenoids},
  	journal = {Lett. in Math. Phys.},
  	volume = {78},
	date={2006},
  	pages = {139--156},
	}
	
	\bib{Persson_even_dim}{article}{
	author={Persson, Mikael},
	title={Zero modes for the magnetic Pauli operator in even-dimensional Euclidean space},
	journal={Lett. Math. Phys.},
	volume={85},
	date={2008},
	number={2-3},
	pages={111--128},
	}
	
	\bib{Philips_spectral_flow}{article}{
  	author = {Phillips, John},
	title = {Self-adjoint Fredholm operators and spectral flow},
  	journal = {Canad. Math. Bull.},
  	volume = {39},
	date={1996},
	number={4},
  	pages = {460--467},
	}
	
	\bib{dirac_s3_paper1}{article}{
	author={Portmann, Fabian},
	author={Sok, J\'er\'emy},
	author={Solovej, Jan Philip},
	title={Self-adjointness \& spectral properties of Dirac operators with magnetic links},
	note={To appear in J. Math. Pures Appl., ArXiv:1701.04987},
	date={2017}
	}
	
	\bib{dirac_s3_paper2}{article}{
	author={Portmann, Fabian},
	author={Sok, J\'er\'emy},
	author={Solovej, Jan Philip},
	title={Spectral flow for Dirac operators with magnetic links},
	note={ArXiv:1701.05044},
	date={2017}
	}
	
	
	
	
	\bib{Rham55}{book}{
	author={de Rham, Georges},
	title={Vari{\'e}t{\'e}s diff{\'e}rentiables},
	language={French},
	year={1955},
	publisher={Hermann, Paris},
	}
	
	\bib{Rolfsen}{book}{
	author={Rolfsen, Dale},
	title={Knots and links},
	series={Mathematics Lecture Series},
	volume={7},
	note={Corrected reprint of the 1976 original},
	publisher={Publish or Perish, Inc., Houston, TX},
	date={1990},
	}
	
	\bib{RozenShiro06}{article}{
	author={Rozenblum, Grigori},
	author={Shirokov, Nikolai},
	title={Infiniteness of zero modes for the Pauli operator with singular magnetic field},
	journal={J. Funct. Anal.},
	volume={233},
	date={2006},
	number={1},
	pages={135--172},
	}
	
	\bib{MR916076}{article}{
	author={Scharlemann, Martin},
   	author={Thompson, Abigail},
   	title={Finding disjoint Seifert surfaces},
   	journal={Bull. London Math. Soc.},
   	volume={20},
   	date={1988},
   	number={1},
   	pages={61--64},
	}
	
	\bib{Seifert35}{article}{
   	author={Seifert, H.},
   	title={\"Uber das Geschlecht von Knoten},
   	language={German},
   	journal={Math. Ann.},
   	volume={110},
   	date={1935},
   	number={1},
   	pages={571--592},
	}
	
	\bib{shigekawa91}{article}{
	author={Shigekawa, Ichir{\=o}},
	title={Spectral properties of Schr\"odinger operators with magnetic
	fields for a spin $\frac12$ particle},
	journal={J. Funct. Anal.},
	volume={101},
	date={1991},
	number={2},
	pages={255--285},
	}
	
	\bib{Spivakvol2}{book}{
	author={Spivak, Michael},
	title={A comprehensive introduction to differential geometry. Vol. II},
	edition={3},
	publisher={Publish or Perish, Inc., Huston, Texas},
	date={1999},
	}
	
	\bib{Spivakvol4}{book}{
	author={Spivak, Michael},
	title={A comprehensive introduction to differential geometry. Vol. IV},
	edition={3},
	publisher={Publish or Perish, Inc., Huston, Texas},
	date={1999},
	}
	
	\bib{Tamura}{article}{
   	author={Tamura, Hideo},
	title={Resolvent convergence in norm for Dirac operator with Aharonov-Bohm field},
	journal={J. Math. Phys..},
   	volume={44},
   	date={2003},
   	number={7},
   	pages={2967--2993},
	}

	\bib{Wahl08}{incollection}{
   	author={Wahl, Charlotte},
	title={A New Topology on the Space of Unbounded Selfadjoint Operators, K-theory and Spectral Flow},
	booktitle={C\textsuperscript{*}-Algebras and Elliptic Theory II},
	editor={Burghelea, Melrose, Mishchenko, Troitsky},
   	year={2008},
   	pages={297--309},
	publisher={Birkh\"auser}
	}
	
	\bib{White69}{article}{
	   author={White, James H.},
	   title={Self-linking and the Gauss integral in higher dimensions},
	   journal={Amer. J. Math.},
	   volume={91},
	   date={1969},
	   pages={693--728},
	}
	
\end{biblist}
\end{bibdiv}
\end{document}